\documentclass[reqno,10pt,a4paper,dvipsnames]{amsart}
\numberwithin{equation}{section}
\allowdisplaybreaks 

\usepackage{amssymb,eucal,mathrsfs,setspace,cite}

\usepackage{array}
\newcolumntype{C}{>{$}c<{$}} 

\usepackage[largesc,theoremfont]{newtxtext}      
\usepackage[libertine,cmbraces,varbb]{newtxmath} 

\begingroup
  \makeatletter
  \@for\theoremstyle:=definition,remark,plain\do{%
    \expandafter\g@addto@macro\csname th@\theoremstyle\endcsname{%
      \addtolength\thm@preskip{.5\baselineskip plus .2\baselineskip minus .2\baselineskip}
      \addtolength\thm@postskip{.5\baselineskip plus .2\baselineskip minus .2\baselineskip}
    }%
  }
\endgroup

\usepackage[inner=24mm, outer=24mm, top=24mm, bottom=24mm, head=10mm, foot=10mm]{geometry}

\usepackage{enumitem}
\setitemize{leftmargin=*}   
\setenumerate{leftmargin=*, 
  label=\textup{(\roman*)}} 

\usepackage{mathtools} 

\usepackage{graphicx}

\usepackage{tikz}
\usetikzlibrary{decorations.markings,positioning,arrows}
\usepackage[section]{placeins}
\usepackage{float}
\tikzset{-<-/.style={thick, decoration={markings, mark=at position .3 with {\arrow{<}}},postaction={decorate}}}

\usepackage[colorlinks=true,citecolor=red,linkcolor=blue]{hyperref} 

\usepackage[capitalise,noabbrev]{cleveref} 

\newcommand{\eps}{\varepsilon}
\newcommand{\pd}{\partial}     

\renewcommand{\ge}{\geq}
\renewcommand{\le}{\leq}

\newcommand{\relmiddle}[1]{\mathrel{}\middle#1\mathrel{}}

\DeclareMathOperator{\Sid}{\overline{\wun}}
\DeclareMathOperator{\vspn}{span}
\DeclareMathOperator{\End}{End}

\newcommand{\dd}{\mathrm{d}}   
\newcommand{\ii}{\mathfrak{i}} 
\newcommand{\ee}{\mathsf{e}}   

\newcommand{\wun}{\vvmathbb{1}}  

\DeclarePairedDelimiter{\brac}{\lparen}{\rparen} 
\DeclarePairedDelimiter{\sqbrac}{\lbrack}{\rbrack} 
\DeclarePairedDelimiter{\set}{\lbrace}{\rbrace}
\newcommand{\st}{\mspace{5mu} : \mspace{5mu}} 
\DeclarePairedDelimiter{\abs}{\lvert}{\rvert}

\DeclarePairedDelimiter{\ang}{\langle}{\rangle}
\DeclarePairedDelimiter{\normord}{{} :}{: {}} 

\DeclarePairedDelimiterX{\comm}[2]{\lbrack}{\rbrack}{#1 , #2}  
\DeclarePairedDelimiterX{\acomm}[2]{\lbrace}{\rbrace}{#1 , #2} 
\DeclarePairedDelimiterX{\inner}[2]{\langle}{\rangle}{#1 , #2} 
\DeclarePairedDelimiterX{\super}[2]{\lparen}{\rparen}{#1 \delimsize\vert \mathopen{} #2} 

\newcommand{\ra}{\rightarrow}
\newcommand{\lra}{\longrightarrow}

\newcommand{\dses}[5]{0 \lra #1 \overset{#2}{\lra} #3 \overset{#4}{\lra} #5 \lra 0} 

\newcommand{\Res}[1]{#1 \mbox{$\downarrow$} {}}           
\newcommand{\Ind}[1]{#1 \raisebox{0.18em}{$\uparrow$} {}} 

\newcommand{\fld}[1]{\mathbb{#1}}    
\newcommand{\alg}[1]{\mathfrak{#1}}  
\newcommand{\grp}[1]{\mathsf{#1}}    
\newcommand{\Mod}[1]{\mathcal{#1}}   
\newcommand{\VOA}[1]{\mathsf{#1}}    
\newcommand{\categ}[1]{\textbf{#1}}  
\newcommand{\sector}[1]{\mathscr{#1}} 

\newcommand{\ZZ}{\fld{Z}}
\newcommand{\NN}{\ZZ_{\ge 0}} 

\newcommand{\RR}{\fld{R}}
\newcommand{\CC}{\fld{C}}

\newcommand{\SLG}[2]{\grp{#1} \brac{#2}}       

\newcommand{\affine}[1]{\widehat{#1}}
\newcommand{\SLA}[2]{\alg{#1} \brac{#2}}                 
\newcommand{\SLSA}[3]{\alg{#1} \super{#2}{#3}}           
\newcommand{\AKMA}[2]{\affine{\alg{#1}} \brac{#2}}       
\newcommand{\AKMSA}[3]{\affine{\alg{#1}} \super{#2}{#3}} 
\newcommand{\EnvAlg}[1]{\mathsf{U}\brac{#1}}             
\newcommand{\EnvAlgk}[1]{\mathsf{U}_k\brac{#1}}
\newcommand{\ideal}[1]{\ang[\big]{#1}}

\newcommand{\cox}{\mathsf{h}}   
\newcommand{\dcox}{\cox^{\vee}} 

\newcommand{\killing}[2]{\kappa \brac{#1 , #2}} 

\newcommand{\finite}[1]{\overline{#1}}
\newcommand{\func}[2]{#1 #2}                         
\newcommand{\parr}{\Pi}                              
\newcommand{\conjaut}{\mathsf{w}}                    
\newcommand{\finconjaut}{\finite{\mathsf{w}}}        
\newcommand{\sfaut}{\sigma}                          
\newcommand{\conjmod}[1]{\func{\conjaut}{#1}}        
\newcommand{\finconjmod}[1]{\func{\finconjaut}{#1}}  
\newcommand{\sfmod}[2]{\func{\sfaut^{#1}}{#2}}       
\newcommand{\parrmod}[1]{\func{\parr}{#1}}           

\newcommand{\algmod}[1]{\categ{Mod}_{#1}}
\newcommand{\wtmod}[1]{\categ{WMod}_{#1}}

\newcommand{\SLMinMod}[2]{\VOA{A}_1 \brac{#1 , #2}}             
\newcommand{\OSPMinMod}[2]{\VOA{B}_{0\vert 1} \brac{#1 , #2}}   
\newcommand{\OSPUniv}[1]{\AKMSA{osp}{1}{2}_{#1}}                

\newcommand{\vac}[1]{\Omega_{#1}}                               

\newcommand{\gvoa}{\VOA{V}}                                     
\newcommand{\orb}{\VOA{O}}                                      

\newcommand{\gtzhu}[2]{\mathsf{Zhu}^{#2}\sqbrac*{#1}}                       
\newcommand{\zhu}[1]{\gtzhu{#1}{}}                                          
\newcommand{\parity}{\tau}                                                  
\newcommand{\tzhu}[1]{\gtzhu{#1}{\parity}}                                  

\newcommand{\parityof}[1]{\abs{#1}}                                         

\newcommand{\zstar}{\ast}                                                   

\newcommand{\simp}{\Mod{J}}                
\newcommand{\nsr}[2]{\sector{S}_{#1}^{#2}} 
\newcommand{\gnsr}[1]{\nsr{#1}{}}
\newcommand{\bnsr}[1]{\nsr{#1}{\orb}}

\newcommand{\even}{(0)} 
\newcommand{\odd}{(1)}

\newcommand{\NS}{\mathrm{NS}}
\newcommand{\R}{\mathrm{R}}
\newcommand{\NSVer}[1]{\Mod{V}^{\NS}_{#1}}   
\newcommand{\RVer}[1]{\Mod{V}^{\R}_{#1}}     

\newcommand{\NSFin}[1]{\Mod{A}_{#1}} 
\newcommand{\NSInf}[1]{\Mod{B}_{#1}} 
\newcommand{\NSRel}[1]{\Mod{C}_{#1}} 
\newcommand{\RInf}[1]{\Mod{D}_{#1}}  
\newcommand{\RRel}[1]{\Mod{E}_{#1}}  
\newcommand{\RFin}[1]{\Mod{F}_{#1}}  

\newcommand{\Fin}[1]{\finite{\Mod{A}}_{#1}} 
\newcommand{\Inf}[1]{\finite{\Mod{B}}_{#1}} 
\newcommand{\Rel}[1]{\finite{\Mod{C}}_{#1}} 
\newcommand{\slFin}[1]{\finite{\Mod{F}}_{#1}}  
\newcommand{\slInf}[1]{\finite{\Mod{D}}_{#1}}  
\newcommand{\slRel}[1]{\finite{\Mod{E}}_{#1}}  

\DeclarePairedDelimiterX{\brakets}[2]{\langle}{\rangle}{#1 \delimsize\vert #2}                     
\DeclarePairedDelimiterX{\brackets}[3]{\langle}{\rangle}{#1 \delimsize\vert #2 \delimsize\vert #3} 

\DeclareMathOperator{\tr}{tr}
\DeclareMathOperator{\str}{str}

\newcommand{\Gr}[1]{\sqbrac[\big]{#1}}                               
\newcommand{\tGr}[1]{\sqbrac{#1}}                                    

\newcommand{\traceover}[1]{\tr_{\raisebox{-2pt}{$\scriptstyle #1$}}} 
\newcommand{\straceover}[1]{\str_{\raisebox{-2pt}{$\scriptstyle #1$}}} 

\newcommand{\nchmap}[1]{\mathrm{ch}_{#1}}                            
\newcommand{\chmap}{\nchmap{0}}
\newcommand{\schmap}{\nchmap{1/2}}

\newcommand{\nch}[2]{\nchmap{#2} \Gr{#1}}                            
\newcommand{\ch}[1]{\nch{#1}{0}}                                     
\newcommand{\fch}[2]{\ch{#1} \brac[\big]{#2}}
\newcommand{\sch}[1]{\nch{#1}{1/2}}                                  
\newcommand{\fsch}[2]{\sch{#1} \brac[\big]{#2}}
\newcommand{\bch}[1]{\nch{#1}{}}                                     

\newcommand{\tnch}[2]{\nchmap{#2} \tGr{#1}}                           
\newcommand{\tch}[1]{\tnch{#1}{0}}                                    

\newcommand{\tsch}[1]{\tnch{#1}{1/2}}                                 

\newcommand{\jth}[1]{\vartheta_{#1}}                                 
\newcommand{\fjth}[2]{\jth{#1} \brac{#2}}                            
\newcommand{\deta}[1]{\eta \brac{#1}}                                
\newcommand{\ddelta}[1]{\delta \brac{#1}}                            
\newcommand{\taufn}[2]{A_{#1} \brac{#2}}                             

\newcommand{\modfont}[1]{\mathsf{#1}}
\newcommand{\modS}{\modfont{S}}                        
\newcommand{\slmodS}{\modfont{s}}                      
\newcommand{\modT}{\modfont{T}}                        
\newcommand{\modch}[2]{\modfont{#1} \set[\big]{#2}}

\newcommand{\Smat}[4]{\modS_{#1, #3} \sqbrac[\big]{#2 \ra #4}} 
\newcommand{\Selt}[2]{\Sigma \sqbrac[\big]{#1 \ra #2}}         
\newcommand{\slSmat}[2]{\slmodS \sqbrac[\big]{#1 \ra #2}}      

\newcommand{\tSmat}[4]{\modS_{#1, #3} \sqbrac{#2 \ra #4}} 

\newcommand{\sker}[2]{K_{#1}^{\modS} \brac{#2}}                      
\newcommand{\tker}[2]{K_{#1}^{\modT} \brac{#2}}

\newcommand{\fuse}[1]{\mathbin{\times_{#1}}}                              
\newcommand{\gfuse}{\fuse{}}                                              
\newcommand{\bfuse}{\fuse{\orb}}                                          

\newcommand{\Grfuse}[1]{\mathbin{\boxtimes_{#1}}}                         
\newcommand{\gGrfuse}{\Grfuse{}}                                          
\newcommand{\bGrfuse}{\Grfuse{\orb}}                                      
\newcommand{\fuscoeff}[4]{\genfrac{[}{]}{0pt}{0}{#3}{#1 \ \ #2}_{#4}}     
\newcommand{\bfuscoeff}[3]{\fuscoeff{#1}{#2}{#3}{\orb}}                   
\newcommand{\tfuscoeff}[4]{\left[ \begin{smallmatrix} #3 \\ #1 \ #2 \end{smallmatrix} \right]_{#4}}

\newcommand{\cft}{conformal field theory}
\newcommand{\cfts}{conformal field theories}
\newcommand{\uea}{universal enveloping algebra}
\newcommand{\uesa}{universal enveloping superalgebra}
\newcommand{\lcft}{logarithmic conformal field theory}
\newcommand{\lcfts}{logarithmic conformal field theories}

\newcommand{\lw}{lowest-weight}
\newcommand{\lwv}{\lw{} vector}
\newcommand{\lwvs}{\lw{} vectors}
\newcommand{\lwm}{\lw{} module}
\newcommand{\lwms}{\lw{} modules}
\newcommand{\hw}{highest-weight}
\newcommand{\hwv}{\hw{} vector}
\newcommand{\hwvs}{\hw{} vectors}
\newcommand{\hwm}{\hw{} module}
\newcommand{\hwms}{\hw{} modules}
\newcommand{\sv}{singular vector}
\newcommand{\svs}{singular vectors}

\newcommand{\voa}{vertex operator algebra}
\newcommand{\voas}{vertex operator algebras}
\newcommand{\vosa}{vertex operator superalgebra}
\newcommand{\vosas}{vertex operator superalgebras}

\newcommand{\opes}{operator product expansions}

\newcommand{\rhs}{right-hand side}
\newcommand{\rhss}{right-hand sides}
\newcommand{\ns}{Neveu-Schwarz}
\newcommand{\pbw}{Poincar\'{e}{-}Birkhoff{-}Witt}

\newcommand{\sltwo}{\SLA{sl}{2}}
\newcommand{\myosp}{\SLSA{osp}{1}{2}}
\newcommand{\myaosp}{\AKMSA{osp}{1}{2}}
\newcommand{\algg}{\alg{g}}                    
\newcommand{\algh}{\alg{h}}                    
\newcommand{\aalgg}{\affine{\alg{g}}}          
\newcommand{\aalgh}{\affine{\alg{h}}}          
\newcommand{\aalgb}{\affine{\alg{b}}}          
\newcommand{\aalgn}{\affine{\alg{n}}}          
\newcommand{\aalgz}{\affine{\alg{z}}}          
\newcommand{\aalgp}{\affine{\alg{p}}}          

\newcommand{\csl}{\mathsf{C}(\algh,\EnvAlg{\sltwo})}    
\newcommand{\cosp}{\mathsf{C}(\algh,\EnvAlg{\myosp})}   

\theoremstyle{plain}
\newtheorem{thm}{Theorem}
\newtheorem{prop}[thm]{Proposition}
\newtheorem{lem}[thm]{Lemma}
\newtheorem{cor}[thm]{Corollary}

\newtheorem{lemma}[thm]{Lemma}
\newtheorem{proposition}[thm]{Proposition}

\newtheorem*{remark}{Remark}

\Crefname{thm}{Theorem}{Theorems}
\Crefname{prop}{Proposition}{Propositions}
\Crefname{lem}{Lemma}{Lemmas}
\Crefname{cor}{Corollary}{Corollaries}
\Crefname{defin}{Definition}{Definitions}

\begin{document}

\title[]{An admissible level $\widehat{\mathfrak{osp}} \left( 1 \middle\vert 2 \right)$-model: \\ modular transformations and the Verlinde formula}

\author[J~Snadden]{John Snadden}
\address[John Snadden]{
Mathematical Sciences Institute \\
Australian National University \\
Acton, Australia, 2601.
}
\email{john.snadden@anu.edu.au}

\author[D~Ridout]{David Ridout}
\address[David Ridout]{
School of Mathematics and Statistics \\
University of Melbourne \\
Parkville, Australia, 3010.
}
\email{david.ridout@unimelb.edu.au}

\author[S~Wood]{Simon Wood}
\address[Simon Wood]{
School of Mathematics\\
Cardiff University \\
Cardiff, United Kingdom, CF24 4AG.
}
\email{woodsi@cardiff.ac.uk}

\subjclass[2010]{Primary 17B69, 81T40; Secondary 17B10, 17B67}

\begin{abstract}
	The modular properties of the simple vertex operator superalgebra associated to the affine Kac-Moody superalgebra $\AKMSA{osp}{1}{2}$ at level $-\frac{5}{4}$ are investigated.  After classifying the relaxed highest-weight modules over this vertex operator superalgebra, the characters and supercharacters of the simple weight modules are computed and their modular transforms are determined.  This leads to a complete list of the Grothendieck fusion rules by way of a continuous superalgebraic analogue of the Verlinde formula.  All Grothendieck fusion coefficients are observed to be non-negative integers.  These results indicate that the extension to general admissible levels will follow using the same methodology once the classification of relaxed highest-weight modules is completed.
\end{abstract}

\maketitle

\onehalfspacing

\section{Introduction} \label{sec:Intro}

The construction of \cfts{} from affine Kac-Moody algebras $\affine{\alg{g}}$ at fractional levels has a long history.  These theories were first proposed by Kent \cite{KenInf86} for $\affine{\alg{g}} = \AKMA{sl}{2}$ as a means of generalising the coset construction of \cite{GodVir85} to non-unitary Virasoro minimal models.  Shortly thereafter, Kac and Wakimoto discovered \cite{KacMod88} that for certain fractional levels (called the \emph{admissible levels}), $\affine{\alg{g}}$ possesses a finite set of simple \hwms{} whose characters span, in a sense, a representation of the modular group $\SLG{SL}{2;\ZZ}$.  It was natural then to expect that one could build a rational \cft{} from these \hwms{}.  However, Koh and Sorba immediately noticed \cite{KohFus88} that this expectation failed, even for $\affine{\alg{g}} = \AKMA{sl}{2}$, because Verlinde's formula \cite{VerFus88} for the (necessarily non-negative integer) fusion coefficients always returned at least one negative number.

Subsequent work \cite{BerFoc90,MatFra90,AwaFus92,RamNew93,FeiFus94,AndOpe95,PetFus96,DonVer97,FurAdm97,MatPri99} on this observation did little to ameliorate the confusion.  However, physicists eventually found reason to consider modules (again for $\affine{\alg{g}} = \AKMA{sl}{2}$) that were neither \hw{} \cite{FeiEqu98,SemEmb97,MalStr01} nor simple \cite{GabFus01,LesLog04}.  Indeed, it seemed that admissible level $\AKMA{sl}{2}$-theories naturally allowed for a continuously parametrised family of simple non-\hwms{}, a fact that had been previously discovered \cite{AdaVer95} by Adamovi\'{c} and Milas.

The root cause of the negative fusion coefficients, predicted by the Verlinde formula, remained obscure until recently.  In \cite{RidSL208}, a careful analysis of $\AKMA{sl}{2}$ at the admissible level $k=-\frac{1}{2}$ showed that the negative results could be traced back to the fact that the simple module characters were not linearly independent.  More precisely, the fundamental error in the preceding analyses was demonstrated to be that the modular transformations of the characters of Kac and Wakimoto did not respect their non-trivial convergence properties.  Subsequent work \cite{RidSL210,CreMod12,CreMod13} extended this to all admissible levels for $\AKMA{sl}{2}$ and proved that properly accounting for convergence regions (by treating characters as distributions, not meromorphic functions) indeed resulted in non-negative integer fusion coefficients.  Moreover, the corresponding Grothendieck fusion rules agreed perfectly with the fusion rules that were known \cite{GabFus01,RidFus10} from independent computations.

We remark that these successes were obtained as one instance of a rather more general methodology, dubbed the \emph{standard module formalism} \cite{CreLog13,RidVer14}, for modular properties and Verlinde-like formulae in \lcfts{}.  Originating in work on theories based on the affine Kac-Moody superalgebra $\AKMSA{gl}{1}{1}$ \cite{CreRel11,CreWAl11,AlfMoc14}, this formalism has since been applied to a wide range of \lcfts{} \cite{BabTak12,CreFal13,RidMod13,RidBos14,CreReg14,MorKac15,CanFusI15,CanFusII15,AugMod17}, all of which are in some sense related to rank $1$ objects such as the $\grp{A}_1$ lattice.

There is therefore a need to explore higher rank \lcfts{} and the standard module formalism is expected to be crucial to this endeavour.  The analysis of higher rank theories constructed from affine Kac-Moody algebras (and superalgebras) at admissible levels is particularly attractive because we expect that they will play a central role in understanding logarithmic models, just as the Wess-Zumino-Witten models do in the rational case.  However, this analysis is still in its infancy.  The relevant simple \hwms{} have been identified by Arakawa \cite{AraRat16} for all admissible levels, but a complete set of (positive-energy) simple modules has only recently been described for $\AKMA{sl}{3}$ \cite{AraWei16}.  More generally, almost nothing is known aside from some partial level-specific results for affine algebras \cite{AdaSom94,AdaCon03,PerVer07,PerVer08,AxtVer11,AxtVer14,AdaRea16} and superalgebras \cite{BowRep97,BowCha98,JohMod00,JohFus01,SalGL106,GotWZN07,SalSU207,QueFre07}.

In this paper, we shall not embark immediately on a study of higher rank \lcfts{}.  Rather, we will describe in detail a particular example based on the affine Kac-Moody superalgebra $\AKMSA{osp}{1}{2}$ at the admissible level $k = -\frac{5}{4}$.  The aim here is to develop and test the standard module formalism in the presence of fermionic degrees of freedom (and determine the precise role of the Ramond sector) before tackling the more challenging, but also more physically interesting, cases of $\AKMSA{sl}{2}{1}$ and $\AKMSA{psl}{2}{2}$.\footnote{We mention that the corresponding analysis for the $\AKMSA{gl}{1}{1}$ \lcft{}, carried out in \cite{CreRel11}, was restricted to the \ns{} sector as the simple characters of this sector closed on themselves under modular transformations.  The same is \emph{not} true for $\AKMSA{osp}{1}{2}$ \cfts{}.}  This particular level is an attractive starting place for two reasons: first, it describes one of the ``smallest'' $\AKMSA{osp}{1}{2}$ minimal models (meaning that it has very few simple \hwms{}) and, second, it is an order $2$ simple current extension of the $\AKMA{sl}{2}$ minimal model of the same (admissible) level \cite{CreMod13}.  The latter property allows us to independently check our $\AKMSA{osp}{1}{2}$ results against the known $\AKMA{sl}{2}$ results.

Of course, \cfts{} with $\AKMSA{osp}{1}{2}$ symmetry have been studied in the past, both at integer and fractional levels \cite{FanMod93,EnnFus97,IohFus01a}.  However, these works only considered simple \hwms{} in the \ns{} (untwisted) sector, ignoring the known issue of negative fusion coefficients.  Here, we discuss a more complete spectrum of simple modules (as well as some of the reducible but indecomposable ones) in both the \ns{} and Ramond sectors.  We moreover emphasise the global parity of each module in order to be able to distinguish the relative parities of the direct summands appearing in each Grothendieck fusion product.  In physics parlance, this is equivalent to computing both the even and odd Grothendieck fusion rules of Sotkov and Stanishkov \cite{SotN=186}.

The results confirm that the standard module formalism applies to the affine superalgebra theory studied here:  characters and supercharacters close under modular transformations and the Grothendieck fusion coefficients are verified to be non-negative integers.  The methodology developed in this paper also applies to the other admissible levels of $\AKMSA{osp}{1}{2}$, so extending these results to general admissible levels will be straightforward, assuming that one can first classify the relaxed \hwms{}.  The latter classification has not yet been completed, though we expect that it can be obtained using the methods developed in \cite{TsuExt13,RidJac14,RidRel15,BloSVir16}.  Because of this, the Grothendieck fusion rules of the $\AKMSA{osp}{1}{2}$ models for general admissible levels will instead be addressed in a forthcoming paper \cite{CreOSP17} using coset technology.  A byproduct of this work will be the relaxed \hwm{} classification for $\AKMSA{osp}{1}{2}$ models of general admissible level.

We begin, in \cref{sec:finosp}, with a quick review of the simple Lie superalgebra $\SLSA{osp}{1}{2}$ and its representation theory.  We prove, in particular, a classification result (\cref{thm:ospclass}) for all simple weight modules of $\SLSA{osp}{1}{2}$ that have at least one finite-dimensional weight space (we were unable to find this result in the literature).  This is followed, in \cref{sec:affosp}, by a quick review of the affine Kac-Moody superalgebra $\AKMSA{osp}{1}{2}$, its automorphisms (conjugation and spectral flow), and the associated \vosas{}.  We also discuss \emph{relaxed} Verma modules over $\AKMSA{osp}{1}{2}$ and their simple quotients, borrowing this notion from \cite{FeiEqu98} where it was introduced for $\AKMA{sl}{2}$ (see \cite[Sec.~2.1]{RidRel15} for a general definition of relaxed \hwms{}).

\Cref{sec:minmod} then specialises to the simple $\AKMSA{osp}{1}{2}$ \vosa{} of level $k=-\frac{5}{4}$ that we study in this work, denoting it by $\OSPMinMod{2}{4}$.  We first give an efficient characterisation of affine Zhu algebras, twisted and untwisted, and identify those of the universal \vosas{} (\cref{prop:UnivZhu}) before explicitly computing the Zhu algebras of $\OSPMinMod{2}{4}$ (\cref{prop:SimpZhu,prop:SimpZhu'}).  This is then used to classify the simple relaxed \hw{} $\OSPMinMod{2}{4}$-modules (\cref{thm:classB24NS,thm:classB24R}) and identify some of the reducible ones.  These are partitioned into standard, typical and atypical modules as per the standard module formalism of \cite{CreLog13,RidVer14}.

Having classified the simple (and standard) $\OSPMinMod{2}{4}$-modules, we turn to the computation of their characters and supercharacters in \cref{sec:chars}.  Such character formulae are easy to compute for the \ns{} \hwms{} because the submodule structure of the associated Verma modules was determined by Iohara and Koga \cite{IohFus01a}.  Spectral flow automorphisms then allow us to deduce the analogous Ramond formulae (\cref{prop:AtypCh}).  We explicitly note the convergence regions of these characters, treated as meromorphic functions, and use the results to determine the characters of the relaxed \hwms{}, treated as distributions (\cref{prop:AtypStCh,prop:TypCh}).  Supercharacter formulae follow easily (\cref{prop:RelSch}) and we conclude by introducing the Grothendieck group of (an appropriate category of) $\OSPMinMod{2}{4}$-modules and showing explicitly that the images of the standard modules form a basis of (a completion of) this Grothendieck group.

This last result (\cref{cor:resgr}) is the key to computing the modular transforms of the $\OSPMinMod{2}{4}$-(super)characters, the topic of \cref{sec:modver}.  We begin by introducing slightly unfamiliar S and T coordinate transforms (following \cite{RidBos14}) before computing the modular group action on the span of the standard $\OSPMinMod{2}{4}$-(super)characters.  Of note is that the S-transform amounts to a generalised Fourier transform on (a countably-infinite number of copies of) the real vector space $\algh_{\RR}$ spanned by the fundamental weight of $\SLSA{osp}{1}{2}$.  The S-transforms are then extended to the simple atypical (super)characters using \cref{cor:resgr}.  We remark that trying to compute these S-transforms directly from the meromorphic characters would lead to nonsensical results (such as negative fusion coefficients) because the modular S-transform does not preserve the convergence regions of the (super)characters.

Finally, \cref{sec:fusion} addresses the Grothendieck fusion rules of the simple (and standard) $\OSPMinMod{2}{4}$-modules.  First, we deduce a version (\cref{thm:Verlindeformula}) of the standard Verlinde formula that works for this \vosa{} --- generally, Verlinde formulae are only expected to apply directly to ($\ZZ$-graded) \voas{}.  The method follows the approach of \cite{EhoFus94} for the $N=1$ minimal model \vosa{} (see also \cite{CanFusII15}) wherein one lifts the Verlinde formula from the bosonic orbifold using simple current technology.  With this formula in hand, we compute all Grothendieck fusion rules, including global parity information, among the simple and standard $\OSPMinMod{2}{4}$-modules (\cref{thm:fuseformulae}).

\section*{Acknowledgements}

DR thanks Kenji Iohara for illuminating discussions on the structure of Verma modules over $\AKMSA{osp}{1}{2}$.
We would like to thank the anonymous referee whose careful reading of the original manuscript and many suggestions significantly improved the article.
JS's research is supported by a University Research Scholarship from the Australian National University.
DR's research is supported by the Australian Research Council Discovery Projects DP1093910 and DP160101520 as well as the Australian Research Council Centre of Excellence for Mathematical and Statistical Frontiers CE140100049.
SW's research is supported by Australian Research Council Discovery Early Career Researcher Award DE140101825 and the Australian Research Council Discovery Project DP160101520.

\section{The basic Lie superalgebra $\myosp$} \label{sec:finosp}

In this section, we quickly review the theory of weight modules over $\sltwo$ and $\myosp$.  The latter algebra is important because of its role as the horizontal subalgebra of the \ns{} $\myaosp$ algebra, the former plays the same role for the Ramond $\myaosp$ algebra.

\subsection{A brief review of $\sltwo$} \label{subsec:sl2}

The simple complex Lie algebra $\VOA{A}_1=\sltwo$ has Cartan-Weyl basis $\set{h,e,f}$, satisfying the commutation relations
\begin{equation} \label{eqn:slcomm}
\comm{h}{e} = 2e, \quad \comm{h}{f} = -2f, \quad \comm{e}{f} = h.
\end{equation}
The Cartan subalgebra $\alg{h} = \CC h$ then gives rise to the root system $\set{\dot{\alpha}, -\dot{\alpha}} \subset \alg{h}^*$ where $\dot{\alpha}(h)=2$ and we choose $\dot{\alpha}$ to be the lone fundamental root. The non-zero entries of the (appropriately normalised) Killing form $\kappa$ on $\sltwo$, with respect to the given basis, are
\begin{equation} \label{eqn:slkill}
\killing{h}{h} = 2, \quad \killing{e}{f} = \killing{f}{e} = 1.
\end{equation}
This induces an inner product on $\alg{h}^*$, defined by $\brac{\dot{\alpha},\dot{\alpha}} = 2$. From these data, one calculates the fundamental weight to be $\omega = \frac{1}{2}\dot{\alpha}$ and that the algebra has dual Coxeter number $\dcox=2$. The Weyl group is isomorphic to $\ZZ_2$, generated by the root reflection $\dot{\alpha} \mapsto -\dot{\alpha}$.

As with all complex semisimple Lie algebras, the finite-dimensional modules of $\sltwo$ are necessarily semisimple weight modules and the finite-dimensional simple modules are uniquely determined (up to isomorphism) by their highest weight. For each $\lambda \in \NN$, we denote the unique (up to isomorphism) $\brac{\lambda+1}$-dimensional simple $\sltwo$-module of highest weight $\lambda \omega$ (and lowest weight $-\lambda \omega$) by $\slFin{\lambda}$.

Extending to the infinite-dimensional case, we no longer have complete reducibility, though we can nevertheless classify the simple weight modules.  Here we include in the definition of a weight module that all of its weight spaces are required to be finite-dimensional. The first class that we consider are the simple \hwms{} $\slInf{\lambda}^+$, where $\lambda \in \CC \setminus \NN$, with highest weight $\lambda \omega$ and no lowest weight. It is straightforward to show that the weight support (the set of weights with non-trivial weight spaces) of such a module is $(\lambda - 2 \NN) \omega$ and that all the weight spaces are one-dimensional. Similarly, we also have the simple \lwms{} $\slInf{\lambda}^-$ where $\lambda \in \CC \setminus \ZZ_{\le 0}$. Here, the weight support is instead $(\lambda + 2 \NN) \omega$ and again all weight spaces are one-dimensional.

Finally, we have the \emph{dense} modules $\slRel{\Lambda,q}$, parametrised by $q \in \CC$ and $\Lambda \in \CC / 2 \ZZ$. These are simple precisely when $q \neq \frac{1}{2} \lambda (\lambda + 2)$ for all $\lambda \in \Lambda$. The weight support of $\slRel{\Lambda,q}$ is precisely $\Lambda \omega$ (that is, $(\lambda_0 + 2 \ZZ) \omega$ for some $\lambda_0$), again with all weight spaces one-dimensional. The parameter $q$ is the (unique) eigenvalue of the Casimir element $Q = \frac{1}{2}h^2 + ef + fe$ which generates the centre of the \uea{} and must therefore act as a scalar multiple of the identity. We remark that dense modules are also referred to as \emph{cuspidal} and \emph{torsion-free} in the literature.

Having introduced the above classes of modules, we can state the following result (see \cite{MazLec10}).
\begin{thm}[Classification of simple $\sltwo$ weight modules] \label{thm:slclass}
Every simple $\sltwo$ weight module is isomorphic to one of the following mutually non-isomorphic modules:
\begin{enumerate}
\item $\slFin{\lambda}$ with $\lambda \in \NN$;
\item $\slInf{\lambda}^+$ with $\lambda \in \CC \setminus \NN$;
\item $\slInf{\lambda}^-$ with $\lambda \in \CC \setminus \ZZ_{\le 0}$;
\item $\slRel{\Lambda, q}$ with $q \in \CC$, $\Lambda \in \CC / 2 \ZZ$ and $q \neq \frac{1}{2} \lambda (\lambda+2)$ for all $\lambda \in \Lambda$.
\end{enumerate}
\end{thm}

We shall also need to consider the reducible, but indecomposable, dense modules that correspond to parameters $\Lambda$ and $q$, where $q = \frac{1}{2} \lambda (\lambda + 2)$ for some $\lambda \in \Lambda$.  We note that these do not exhaust the reducible but indecomposable dense modules of $\AKMA{sl}{2}$.  The latter are classified (somewhat explicitly) in \cite{MazLec10}.  However, the others will not be needed in what follows.  If this condition is met, then either the unique (up to rescaling) state $v_{\lambda}$ of weight $\lambda \omega$ is a \hwv{} or $e v_{\lambda}$ is a \lwv{}.  Both possibilities occur independently:\footnote{We assume throughout that when $\lambda$ solves the reducibility condition, then it is the unique element of $\Lambda$ that does so.  This need not be the case, as for certain $q$ there are two solutions for $\lambda$ in $\Lambda$.  However, the corresponding reducible modules will, again, not be needed here.} the first gives rise to a \hw{} submodule isomorphic to the \hwm{} $\slInf{\lambda}^+$, while the second gives rise to a \lw{} submodule isomorphic to the \lwm{} $\slInf{\lambda+2}^-$.  The structures of these indecomposable dense $\sltwo$-modules, which we denote by $\slRel{\Lambda, q}^+$ and $\slRel{\Lambda, q}^-$, respectively, where $\Lambda = \lambda + 2 \ZZ$, are thus determined by the following non-split short exact sequences:
\begin{equation} \label{eqn:sese}
	\dses{\slInf{\lambda-2}^+}{}{\slRel{\lambda + 2 \ZZ, \lambda (\lambda - 2) / 2}^+}{}{\slInf{\lambda}^-}, \qquad
	\dses{\slInf{\lambda+2}^-}{}{\slRel{\lambda + 2 \ZZ, \lambda (\lambda + 2) / 2}^-}{}{\slInf{\lambda}^+}.
\end{equation}
We have shifted $\lambda$ by $2$ in the first sequence for clarity.

We recall a concrete construction of certain dense $\sltwo$-modules that will be useful when we generalise to dense $\myosp$-modules in \cref{subsec:strthy}.  First, note that the elements of the \uea{} which commute with $h \in \sltwo$ (more generally, with the Cartan subalgebra) will preserve weight spaces. Such elements form the centraliser $\csl$, which (by the \pbw{} theorem) is just the polynomial subalgebra $\CC \sqbrac{h,Q} \subseteq \EnvAlg{\sltwo}$. As this is abelian, its simple modules are all one-dimensional. Suppose that $\Mod{W}_{\lambda, q} = \CC w$ is one such module, with $h$ and $Q$ acting as complex scalars $\lambda$ and $q$, respectively. We can then induce from this to the full $\EnvAlg{\sltwo}$-module
\begin{equation} \label{eqn:slinduc}
\finite{\Mod{W}}_{\lambda, q} = \operatorname{Ind}_{\csl}^{\EnvAlg{\sltwo}} \Mod{W}_{\lambda, q},
\end{equation}
which (again using the \pbw{} theorem) has basis
\begin{equation} \label{slinducbasis}
\set{w, e^n w, f^n w \st n \in \ZZ_{>0}}.
\end{equation}
The induced module $\finite{\Mod{W}}_{\lambda, q}$ obviously has weight support $\Lambda \omega$, where $\Lambda = \lambda + 2 \ZZ$, and one-dimensional weight spaces, hence it is dense.  It is moreover clear that $\finite{\Mod{W}}_{\lambda, q}$ is a simple $\sltwo$-module, and is thus isomorphic to $\slRel{\Lambda, q}$, unless $e^n w$ is a \lwv{} or $f^n w$ is a \hwv{}, for some $n \in \ZZ_{>0}$.  When $\finite{\Mod{W}}_{\lambda, q}$ is not simple, it is isomorphic to $\slRel{\Lambda, q}^-$ or $\slRel{\Lambda, q}^+$, respectively.

\subsection{A brief review of $\myosp$} \label{subsec:strthy}

The simple complex Lie superalgebra $\algg = \myosp$ has basis $\set{h,e,f,x,y}$, where the elements of $\algg^{\even} = \vspn{\set{h,e,f}}$ and $\algg^{\odd} = \vspn{\set{x,y}}$ are declared to be even and odd, respectively. The even subalgebra is (as its elements suggest) isomorphic to $\sltwo$, thus \eqref{eqn:slcomm} still holds. The remaining (anti)commutation relations are:
\begin{equation} \label{eqn:ospcomm}
\begin{aligned}
\comm{h}{x}&=x, &\comm{e}{x}&=0, &\comm{f}{x}&=-y, \\
\comm{h}{y}&=-y, &\comm{e}{y}&=-x, &\comm{f}{y}&=0, \\
\acomm{x}{y}&=h, &\acomm{x}{x}&=2e, &\acomm{y}{y}&=-2f.
\end{aligned}
\end{equation}
In this basis, the (rescaled) Killing form has non-zero entries given by
\begin{equation} \label{eqn:ospkill}
\killing{h}{h} = 2, \quad \killing{e}{f} = \killing{f}{e} = 1, \quad \killing{x}{y} = - \killing{y}{x} = 2.
\end{equation}
Due to the existence of a non-degenerate even supersymmetric bilinear form, $\myosp$ is an example of a \emph{basic} Lie superalgebra \cite{CheDua12}. In the classification \cite{KacLie77} of such algebras, the isomorphism class of $\myosp$ is denoted by $\VOA{B}_{0\vert 1}$.

We consider the Cartan subalgebra $\alg{h} = \CC h$, with root system $\set{-2\alpha, -\alpha, \alpha, 2\alpha}$, where $\alpha(h)=1$ (so $2 \alpha$ is identified with $\dot{\alpha}$), and choose $\alpha$ to be positive (hence simple). The inner product on $\alg{h}^*$ induced by the Killing form is given by $\brac{\alpha,\alpha}=\frac{1}{2}$, from which one can calculate that the dual Coxeter number is $\dcox=\frac{3}{2}$ and that the fundamental weight is $\alpha$. Since the Weyl group of a Lie superalgebra is generated by reflections in the even roots, it is precisely the Weyl group of its even subalgebra.  As the even subalgebra of $\myosp$ is isomorphic to $\sltwo$, its Weyl group is also of order 2, generated by the reflection $\alpha \mapsto -\alpha$.

As with all superalgebras, modules $\Mod{M}$ of $\myosp$ are required to carry a compatible $\ZZ_2$-grading: that is, they must decompose as a direct sum $\Mod{M}^{\even} \oplus \Mod{M}^{\odd}$, such that $\alg{g}^{(i)} \Mod{M}^{(j)} \subseteq \Mod{M}^{(i+j)}$, for all $i,j \in \ZZ_2$. Having identified such an $\Mod{M}^{\even}$ and $\Mod{M}^{\odd}$, these summands are then referred to as the even and odd subspaces, respectively. Similarly, elements of the even and odd subspaces are said to have even and odd parity, respectively. However, it should be apparent that reversing these labels, whilst maintaining the same module structure, still gives a valid grading. As such, on any category of modules we might consider, we require there to be an involutive functor $\parr$ taking any module to its \emph{parity reversal}.  In principle, a module may be isomorphic to its parity reversal.  This does not happen for the simple weight modules of $\myosp$.

All finite-dimensional $\myosp$-modules are semisimple weight modules. The simple ones must have a unique highest (and lowest) weight. Indeed, for each $\lambda \in \NN$, there is a unique (up to isomorphism) $(2 \lambda + 1)$-dimensional simple $\myosp$-module with highest weight $\lambda \alpha$ and lowest weight $-\lambda \alpha$, for which the \hwvs{} (and thus also the \lwvs{}) are assigned even parity. We will denote this module by $\Fin{\lambda}$. It has weight support \mbox{$\set{\mu \alpha \st \mu \in \ZZ, \; \abs{\mu} \le \lambda}$}, with all weight spaces one-dimensional.

In addition, for each $\lambda \in \CC \setminus \NN$, $\myosp$ has an infinite-dimensional simple \hwm{} $\Inf{\lambda}^+$, generated by an even \hwv{} of weight $\lambda \alpha$, whose weight support is $(\lambda - \NN) \alpha$. Similarly, there is also the simple \lwm{} $\Inf{\lambda}^-$, for each $\lambda \in \CC \setminus \ZZ_{\le 0}$.  This module is generated by an even \lwv{} of weight $\lambda \alpha$ and its weight support is $(\lambda + \NN) \alpha$.

The modules listed above, together with their parity reversals, exhaust the simple highest- and \lwms{} of $\myosp$.  The proof is elementary, following the same steps used to classify \hw{} $\sltwo$-modules. Moreover, again as with $\sltwo$, there is an additional infinite family of simple weight modules with no highest nor lowest weights. However, a little care is needed here to characterise these in a meaningful way.

Recall from \cref{subsec:sl2} the concrete construction of certain dense $\sltwo$-modules.  For $\algg = \myosp$, the centre of the \uesa{} $\EnvAlg{\algg}$ is still generated by a Casimir element
\begin{equation} \label{eqn:casimir1}
Q' = \frac{1}{2} h^2 + ef + fe - \frac{1}{2} xy + \frac{1}{2} yx,
\end{equation}
which must therefore act as a scalar on any simple module, but the centraliser $\cosp$ is not a polynomial algebra in $h$ and $Q'$. For example, by rewriting the previous equation in the form
\begin{equation} \label{eqn:casimir2}
Q' = \frac{1}{2} h^2 + \frac{1}{2} h - yx + 2 (yx)^2 - 2 h(yx),
\end{equation}
we see that $yx$ cannot be expressed as a polynomial in these elements, though indeed $(yx)h = h(yx)$.

This motivates introducing the \emph{super-Casimir} (or sCasimir)\cite{ArnCas97}
\begin{equation} \label{eqn:scasimir}
\Sigma = xy - yx + \frac{1}{2} \in \cosp.
\end{equation}
Though this element of $\EnvAlg{\algg}$ is not central, it satisfies
\begin{equation} \label{eqn:scomm}
\comm{\Sigma}{\alg{g}^{\even}} = \acomm{\Sigma}{\alg{g}^{\odd}} = 0,
\end{equation}
from which it follows that $\Sigma$ is diagonalisable on a simple weight module, taking eigenvalues $s$ and $-s$ on the even and odd subspaces, respectively, for some $s \in \CC$. We can now identify $\cosp$ with the polynomial algebra $\CC[h,\Sigma]$. In particular, we may write $yx$ and $Q'$ as polynomials in $h$ and $\Sigma$ as follows:
\begin{equation}
	yx = \frac{1}{2}\brac*{h - \Sigma + \frac{1}{2}}, \quad Q' = \frac{1}{2} \Sigma^2 - \frac{1}{8}.
\end{equation}
We note that the eigenvalue of $\Sigma$ is $\lambda + \frac{1}{2}$ on a \hwv{}
of weight $\lambda \alpha$ and $-\lambda + \frac{1}{2}$ on a \lwv{} of the
same weight.  As mentioned above, this eigenvalue is denoted by $s$ if the highest-/\lwv{} is even and $-s$ if it is odd.

One can carry out an induction procedure analogous to that described by \eqref{eqn:slinduc}, giving weight $\algg$-modules $\finite{\Mod{W}}'_{\lambda, s}$ with bases
\begin{equation} \label{eqn:ospinducbasis}
\set{w,x^n w, y^n w \st n \in \ZZ_{>0}},
\end{equation}
where $h w = \lambda w$ and $\Sigma w = s w$, for some $\lambda, s \in \CC$, assigning even parity to $w$. As in the $\AKMA{sl}{2}$ case, this is a \emph{dense} module: its weight support is $\lambda + \ZZ$ and its weight spaces are one-dimensional.

$\finite{\Mod{W}}'_{\lambda, s}$ is reducible if and only if it has either a highest- or a \lwv{}, thus if one of the basis elements above is annihilated by either $x$ or $y$. If $s = \mu + \frac{1}{2}$, for some $\mu \in \lambda + 2 \ZZ$ (so that the corresponding weight vector $v_{\mu}$ has even parity), then either $v_{\mu}$ is a \hwv{} or $x v_{\mu}$ is a \lwv{}.  Similarly, if $s = -\mu + \frac{1}{2}$, for some $\mu \in \lambda + 2 \ZZ$, then either $v_{\mu}$ is a \lwv{} or $y v_{\mu}$ is a \hwv{}.\footnote{The analogous analysis in which $v_{\mu}$ has odd parity leads to equivalent constraints on $s$.}  If neither constraint is satisfied, that is if $\mu^2 \neq (s-\frac{1}{2})^2$ for every $\mu \in \lambda + 2 \ZZ$, then $\finite{\Mod{W}}'_{\lambda, s}$ is simple. In this instance, we can unambiguously label these simple dense modules as $\Rel{\Lambda, s}$, where $\Lambda = \lambda + 2 \ZZ$, in analogy with the notation used for simple dense $\sltwo$-modules.  Clearly, $\Rel{\Lambda, s}$ has weight support $\brac*{\Lambda + \set{0,1}} \alpha$ and $\Sigma$ acts as $(-1)^j s$ on weight vectors whose weights lie in $\brac*{\Lambda + j} \alpha$.  We also note the isomorphisms
\begin{equation} \label{eqn:modperiodicity}
\Rel{\Lambda, s} \cong \parrmod{\Rel{\Lambda+1, -s}}.
\end{equation}

This argument shows that the simple weight modules are classified in a manner entirely analogous to \cref{thm:slclass}. In particular we have:
\begin{thm}[Classification of simple $\myosp$ weight modules] \label{thm:ospclass}
Every simple $\myosp$ weight module is isomorphic to one of the following mutually non-isomorphic modules, or their parity reversals:
\begin{enumerate}
\item $\Fin{\lambda}$ with $\lambda \in \NN$; \label{ospclass1}
\item $\Inf{\lambda}^+$ with $\lambda \in \CC \setminus \ZZ_{\ge 0}$; \label{ospclass2}
\item $\Inf{\lambda}^-$ with $\lambda \in \CC \setminus \ZZ_{\le 0}$; \label{ospclass3}
\item $\Rel{\Lambda, s}$ with $s \in \CC$, $\Lambda \in \CC / 2\ZZ$ and $\lambda^2 \neq (s-\frac{1}{2})^2$ for every
$\lambda \in \Lambda$. \label{ospclass4}
\end{enumerate}
\end{thm}

\begin{proof}
Let $\Mod{M}$ be a simple weight module over $\myosp$ and suppose that $w \in \Mod{M}$ is a vector of weight $\lambda \alpha$ so that $h w = \lambda w$.  If the even subspace of $\Mod{M}$ is zero, then $w$ is odd and both $x w$ and $y w$ vanish, so that $\Mod{M} = \CC w \cong \parrmod{\Fin{0}}$. Otherwise, without loss of generality, we may assume that $w$ is of even parity.

Simplicity implies that $\Mod{M} = \EnvAlg{\algg} w$ and, in particular, that the weight space of weight $\lambda \alpha$ is
\begin{equation}
	\Mod{M}\brac*{\lambda} = \cosp w.
\end{equation}
Now, if $\Mod{M}\brac*{\lambda}$ had a $\cosp$-submodule, then this would generate a proper $\algg$-submodule of $\Mod{M}$, contradicting its simplicity. It therefore must be that $\Mod{M}\brac*{\lambda}$ is a simple $\cosp$-module and, since $\cosp$ is abelian, $\Mod{M}(\lambda)$ is thus one-dimensional. It follows that $w$ is an eigenvector of $\Sigma$ with, say, $\Sigma w = s w$ and, by the \pbw{} theorem, $\Mod{M}$ is spanned by
\begin{equation} \label{eqn:ospmodbasis}
\set{w, x^n w, y^n w \st n \in \ZZ_{>0}}.
\end{equation}
By iterative application of the (anti)commutation relations \eqref{eqn:ospcomm}, one can uniquely determine the action of any element of $\algg$ on $\Mod{M}$ in terms of the parameters $\lambda$ and $s$. If the spanning set \eqref{eqn:ospmodbasis} is linearly dependent, then $\Mod{M}$ must be highest- and/or lowest-weight and thus belongs to classes \ref{ospclass1}, \ref{ospclass2} or \ref{ospclass3}, as discussed above. Otherwise $\Mod{M} \cong \finite{\Mod{W}}'_{\lambda, s}$, so it belongs to class \ref{ospclass4}.
\end{proof}

Whilst all finite-dimensional $\myosp$-modules are semisimple and can therefore be decomposed into a direct sum of a finite number of the $\Fin{\lambda}$ and $\parrmod{\Fin{\lambda}}$, there are infinite-dimensional modules which are not. In particular, the dense module $\finite{\Mod{W}}'_{\lambda, \lambda + 1/2}$ is reducible, but indecomposable, and is characterised by the following non-split short exact sequence:
\begin{subequations} \label{eqn:sesc}
\begin{equation} \label{eqn:sesc-}
\dses{\parrmod{\Inf{\lambda + 1}^-}}{}{\finite{\Mod{W}}'_{\lambda, \lambda + 1/2}}{}{\Inf{\lambda}^+}.
\end{equation}
We shall denote this reducible module $\finite{\Mod{W}}'_{\lambda, \lambda + 1/2}$ by $\Rel{\Lambda,s}^-$, where $\Lambda = \lambda + 2 \ZZ$ and $s = \lambda + \frac{1}{2}$, to emphasise its denseness. The superscript $-$ refers to the existence of a \lw{} submodule. Indeed, one similarly arrives at the non-split short exact sequence
\begin{equation} \label{eqn:sesc+}
\dses{\parrmod{\Inf{\lambda - 1}^+}}{}{\finite{\Mod{W}}'_{\lambda, -\lambda + 1/2}}{}{\Inf{\lambda}^-}
\end{equation}
\end{subequations}
in an entirely analogous manner.  We therefore denote $\finite{\Mod{W}}'_{\lambda, -\lambda + 1/2}$ by $\Rel{\Lambda,s}^+$, where $\Lambda = \lambda + 2 \ZZ$, $s=-\lambda + \frac{1}{2}$ and the superscript $+$ indicates a \hw{} submodule.

It will turn out, in \cref{sec:minmod}, that these reducible, but indecomposable, modules are the keys to the analysis of the \cft{}.
We remark that if $s \in \ZZ + \frac{1}{2} \setminus \set{\frac{1}{2}}$, then the indecomposable structure of the induced module $\finite{\Mod{W}}'_{\lambda, s}$ is slightly more complicated that that discussed above.  However, this case turns out not to be relevant for the \cft{} that we shall explore, hence it will not be considered any further.  There are, in addition, many other indecomposable dense $\myosp$-modules beyond those discussed here which are likewise irrelevant to what follows.

\subsection{The Weyl group} \label{subsec:weylgrp}

Although we have, in both \cref{subsec:sl2,subsec:strthy}, treated the Weyl group $\SLG{W}{\algh}$ associated with a chosen Cartan subalgebra as the subgroup of $\SLG{GL}{\algh^*}$ generated by (even) root reflections, it is possible (and useful) to view it in a number of other ways. One equivalent definition (for any simple Lie superalgebra $\algg$ with Cartan subalgebra $\algh$) is
\begin{equation} \label{eqn:altweyl}
\SLG{W}{\algh} = \frac{ \SLG{N}{\algh} }{ \SLG{Z}{\algh} },
\end{equation}
where $\SLG{N}{\algh}$ and $\SLG{Z}{\algh}$ are the subgroups of the group $\SLG{Inn}{\algg}$ of inner automorphisms of $\algg$ given by
\begin{equation} \label{eqn:normcent}
\SLG{N}{\algh} = \set{\phi \in \SLG{Inn}{\algg} \st \phi(\algh) = \algh}, \quad
\SLG{Z}{\algh} = \set{\phi \in \SLG{Inn}{\algg} \st \phi(x)=x\ \text{for all}\ x \in \algh}.
\end{equation}
We recall that for a simple Lie superalgebra, the group of inner automorphisms is generated by exponentiating the adjoint actions of the even subalgebra elements.  Now, if the quotient \eqref{eqn:altweyl} splits, so that
\begin{equation} \label{eqn:nzsplit}
\SLG{N}{\algh} = \SLG{Z}{\algh} \rtimes \SLG{W}{\algh},
\end{equation}
then we may treat the Weyl group as a subgroup of $\SLG{Inn}{\algg}$; in particular, one which preserves the choice of Cartan subalgebra.

For example, for the Cartan subalgebra $\algh$ of $\sltwo$ used in \cref{subsec:sl2}, $\SLG{N}{\algh}$ is the union of two disjoint subsets: those maps taking $h \mapsto h$ and those taking $h \mapsto -h$. The first is of course $\SLG{Z}{\algh}$, so taking the quotient as in \eqref{eqn:altweyl} indeed gives a copy of the Weyl group $\ZZ_2$. A choice of coset representatives are the identity map $\mathrm{id}_\algg$ and the linear involution defined by
\begin{equation} \label{eqn:slconj}
h \mapsto -h, \quad e \mapsto -f, \quad f \mapsto -e,
\end{equation}
demonstrating the splitting \eqref{eqn:nzsplit}.

Now, for $\myosp$, with $\algh$ as in \cref{subsec:strthy}, we again find that $\SLG{Z}{\algh}$ is a normal subgroup of $\SLG{N}{\algh}$ of index 2, so that indeed the Weyl group is isomorphic to $\ZZ_2$. However, this quotient no longer splits. We can see this by again considering coset representatives. Here, we may choose these to be the identity map and the \emph{conjugation automorphism} $\finconjaut$, which acts according to \eqref{eqn:slconj} on the even subalgebra and on odd elements as
\begin{equation} \label{eqn:ospconj}
\finconjmod{x} = -y, \quad \finconjmod{y} = x.
\end{equation}
Note that $\finconjaut$ is not involutive, but rather squares to an element $\Sid \in \SLG{Z}{\algh}$ which acts as the identity on even elements and minus the identity on odd ones. Indeed, no element of that coset squares to the identity, so unfortunately we cannot here realise the Weyl group as a subgroup of $\SLG{Inn}{\algg}$.

It is also useful to consider automorphisms of $\algg$ as defining invertible functors on the category $\algmod{\algg}$ of $\algg$-modules. Taking any $\varphi \in \SLG{Aut}{\algg}$ and any $\Mod{M} \in \algmod{\algg}$, let $\widetilde{\varphi} \colon \Mod{M} \to \func{\widetilde{\varphi}}{\Mod{M}}$ be an isomorphism of vector superspaces and let $X \in \algg$ act on $\widetilde{\varphi}(m) \in \func{\widetilde{\varphi}}{\Mod{M}}$ according to
\begin{equation} \label{eqn:twistmod}
X \widetilde{\varphi}(m) = \widetilde{\varphi}(\varphi^{-1}(X) m).
\end{equation}
This gives $\func{\widetilde{\varphi}}{\Mod{M}}$ the structure of a $\algg$-module.  We emphasise that $\func{\widetilde{\varphi}}{\Mod{M}}$ may or may not be isomorphic to $\Mod{M}$. The assignment $\Mod{M} \mapsto \func{\widetilde{\varphi}}{\Mod{M}}$ is called \emph{twisting} by $\varphi$ and indeed defines a functor on the category of $\algg$-modules (acting in the obvious way on morphisms).  The resulting homomorphism $\SLG{Aut}{\algg} \lra \SLG{Aut}{\algmod{\algg}}$ is then a strict $\SLG{Aut}{\algg}$-action on $\algmod{\algg}$.  These functors obviously commute with parity reversal:  $\parr \widetilde{\varphi} = \widetilde{\varphi} \parr$.  For notational simplicity, we shall drop the tildes that distinguish the automorphism from its induced functor in what follows.

Two notable properties of these twisting functors are that they preserve indecomposable structures and take weight modules to weight modules, albeit with respect to possibly different Cartan subalgebras. Indeed, if $\Mod{M}$ is a weight module for $\algh$, then $\varphi(\algh)$ is another Cartan subalgebra for which $\func{\varphi}{\Mod{M}}$ is a weight module. However, since all Cartan subalgebras are related to one another by inner automorphisms, we are justified in restricting attention to $\varphi \in \SLG{N}{\algh}$ for a given $\algh$. Moreover, twisting by $\varphi \in \SLG{Z}{\algh}$ takes (isomorphism classes of) weight modules to themselves, thus we need only consider those twists defined (up to isomorphism) by the cosets in $\SLG{W}{\algg}$. We have, in this way, obtained an action of the Weyl group on the set of isomorphism classes of the objects of $\wtmod{\algg}$, the category of weight $\algg$-modules (with fixed Cartan subalgebra $\algh$).  In general, one can also act with (equivalence classes of) those outer automorphisms that preserve the chosen Cartan subalgebra.  However, for $\sltwo$ and $\myosp$, there are no such outer automorphisms.

We illustrate this action of the Weyl group with the following two examples.  If $\algg = \sltwo$ and $\varphi$ is the inner automorphism defined in \eqref{eqn:slconj}, then we have
\begin{subequations} \label{eq:FinConj}
	\begin{alignat}{3}
		\func{\varphi}{\slFin{\lambda}} &\cong \slFin{\lambda}, &\quad
		\func{\varphi}{\slInf{\lambda}^\pm} &\cong \slInf{-\lambda}^\mp, &\quad
		\func{\varphi}{\slRel{\Lambda,q}} &\cong \slRel{-\Lambda,q}. \label{eqn:sltwist}
	\intertext{Similarly, for $\myosp$ we have}
		\finconjmod{\Fin{\lambda}} &\cong \Fin{\lambda}, &\quad
		\finconjmod{\Inf{\lambda}^\pm} &\cong \Inf{-\lambda}^\mp, &\quad
		\finconjmod{\Rel{\Lambda,s}} &\cong \Rel{-\Lambda,s}. \label{eqn:osptwist}
	\end{alignat}
\end{subequations}
Note that in these examples, the induced action on the weight support of the module is precisely that of the corresponding Weyl reflection. This generalises, with the twisting functor corresponding to a Weyl group element acting on a module's weight support via its standard linear action on the dual Cartan subalgebra.

We remark that it is possible to lift the action of the Weyl group from isomorphism classes of weight modules to the category $\wtmod{\algg}$.  This is trivial for $\algg = \sltwo$, because of \eqref{eqn:slconj}, but not for $\algg = \myosp$ as $\finconjmod^2 \neq \wun$.  In the latter case, we instead have natural isomorphisms $\eta_{\Mod{M}} \colon \func{\finconjaut^2}{\Mod{M}} \to \Mod{M}$, for each weight module $\Mod{M}$, given by $\eta_{\Mod{M}} \brac*{\finconjaut^2(m)} = (-1)^{\parityof{m}} m$, for all homogeneous elements $m \in \Mod{M}$.  Here, $\parityof{m} \in \set{0,1}$ denotes the parity of $m$.  It is now easy to check that these natural isomorphisms, along with the identity, satisfy the associativity constraints required to give $\wtmod{\myosp}$ a Weyl group action.  However, this action is not essential for much of the analysis to follow because we will be chiefly concerned with identifying weight modules up to isomorphism.

\section{The affine Kac-Moody superalgebra $\myaosp$} \label{sec:affosp}

We now turn to the affinisation $\myaosp$ in its \ns{} and Ramond guises as well as the associated \vosas{}.  Verma modules and their generalisations, the relaxed Verma modules, are introduced along with their simple quotients.  The conjugation and spectral flow automorphisms are used to twist the latter and thereby construct a large collection of simple smooth weight modules over $\myaosp$, almost none of which are positive-energy. We recall that a module being smooth means that for all $j \in \myosp$ and $v$ in the module, $j_m \cdot v$ vanishes for $m$ sufficiently large (see below for notation).

\subsection{The affine algebra} \label{subsec:affalg}

The affine Kac-Moody superalgebra $\myaosp$ may be defined, as a vector space, by choosing a basis.  The standard choice is
\begin{equation} \label{eqn:affbasis}
\set{h_m, e_m, f_m \st m \in \ZZ} \cup \set{x_m, y_m \st m \in \ZZ + \xi} \cup \set{K,L_0},
\end{equation}
where $\xi$ is either $0$ or $\frac{1}{2}$, giving what we will call the \ns{} and Ramond $\myaosp$ algebras, $\aalgg_{\NS}$ and $\aalgg_{\R}$, respectively. As we will demonstrate later, these two choices of indexing give isomorphic algebras, hence we will generally suppress the subscripts and just write $\aalgg = \myaosp$. However, many of the representation-theoretic constructions that we shall consider depend on this choice. In particular, the representation theory splits into two sectors according to which algebra, \ns{} or Ramond, is acting on the module.  The modules on which $\aalgg_{\NS}$ acts constitute the \emph{\ns{} sector} and the $\aalgg_{\R}$-modules constitute the \emph{Ramond sector}.

In both cases, the even subalgebra is defined to be spanned by the $h_m$, $e_m$ and $f_m$, as well as $K$ and $L_0$.  The $x_m$ and $y_m$ are declared to be odd. Letting $j, j' \in \set{h,e,f,x,y}$ denote arbitrary basis vectors of $\myosp$, the (anti)commutation relations of the affine basis vectors \eqref{eqn:affbasis} take the form
\begin{equation} \label{eqn:affcomm}
	\begin{aligned}
		\comm{j_m}{j'_n} &= \comm{j}{j'}_{m+n} + m \killing{j}{j'} \delta_{m+n,0} K & &\text{if \(j\) or \(j'\) is even}, \\
		\acomm{j_m}{j'_n} &= \acomm{j}{j'}_{m+n} + m \killing{j}{j'} \delta_{m+n,0} K & &\text{if \(j\) and \(j'\) are odd},
	\end{aligned}
	\qquad \comm{L_0}{j'_n} = -n j'_n
\end{equation}
and $K$ is central. We recall that the (anti)commutators of the basis elements of $\myosp$ were given in \eqref{eqn:slcomm} and \eqref{eqn:ospcomm}, while the non-zero values taken by the (normalised) Killing form were listed in \eqref{eqn:slkill} and \eqref{eqn:ospkill}.

Note that the \ns{} algebra $\aalgg_{\NS}$ has a finite-dimensional subalgebra spanned by $\set{h_0,e_0,f_0,x_0,y_0}$, isomorphic to $\myosp$. This is the \emph{horizontal subalgebra} of $\aalgg_{\NS}$. The inclusion allows us to carry much of the representation-theoretic data we have for $\myosp$ over to $\myaosp$. The horizontal subalgebra of the Ramond algebra $\aalgg_{\R}$ is defined instead to be spanned by $\set{h_0,e_0,f_0}$ (because odd elements do not have zero modes in $\aalgg_{\R}$), hence it is isomorphic to $\sltwo$.  Our study of the Ramond sector will therefore be closely related to $\sltwo$ representation theory.

The Cartan subalgebra $\aalgh$ of both $\aalgg_{\NS}$ and $\aalgg_{\R}$ is defined to be the abelian subalgebra spanned by $h_0$, $K$ and $L_0$. With respect to $\aalgh$, these algebras have root systems
\begin{equation} \label{eqn:affroots}
\set{\pm 2 \alpha + n \delta \st n \in \ZZ} \cup \set{\pm \alpha + n \delta \st n \in \ZZ + \xi} \cup \set{n \delta \st n \in \ZZ_{\neq 0}},
\end{equation}
where the roots $\alpha,\delta \in \aalgh^*$ are defined by
\begin{equation} \label{eqn:alphadelta}
\begin{aligned}
\alpha(h_0)&=1, &&& \alpha(K)&=\alpha(L_0)=0,\\
\delta(h_0)&=\delta(K)=0, &&& \delta(L_0)&=-1.
\end{aligned}
\end{equation}
The positive roots are taken to be those in \eqref{eqn:affroots} with $n>0$, along with $\alpha$ and $2 \alpha$.
The simple roots are then $\set{\alpha, -2\alpha+\delta}$ in the \ns{} case and $\set{2\alpha, -\alpha+\frac{1}{2} \delta}$ in the Ramond case. These choices of simple roots give rise to identical generating reflections for the Weyl group.

\subsection{Generalised Verma modules and vertex operator superalgebras} \label{subsec:verma} \label{subsec:voa}

Given the choice of positive roots above, we obtain \emph{triangular decompositions} for $\aalgg = \aalgg_{\NS}$ and $\aalgg_{\R}$:
\begin{equation} \label{eqn:tridecomp}
\aalgg = \aalgg^- \oplus \; \aalgh \; \oplus \; \aalgg^+.
\end{equation}
Here, $\aalgg^+$ and $\aalgg^-$ denote the subalgebras of $\aalgg$ spanned by the positive and negative root vectors, respectively. For example, $\aalgg^+_{\NS}$ is spanned by $x_0$, $e_0$ and all the $j_m$, $j=e,x,h,y,f$, with $m \in \ZZ_{>0}$. Associated with each decomposition is the \emph{Borel subalgebra} $\aalgb = \aalgh \oplus \aalgg^+$.

We define a weight space of an $\myaosp$-module to be a simultaneous eigenspace of $h_0$ and $K$ that is also a generalised eigenspace of $L_0$.  We then define a weight module over $\myaosp$ to be a $\ZZ_2$-graded module that decomposes (as a vector space) into a direct sum of finite-dimensional weight spaces.  Note that although $L_0$ is permitted to act non-semisimply on a weight module, its Jordan blocks will have finite rank.

Consider now a one-dimensional $\aalgb$-module spanned by some $v \neq 0$ on which $\aalgg^+$ acts trivially and $\aalgh$ acts via
\begin{equation}
	h_0 v = \lambda v, \quad K v = k v, \quad L_0 v = \Delta v,
\end{equation}
for some $k,\lambda,\Delta \in \CC$. We call $k$ the \emph{level}, $\lambda \alpha$ the \emph{$\myosp$-weight}, and $\Delta$ the \emph{conformal weight} of $v$. We then promote this $\aalgb$-module to a $\aalgg$-module via induction:
\begin{equation} \label{eqn:vermainduc}
\Mod{V}^{\NS/\R}_{k, \lambda} = \operatorname{Ind}_{\aalgb}^{\aalgg} \CC v,
\end{equation}
which we call a \emph{Verma module} of (\ns{} or Ramond, as indicated by the superscript) $\myaosp$. Note, in particular, that $L_0$ acts diagonalisably on both modules. Moreover, they are clearly \hwms{}, generated by the \hwv{} $v$, whose weight spaces all have finite dimension.  As usual, any \hw{} $\aalgg$-module can be written as a quotient of the Verma module of the same highest weight.

Let $\vac{k}$ denote the \hwv{} generating the \ns{} Verma module $\NSVer{k,0}$.  The vector $y_0 \vac{k}$ is always singular in this module and it generates a proper Verma submodule.  If the level $k$ is non-critical, meaning that $k \neq -\dcox = -\frac{3}{2}$, quotienting by this submodule gives a \hwm{} that carries the structure of a \vosa{}, called the level $k$ \emph{universal} \vosa{} of $\myaosp$ (the conformal structure will be given in \eqref{eq:Sugawara} below).  We shall denote it by $\OSPUniv{k}$.  For generic values of $k$, this proper submodule is the unique maximal submodule, hence $\OSPUniv{k}$ is simple as a \vosa{}.

The algebraic structure of the universal \vosa{} $\OSPUniv{k}$ is completely determined by the \opes{} of the generating fields $h(z)$, $e(z)$, $f(z)$, $x(z)$ and $y(z)$:
\begin{equation}
	j(z) = \sum_{n \in \ZZ+\xi} j_n z^{-n-1} \quad \text{(\(j = h, e, f, x, y\)).}
\end{equation}
The \opes{} themselves have the form
\begin{equation}
	j(z) j'(w) \sim \frac{\killing{j}{j'} k}{(z-w)^2} + \frac{\comm{j}{j'}(w)}{z-w} \quad \text{(\(j,j' = h, e, f, x, y\))}
\end{equation}
and these are equivalent to the (anti)commutation relations \eqref{eqn:affcomm}.  We remark that the \vosa{} $\OSPUniv{k}$ is universal in the sense that any \vosa{} whose fields are normally ordered products of derivatives of generating fields satisfying these \opes{} is a quotient of $\OSPUniv{k}$.  This follows immediately from the universality of Verma modules and the state-field correspondence of vertex algebras.

The conformal structure of the \vosa{} $\OSPUniv{k}$ is defined by the Sugawara construction.  Explicitly, the energy-momentum tensor is
\begin{equation} \label{eq:Sugawara}
	T(z) = \sum_{n \in \ZZ} L_n z^{-n-2} = \frac{1}{2t} \sqbrac*{\frac{1}{2} \normord{h(z)h(z)} + \normord{e(z)f(z)} + \normord{f(z)e(z)} - \frac{1}{2} \normord{x(z)y(z)} + \frac{1}{2} \normord{y(z)x(z)}},
\end{equation}
where $t = k + \dcox = k + \frac{3}{2}$, and the generating fields are all weight $1$ conformal primaries with respect to this structure.  The Virasoro modes $L_n$ are thus expressed as infinite sums of normally ordered products of modes in (an appropriate completion of) the \uea{} $\EnvAlg{\aalgg}$. Note that we need not specify a completion, as for smooth modules the action of each infinite sum truncates to a finite one.

\begin{prop}[The Sugawara Construction] \label{thm:sugawara}
In any smooth representation of $\myaosp$ on which $K$ acts as multiplication by $k \in \CC \setminus \set{-\frac{3}{2}}$, the operators
\begin{equation}
L_m = \frac{1}{2t} \brac*{ \; \frac{1}{2}\normord{hh}_m + \normord{ef}_m + \normord{fe}_m - \frac{1}{2}\normord{xy}_m + \frac{1}{2}\normord{yx}_m}
\end{equation}
furnish a representation of the Virasoro algebra of central charge
\begin{equation}
	c=1-\frac{3}{2t}=\frac{k}{t} \quad \text{(\(t = k + \frac{3}{2}\)).}
\end{equation}
\end{prop}

For the bosonic modes, normally-ordered products are defined by the usual formula
\begin{subequations} \label{eqn:normord}
	\begin{equation} \label{eqn:bosnormord}
		\normord{A B}_n = \sum_{m \le -1} A_m B_{n-m} + \sum_{m > -1} B_{n-m} A_m \quad \text{(\(A,B = h, e, f\)).}
	\end{equation}
	However, for the fermionic fields $x(z)$ and $y(z)$, there is some subtlety in defining normal-ordering, depending on whether we are considering the \ns{} or Ramond sector of the \cft{}. The definition follows from considering the \emph{generalised commutation relations}:
	\begin{equation} \label{eqn:fermgcr}
		\begin{aligned}
			\sum_{r<0} x_{a+r} y_{b-r} - \sum_{r \ge 0} y_{b-r} x_{a+r} &= +a(a-1) \delta_{a+b,0} K + a h_{a+b} + \normord{x y}_{a+b}, \\
			\sum_{r<0} y_{a+r} x_{b-r} - \sum_{r \ge 0} x_{b-r} y_{a+r} &= -a(a-1) \delta_{a+b,0} K + a h_{a+b} + \normord{y x}_{a+b},
		\end{aligned}
	\end{equation}
\end{subequations}
which hold for all $a,b \in \ZZ + \xi$.  In both sectors, the normally-ordered fields $\normord{xy}$ and $\normord{yx}$ have integer-indexed modes (they are bosonic).  Of course, there are also many other normally-ordered products involving $x$ or $y$; for the calculations below those given here are sufficient.

Note that the Virasoro mode $L_0$ obtained from the Sugawara construction obeys the same commutation relations as the $\myaosp$ basis element of the same symbol. It is standard to identify these by restricting attention to modules on which they act as the same endomorphism.  Depending on whether we are considering Verma modules over the \ns{} or Ramond $\myaosp$ algebra, this identification leads to
\begin{equation} \label{eqn:cwhw}
\Delta = \frac{\lambda \brac{\lambda+1}}{4t} \quad \text{or} \quad \Delta = \frac{1}{2t} \brac*{\frac{\lambda \brac{\lambda+2}}{2} - \frac{k}{4}},
\end{equation}
respectively.

Of course, there are levels for which the universal \vosa{} $\OSPUniv{k}$ is not simple.  By studying embedding diagrams, or Shapovalov-type forms on $\OSPUniv{k}$, one deduces that this happens precisely for the (non-critical) levels satisfying the following condition \cite{GorSim07}:
\begin{equation} \label{eq:DefUV}
	2t=2\brac*{k+\frac{3}{2}} = \frac{u}{v} \quad \text{(\(u \in \ZZ_{\ge 2}\), \(v \in \ZZ_{\ge 1}\), \(u-v \in 2\ZZ\) and \(\gcd\set*{u,\frac{u-v}{2}}=1\)).}
\end{equation}
For these levels, called the \emph{admissible} levels, $\OSPUniv{k}$ has a maximal proper ideal generated by a single \sv{} $\chi_k$.  The simple quotient \vosa{} is the level $k$ \emph{minimal model} of $\myaosp$, which we denote by $\OSPMinMod{u}{v}$.

The construction of affine Verma modules via induction of $\aalgb$-modules generalises so that one can induce from an arbitrary module over the horizontal subalgebra.  For this, we replace \eqref{eqn:tridecomp} by
\begin{equation}
	\aalgg = \aalgn \oplus \aalgz \oplus \aalgp,
\end{equation}
where $\aalgn$ and $\aalgp$ are the subalgebras spanned by the modes with negative and positive indices, respectively, and $\aalgz$ is the subalgebra spanned by $K$ and the modes with index $0$ (the \emph{zero modes}).  Any module over the horizontal subalgebra may be extended to a $\aalgz$-module by requiring that $K$ act as $k$ times the identity, then to a $\aalgz \oplus \aalgp$-module by letting $\aalgp$ act as $0$.  If the module for the horizontal subalgebra is simple, then the result of inducing this $\aalgz \oplus \aalgp$-module to a $\aalgg$-module is called a \emph{generalised Verma module}.

When the simple module is a Verma module for the horizontal subalgebra, then the result of the induction described above is just a Verma module for $\aalgg$.  When the simple module is the trivial module, then the generalised Verma module may be identified as the $\aalgg$-module underlying the universal \vosa{} $\OSPUniv{k}$.  However, there are many other possibilities for the initial simple module (see \cref{thm:ospclass}).  We remark that generalised Verma modules are examples of \emph{relaxed \hwms{}} \cite{FeiEqu98}, these being modules generated by a single weight vector, called a \emph{relaxed \hwv{}}, that is annihilated by $\aalgp$.  They are also examples of \emph{positive-energy} modules, these being weight modules for which the conformal weights are bounded from below.

In contrast to \hwms{}, the conformal weight \(\Delta\) of a relaxed \hwv{} is not necessarily determined by the weight but rather by the eigenvalue \(s\) of the super-Casimir $\Sigma$, in the \ns{} sector, and by the eigenvalue \(q\) of the \(\sltwo\) Casimir $Q$, in the Ramond sector. The respective formulae are
\begin{align} \label{eq:RelConfWts}
	\Delta=\frac{s^2-1/4}{4t}\quad \text{and} \quad	\Delta=\frac{q-k/4}{2t}.
\end{align}

\subsection{Simple weight modules} \label{subsec:irrwtmods}

In what follows, we shall be chiefly interested, not in these generalised Verma modules over $\myaosp$, but rather in their simple quotients.  Our notation for these follows that used for the simple modules of $\sltwo$ and $\myosp$ in \cref{thm:slclass,thm:ospclass}.  More specifically, the simple quotients of the level $k$ \ns{} generalised Verma modules induced from the simple $\myosp$-modules $\Fin{\lambda}$, $\Inf{\lambda}^{\pm}$ and $\Rel{\Lambda, s}$ will be denoted by $\NSFin{\lambda}$, $\NSInf{\lambda}^{\pm}$ and $\NSRel{\Lambda, s}$, respectively.  Similarly, the simple quotients of the level $k$ Ramond generalised Verma modules induced from the simple $\sltwo$-modules $\slFin{\lambda}$, $\slInf{\lambda}^{\pm}$ and $\slRel{\Lambda, q}$ will be denoted by $\RFin{\lambda}$, $\RInf{\lambda}^{\pm}$ and $\RRel{\Lambda, q}$, respectively.  In all cases, the level dependence will be implicit.

We shall also need to consider quotients of the \ns{} $\myaosp$-modules that are induced from the reducible, but indecomposable, $\myosp$-modules $\Rel{\Lambda, s}^{\pm}$ and the Ramond $\myaosp$-modules that are induced from the reducible, but indecomposable, $\sltwo$-modules $\slRel{\Lambda, q}^{\pm}$.  Specifically, we want the quotient by the (unique) maximal submodule whose intersection with the subspace of vectors of minimal conformal weight is zero.  We denote these quotients by $\NSRel{\Lambda, s}^{\pm}$, in the \ns{} sector, and by $\RRel{\Lambda, q}^{\pm}$, in the Ramond sector.  Their structures are determined, up to isomorphism, by the following non-split short exact sequences (see \eqref{eqn:sese} and \eqref{eqn:sesc}):\footnote{
	This is actually non-trivial to prove rigorously and we shall not do so here, instead referring to \cite{KawRel18} for the details.
}
\begin{equation} \label{ses:CE}
	\begin{aligned}
		&\dses{\parrmod{\NSInf{\lambda - 1}^+}}{}{\NSRel{\lambda + 2 \ZZ, -\lambda + 1/2}^+}{}{\NSInf{\lambda}^-}, &&&
		&\dses{\RInf{\lambda-2}^+}{}{\RRel{\lambda + 2 \ZZ, \lambda (\lambda - 2) / 2}^+}{}{\RInf{\lambda}^-}, \\
		&\dses{\parrmod{\NSInf{\lambda + 1}^-}}{}{\NSRel{\lambda + 2 \ZZ, +\lambda + 1/2}^-}{}{\NSInf{\lambda}^+}, &&&
		&\dses{\RInf{\lambda+2}^-}{}{\RRel{\lambda + 2 \ZZ, \lambda (\lambda + 2) / 2}^-}{}{\RInf{\lambda}^+}.
	\end{aligned}
\end{equation}

In \cref{subsec:simpchars}, we will compute the characters of (some of) these quotients.  We shall focus initially on genuine Verma modules, for which this computation may be performed by constructing a resolution, in the sense of Bern\v{s}te\u{\i}n-Gel'fand-Gel'fand \cite{BerDif75}, of the simple quotient in terms of (direct sums of) Verma modules.  Such a resolution was constructed for the simple quotients of \ns{} Verma modules over $\myaosp$ by Iohara and Koga in \cite{IohFus01a}.  For the simple quotient of a \ns{} Verma module $\NSVer{}$, each Verma module that appears in the resolution is generated by a \sv{} of $\NSVer{}$.  The fact that a \sv{} generates a Verma submodule of $\NSVer{}$ follows from the fact that the \uea{} of $\myaosp$ has no zero divisors \cite{AubZer85}.

The computation of the character of the simple quotient of a \ns{} Verma module $\NSVer{}$ therefore reduces to the identification of its \emph{\svs}. This is achieved by means of the \emph{Shapovalov form}: an invariant bilinear form $\brac*{\cdot,\cdot}$ on $\NSVer{}$. We will not go here into the details of how it is defined, instead being content to simply state and use some of its properties. If $u,w \in \NSVer{}$ are weight vectors of distinct weights, then $\brac*{u,w}$ vanishes, so we may separately consider the restriction of this to a form $F_{\mu}$ on each weight space $\NSVer{} \brac*{\mu}$, $\mu \in \aalgh^*$. For \svs{} $u \in \NSVer{} \brac*{\mu}$ and their descendant vectors, the linear functional $F_{\mu} \brac*{u,\cdot}$ identically vanishes.

As we will show in \cref{subsec:aut}, it is possible to relate Ramond Verma modules to \ns{} ones in such a way that all information concerning their submodule structure carries over. We therefore need only look for \svs{} in the \ns{} Verma modules, aided by the following formula.
\begin{thm}[Kac-Kazhdan Determinant Formula \protect{\cite{KacStr79,BowRep97}}] \label{thm:kkdet}
Let $\Lambda \in \aalgh^*$ denote the highest weight of the \ns{} Verma module $\NSVer{k,\lambda}$ over $\myaosp$. Then, for any non-negative integer linear combination $\eta$ of simple roots, the determinant of the Shapovalov form restricted to the weight space $\NSVer{k,\lambda} \brac*{\Lambda - \eta}$ is given by
\begin{multline} \label{eq:KK}
\det\brac*{F_{\Lambda - \eta}} = \prod_{\ell=1}^\infty \set[\Bigg]{
\brac*{ \frac{\lambda + 1 - \ell}{2}}^{P\brac{ \eta - (2 \ell -1) \alpha }}
\prod_{n=1}^\infty \sqbrac[\bigg]{
\brac*{t(2n-1)}^{P\brac{\eta - \ell (2n-1)\delta}} \\
\cdot\brac*{\frac{\lambda + 1 +2nt - \ell}{2}}^{P\brac{\eta - (2\ell -1)(\alpha + n\delta) }}
\brac*{\frac{-\lambda +2nt - \ell}{2}}^{P\brac{\eta - (2\ell -1)(-\alpha + n\delta)}} \\
\cdot\brac*{\lambda + \frac{1}{2} + t(2n-1) - \ell}^{P\brac{\eta - \ell(2 \alpha + (2n-1)\delta)}}
\brac*{-\lambda - \frac{1}{2} + t(2n-1) - \ell}^{P\brac{\eta - \ell(-2 \alpha + (2n-1)\delta)}}}},
\end{multline}
where $P(\mu)$ is the number of ways that $\mu \in \aalgh^*$ can be written as a linear combination $\mu = \sum n_\beta \beta$ of positive roots, with coefficients $n_\beta \in \set{0,1}$ if $2 \beta$ is itself a positive root and coefficients $n_\beta \in \ZZ_{\ge 0}$ otherwise.
\end{thm}

When $\det\brac*{F_{\Lambda - \eta}} = 0$, for some non-negative integer linear combination of positive roots $\eta$, the weight space $\NSVer{k,\lambda} \brac*{\Lambda - \eta}$ contains an element of some proper submodule of $\NSVer{k,\lambda}$.  If the vanishing factor of \eqref{eq:KK} has an exponent where the argument of $P$ also vanishes, then the weight space contains a \sv{}.  Unfortunately, we do not obtain all \svs{} of $\NSVer{k,\lambda}$ from such vanishing factors and instead obtain the rest by analysing the Kac-Kazhdan determinant where the highest weight $\Lambda$ is that of one of the \svs{} that we have already identified (here it is important that the \uea{} of $\myaosp$ has no zero divisors).  We illustrate this determination of \svs{} by depicting the corresponding submodule inclusions (also known as embedding diagrams) of two $k=-\frac{5}{4}$ Verma modules in \cref{fig:nsverma}.

Unfortunately, an analogous determinant formula for generalised \hwms{} over $\myaosp$ does not appear to have been considered in the literature, though there is some work pertaining to $\AKMA{sl}{2}$, for instance \cite{KhoDet99}.  In the absence of structural data, we shall have to resort to indirect means to compute the characters of the simple quotients of the generalised Verma modules (see \cref{subsec:relchars}).

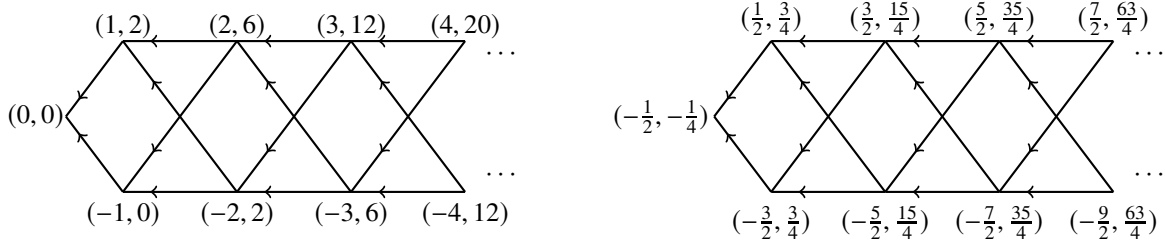
\begin{figure}
\begin{tikzpicture}[baseline=(c.center)]
\draw [-<-] (0.75,0)--(1.5,1);
\draw [-<-] (0.75,0)--(1.5,-1);
\draw [-<-] (1.5,1)--(3,1);
\draw [-<-] (1.5,1)--(3,-1);
\draw [-<-] (1.5,-1)--(3,-1);
\draw [-<-] (1.5,-1)--(3,1);
\draw [-<-] (3,1)--(4.5,1);
\draw [-<-] (3,1)--(4.5,-1);
\draw [-<-] (3,-1)--(4.5,-1);
\draw [-<-] (3,-1)--(4.5,1);
\draw [-<-] (4.5,1)--(6,1);
\draw [-<-] (4.5,1)--(6,-1);
\draw [-<-] (4.5,-1)--(6,-1);
\draw [-<-] (4.5,-1)--(6,1);
\node (c) at (.35,0) {$(0,0)$};
\node at (1.5,1.2) {$(1,2)$};
\node at (1.5,-1.3) {$(-1,0)$};
\node at (3,1.2) {$(2,6)$};
\node at (3,-1.3) {$(-2,2)$};
\node at (4.5,1.2) {$(3,12)$};
\node at (4.5,-1.3) {$(-3,6)$};
\node at (6,1.2) {$(4,20)$};
\node at (6,-1.3) {$(-4,12)$};
\node at (6.5,.8) {$\cdots$};
\node at (6.5,-.8) {$\cdots$};
\end{tikzpicture}
\hspace{0.05\textwidth}
\begin{tikzpicture}[baseline=(c.center)]
\draw [-<-] (0.75,0)--(1.5,1);
\draw [-<-] (0.75,0)--(1.5,-1);
\draw [-<-] (1.5,1)--(3,1);
\draw [-<-] (1.5,1)--(3,-1);
\draw [-<-] (1.5,-1)--(3,-1);
\draw [-<-] (1.5,-1)--(3,1);
\draw [-<-] (3,1)--(4.5,1);
\draw [-<-] (3,1)--(4.5,-1);
\draw [-<-] (3,-1)--(4.5,-1);
\draw [-<-] (3,-1)--(4.5,1);
\draw [-<-] (4.5,1)--(6,1);
\draw [-<-] (4.5,1)--(6,-1);
\draw [-<-] (4.5,-1)--(6,-1);
\draw [-<-] (4.5,-1)--(6,1);
\node (c) at (.05,0) {$(-\frac{1}{2},-\frac{1}{4})$};
\node at (1.5,1.3) {$(\frac{1}{2},\frac{3}{4})$};
\node at (1.5,-1.4) {$(-\frac{3}{2},\frac{3}{4})$};
\node at (3,1.3) {$(\frac{3}{2},\frac{15}{4})$};
\node at (3,-1.4) {$(-\frac{5}{2},\frac{15}{4})$};
\node at (4.5,1.3) {$(\frac{5}{2},\frac{35}{4})$};
\node at (4.5,-1.4) {$(-\frac{7}{2},\frac{35}{4})$};
\node at (6,1.3) {$(\frac{7}{2},\frac{63}{4})$};
\node at (6,-1.4) {$(-\frac{9}{2},\frac{63}{4})$};
\node at (6.5,.8) {$\cdots$};
\node at (6.5,-.8) {$\cdots$};
\end{tikzpicture}
\caption{Submodule inclusions in the $k=-\frac{5}{4}$ Verma modules $\NSVer{-5/4,0}$ (left) and $\NSVer{-5/4,-1/2}$ (right).  Each vertex indicates a \sv{} generating a Verma submodule.  The pairs $(\lambda,\Delta)$ attached to each vertex give the $\myosp$-weight $\lambda$ and conformal weight $\Delta$ of the corresponding \sv{}.} \label{fig:nsverma}
\end{figure}

\subsection{Automorphisms} \label{subsec:aut}

The Weyl group of an affine Kac-Moody superalgebra, defined as the subgroup of $\SLG{GL}{\aalgh}$ generated by the reflections about the hyperplanes orthogonal to the even roots, has a relatively simple structure which can be understood in terms of data arising from its horizontal subalgebra, despite having infinite order. Abstractly, if $\grp{W}$ is the Weyl group of the horizontal subalgebra $\algg$, then the Weyl group $\affine{\grp{W}}$ of the full affine superalgebra $\aalgg$ decomposes as
\begin{equation} \label{eqn:weyldecomp}
\affine{\grp{W}} \cong \grp{W} \ltimes \grp{Q}_0^\vee,
\end{equation}
where $\grp{Q}_0^\vee$ is the coroot lattice of the even subalgebra $\algg^{\even}$: the $\ZZ$-span of the coroots $\alpha_j^\vee$ for any choice of simple roots $\set{\alpha_j}$ of $\algg^{\even}$. The semidirect product structure is defined by the standard linear action of $\grp{W}$ on this lattice, as a subset of $\algh$.

In the case of $\myosp$, the (lone) simple coroot of the even subalgebra is $(2\alpha)^\vee$, giving $\grp{Q}_0^\vee \cong \ZZ$, with the Weyl group $\grp{W} \cong \ZZ_2$ acting by inversion on this lattice.  The Weyl group $\affine{\grp{W}}$ of $\myaosp$ is thus isomorphic to the non-trivial semidirect product of $\ZZ_2$ with $\ZZ$: the infinite dihedral group.

For both the \ns{} and Ramond algebras, we define the \emph{conjugation} automorphism $\conjaut$ by
\begin{subequations}
	\begin{align}
		&
		\begin{aligned}
			\conjmod{e_n} &= -f_n, \\
			\conjmod{x_n} &= -y_n,
		\end{aligned}
		&&
		\begin{aligned}
			\conjmod{h_n} &= -h_n, \\
			\conjmod{K} &= K, \\
			\conjmod{L_0} &= L_0,
		\end{aligned}
		&&
		\begin{aligned}
			\conjmod{f_n} &= -e_n, \\
			\conjmod{y_n} &= x_n,
		\end{aligned}
		\label{eqn:affconjaut}
  \intertext{where we note that the restriction to the horizontal subalgebra gives precisely the conjugation automorphism $\finconjaut$ of $\myosp$, as defined in \cref{subsec:weylgrp}.  Similarly, the \emph{spectral flow} automorphism $\sfaut$ is defined by}
    &
		\begin{aligned}
			\sfmod{}{e_n} &= e_{n-2}, \\
			\sfmod{}{x_n} &= x_{n-1},
		\end{aligned}
		&&
		\begin{aligned}
			\sfmod{}{h_n} &= h_{n} - 2\delta_{n,0} K, \\
			\sfmod{}{K} &= K, \\
			\sfmod{}{L_0} &= L_0 - h_0 + K,
		\end{aligned}
		&&
		\begin{aligned}
			\sfmod{}{f_n} &= f_{n+2}, \\
			\sfmod{}{y_n} &= y_{n+1}.
		\end{aligned}
		\label{eqn:sfaut}
	\end{align}
	When we need to distinguish the algebra on which the automorphism is acting, we shall furnish it with a subscript, as in $\conjaut_{\NS}$ or $\conjaut_{\R}$.  In addition, we define isomorphisms $\tau \colon \aalgg_{\NS} \to \aalgg_{\R}$ and $\tau' \colon \aalgg_{\R} \to \aalgg_{\NS}$ between the \ns{} and Ramond algebras, both of which act on the basis elements according to
	\begin{align}\label{eqn:halfsf}
		&
		\begin{aligned}
			e_n &\longmapsto e_{n-1}, \\
			x_n &\longmapsto x_{n-1/2},
		\end{aligned}
		&&
		\begin{aligned}
			h_n &\longmapsto h_{n} - \delta_{n,0} K, \\
			K &\longmapsto K, \\
			L_0 &\longmapsto L_0 - \tfrac{1}{2} h_0 + \tfrac{1}{4} K,
		\end{aligned}
		&&
		\begin{aligned}
			f_n &\longmapsto f_{n+1}, \\
			y_n &\longmapsto y_{n+1/2}
		\end{aligned}
	\end{align}
\end{subequations}
(the indices $n$ are constrained to range over the appropriate domains for $\aalgg_{\NS}$ and $\aalgg_{\R}$). It is not hard to see that $\tau' \circ \tau = \sfaut_{\NS}$ and $\tau \circ \tau' = \sfaut_{\R}$. As such, we will from here on denote both $\tau$ and $\tau'$ by $\sfaut^{1/2}$, using a subscript to distinguish (where necessary) which algebra is being acted upon.

Note that the automorphisms $\conjaut$ and $\sfaut$ both preserve the Cartan subalgebra $\aalgh$ and that their restrictions to $\aalgh$ generate the Weyl group $\affine{\grp{W}} \subseteq \SLG{GL}{\aalgh}$. However, the automorphisms themselves satisfy $\conjaut^2 = \brac*{\conjaut \sfaut}^2 = \Sid$ and therefore generate a group isomorphic to $\ZZ_4 \ltimes \ZZ$. As with $\myosp$, we cannot realise the Weyl group $\affine{\grp{W}}$ as a subgroup of the inner automorphisms of $\myaosp$.

Nevertheless, just as we did in \cref{subsec:weylgrp}, we may promote both $\conjaut$ and $\sfaut$ to functors on the module category of either the \ns{} or Ramond algebra. In fact, since $\sfaut$ is of infinite order, we get a functor $\sfaut^n$ for each $n \in \ZZ$.  We shall refer to the images of a module $\Mod{M}$ under $\conjaut$ and $\sfaut^n$ as the conjugate and spectral flow of $\Mod{M}$, respectively.  As in \cref{subsec:weylgrp}, it is clear that these functors commute with parity reversal.  They moreover satisfy
\begin{equation} \label{eqn:twistrels}
\func{\conjaut \conjaut}{\Mod{M}} \cong \Mod{M}, \quad
\func{\sfaut^m \sfaut^n}{\Mod{M}} \cong \func{\sfaut^{m+n}}{\Mod{M}} \quad \text{and} \quad
\func{\conjaut \sfaut^n}{\Mod{M}} \cong \func{\sfaut^{-n} \conjaut}{\Mod{M}},
\end{equation}
for all $m,n \in \ZZ$, so we have indeed constructed a strict $\affine{\grp{W}}$-action on the set of $\myaosp$-module isomorphism classes.  As in \cref{subsec:weylgrp}, this action can be lifted to the category of weight $\myaosp$-modules.  We may also carry this out for the isomorphisms $\sfaut^{1/2}$, constructing functors between the two categories. These satisfy
\begin{equation} \label{eqn:halftwist}
\func{\sfaut^{1/2} \sfaut^{1/2}}{\Mod{M}} \cong \func{\sfaut}{\Mod{M}},
\end{equation}
so in fact we have spectral flow functors $\sfaut^n$ for each $n \in \frac{1}{2} \ZZ$.  The isomorphisms of \eqref{eqn:twistrels} also hold for $m,n \in \frac{1}{2} \ZZ$.

Because the functors $\conjaut$ and $\sfaut^n$, $n \in \frac{1}{2} \ZZ$, are obviously invertible, they define \emph{equivalences} between the corresponding module categories. As a consequence, they preserve structure (submodules, quotients, Loewy diagrams, and so on). These functors turn out to be incredibly useful in analysing the representation theory of $\myaosp$ and its associated \vosas{}.  As an example, we remark that the claim made in \cref{subsec:irrwtmods} --- that the structures of the Ramond Verma modules may be deduced from those of the \ns{} Verma modules --- now follows immediately from the identifications
\begin{equation} \label{eqn:rverma}
\RVer{k,\lambda} \cong \func{\conjaut \sfaut^{-1/2}}{\NSVer{k,k-\lambda}}.
\end{equation}
Proving \eqref{eqn:rverma} is straightforward but illustrative.  First, $\NSVer{k,k-\lambda}$ is generated by a \hwv{} $v$ of $\myosp$-weight $(k-\lambda) \alpha$, so $\func{\conjaut \sfaut^{-1/2}}{\NSVer{k,k-\lambda}}$ is generated by $\conjaut \sfaut^{-1/2}(v)$ (by the invertibility of the functors).  Second, the $\myosp$-weight of $\conjaut \sfaut^{-1/2}(v)$ is $\lambda \alpha$:
\begin{equation}
	h_0 \conjaut \sfaut^{-1/2}(v) = \conjaut \sfaut^{-1/2}(\sfaut^{1/2} \conjaut(h_0) v) = \conjaut \sfaut^{-1/2}\brac[\big]{(-h_0 + K) v} = \lambda \conjaut \sfaut^{-1/2}(v).
\end{equation}
Third, $\conjaut \sfaut^{-1/2}(v)$ is a \hwv{} (in the Ramond sector):
\begin{equation}
	\begin{aligned}
		e_0 \conjaut \sfaut^{-1/2}(v) &= \conjaut \sfaut^{-1/2}(\sfaut^{1/2} \conjaut(e_0) v) = \conjaut \sfaut^{-1/2}(-f_1 v) = 0, \\
		y_{1/2} \conjaut \sfaut^{-1/2}(v) &= \conjaut \sfaut^{-1/2}(\sfaut^{1/2} \conjaut(y_{1/2}) v) = \conjaut \sfaut^{-1/2}(x_0 v) = 0.
	\end{aligned}
\end{equation}
Finally, $\NSVer{k,k-\lambda}$ is freely generated as a $\EnvAlg{\aalgg^-_{\NS}}$-module so $\func{\conjaut \sfaut^{-1/2}}{\NSVer{k,k-\lambda}}$ is freely generated as a $\EnvAlg{\aalgg^-_{\R}}$-module ($\conjaut \sfaut^{-1/2}$ maps $\aalgg^-_{\NS}$ onto $\aalgg^-_{\R}$).  This completes the proof.

It is important to note that the conjugation and spectral flow automorphisms of \cref{subsec:aut} obviously extend to automorphisms of the \uea{} of $\myaosp$ and thereby define automorphisms of the (universal) \emph{vertex superalgebra} obtained from $\OSPUniv{k}$ by forgetting the conformal structure.  Twisting by conjugation or spectral flow therefore preserves the property of being a module of this vertex superalgebra.\footnote{
	We remark that our working definition of module of a vertex superalgebra follows that given by Frenkel and Ben-Zvi \cite[Ch.~5.1]{FreVer01}, adding the requirement to be $\ZZ_2$-graded by parity (as in \cref{subsec:strthy}).  In particular, a module over the level-$k$ universal vertex superalgebra associated to $\myosp$ is just a smooth $\ZZ_2$-graded level-$k$ $\myaosp$-module.
}  In fact, these twists also preserve the property of being a module of the \vosa{}.  The only difference is that the spectral flow images of certain modules may now be distinguished from the original modules because spectral flow does not preserve the conformal structure \eqref{eq:Sugawara}, hence the conformal weights will change in general.  In particular, this is the case for the \emph{vacuum module} which is the \vosa{} regarded as a module over itself.

Finally, note that because automorphisms necessarily preserve the maximal proper ideal of $\OSPUniv{k}$, regarded as a vertex superalgebra, the arguments of the previous paragraph also apply to the minimal model \vosas{} $\OSPMinMod{u}{v}$.  In particular, the category of $\OSPMinMod{u}{v}$-modules is closed under twisting by conjugation and spectral flow.
\begin{prop} \label{prop:AutVOAMod}
	\leavevmode
	\begin{enumerate}
		\item If \(\Mod{M}\) is an \(\OSPUniv{k}\) module, then its twists by
			spectral flow and conjugation, \(\sfmod{n}{\Mod{M}}\) and \(\conjmod{\Mod{M}}\)
			respectively, are also \(\OSPUniv{k}\) modules.
		\item If \(\Mod{M}\) is an \(\OSPMinMod{u}{v}\) module, then its twists by
			spectral flow and conjugation, \(\sfmod{n}{\Mod{M}}\) and \(\conjmod{\Mod{M}}\)
			respectively, are also \(\OSPMinMod{u}{v}\) modules.
	\end{enumerate}
\end{prop}

\section{The affine minimal model $\OSPMinMod{2}{4}$: Modules} \label{sec:minmod}

\Cref{prop:AutVOAMod} is noteworthy because the spectral flows of a given positive-energy module $\Mod{M}$ will not (usually) be positive-energy for almost all $n$.  It follows that the appropriate category of $\OSPMinMod{u}{v}$-modules for constructing a consistent \cft{} is not likely to be a subcategory of the category of positive-energy weight modules.  Consequently, much of the representation theory of these \vosas{} cannot be detected directly by Zhu's algebra.  Nevertheless, it appears \cite{CreMod12,CreLog13,CreMod13,RidBos14} that combining positive-energy classifications with spectral flow does lead to module categories that satisfy key consistency requirements, modular invariance for example.

We shall therefore proceed with the classification of simple positive-energy $\OSPMinMod{u}{v}$-modules using Zhu technology, specialising to $k=-\frac{5}{4}$ ($u=2$, $v=4$).  The central charge of this minimal model is $c=-5$.  We begin by reviewing the theory of (twisted) Zhu algebras, emphasising their realisations in terms of zero modes of \vosa{} elements.

\subsection{Zhu's algebra} \label{subsec:zhualg}

The associative algebras now known as Zhu algebras \cite{ZhuMod96} form an invaluable formalism for classifying positive-energy modules over \vosas{}. Essentially, the Zhu algebra of a \vosa{} is the associative unital algebra of zero modes (of all fields) restricted to only act on ground states, these being vectors that are annihilated by all positive field modes (we will make this precise below).  The relaxed \hwvs{} of \cref{subsec:verma} are salient examples of ground states. Zhu's formalism then guarantees that any module \(\overline{\Mod{M}}\) over Zhu's algebra can be induced to a \vosa{} module \(\Mod{M}\), containing \(\overline{\Mod{M}}\) in its space of ground states. Further, if \(\overline{\Mod{M}}\) is simple, then the space of ground states of the unique simple quotient of \(\Mod{M}\) is precisely \(\overline{\Mod{M}}\). So simple positive-energy modules are classified by first classifying simple modules over Zhu's algebra and then taking the simple quotients of their inductions.

The Zhu algebra $\zhu{\VOA{V}}$ for untwisted modules over a \vosa{} $\VOA{V}$ was first studied in \cite{Kacn1z94}, while the Zhu algebra $\tzhu{\VOA{V}}$ for modules twisted by a finite order automorphism $\parity$ was introduced in \cite{DonTwi98}. In the case at hand, the untwisted modules form the \ns{} sector and the Ramond sector corresponds to modules twisted by the parity automorphism.  We refer to \cite[App.~A]{BloSVir16} for an introduction to Zhu's algebras for general \vosas{}, in both the twisted and untwisted cases.  This introduction emphasises the fact that Zhu's algebra is nothing but an abstraction of the algebra of zero modes acting on ground states.

With this fact in mind, the most straightforward way to define the untwisted Zhu algebra for affine \vosas{} is as follows.  Let $\EnvAlgk{\myaosp}$ denote the quotient of $\EnvAlg{\myaosp}$ by the ideal generated by $K - k \wun$ and let $\EnvAlgk{\myaosp}_0$ be its conformal weight zero subalgebra (the centraliser of $L_0$ in $\EnvAlgk{\myaosp}$).  Then, there is a projection $\pi_0 \colon \EnvAlgk{\myaosp}_0 \to \EnvAlg{\myosp}$ whose kernel is spanned by the \pbw{} basis elements, ordered by increasing mode index, that involve modes with non-zero indices (we identify zero modes with elements of $\myosp$).  The untwisted Zhu algebra $\zhu{\OSPUniv{k}}$ is then the image of the map
\begin{equation} \label{eq:DefZhuMap}
	v \in \OSPUniv{k} \longmapsto [v] = \pi_0(v_0),
\end{equation}
where $v_0$ is the zero mode of the field\footnote{For $v$ of definite conformal weight $\Delta_v$, we assume a mode expansion for the corresponding field of the form $v(z) = \sum_n v_n z^{-n-\Delta_v}$.} corresponding to $v$ (and $\pi_0$ has been extended to an appropriate completion of $\EnvAlgk{\myaosp}_0$).  Zhu's associative product $\zstar$ is then given by
\begin{equation} \label{eq:defstar}
	[u] \zstar [v] = \pi_0(u_0 v_0).
\end{equation}
The equivalence of this definition with Zhu's original one, in the case of affine \vosas{}, follows from noting that $\pi_0$ merely implements the constraint that the zero modes act on ground states.  We refer to \cite[App.~A]{BloSVir16} for further information.

The twisted Zhu algebra $\tzhu{\OSPUniv{k}}$ and its associative product are defined in almost precisely the same manner using the same formulae.  We shall distinguish the twisted image of an element $v$ from its untwisted cousin $[v]$ by a superscript:  $[v]^\parity \in \tzhu{\OSPUniv{k}}$.  The main difference is that we must restrict the defining map to the bosonic orbifold of $\OSPUniv{k}$ (the subalgebra of even elements) because the odd elements have no zero mode when acting in the Ramond sector.

The following \lcnamecref{prop:UnivZhu} was proved in \cite{FreVer92} for affine Kac-Moody algebras and in \cite{Kacn1z94} for the untwisted Zhu algebras of affine Kac-Moody superalgebras.  We are not aware of a source that proves the twisted Zhu algebra result in the latter case though it is surely very well known.
\begin{prop} \label{prop:UnivZhu}
  The untwisted and twisted Zhu algebras of $\OSPUniv{k}$ are
  \begin{equation}
		\zhu{\OSPUniv{k}} = \EnvAlg{\SLSA{osp}{1}{2}}, \quad
		\tzhu{\OSPUniv{k}} = \EnvAlg{\SLA{sl}{2}}.
  \end{equation}
\end{prop}
\begin{proof}
	By construction, $\zhu{\OSPUniv{k}}$ lies in $\EnvAlg{\myosp}$.  However, it is easy to check that any monomial $j^1 \cdots j^n \in \EnvAlg{\myosp}$ can be realised, up to a sign coming from parities, as the image in $\zhu{\OSPUniv{k}}$ of an element of $\OSPUniv{k}$, namely $j^n_{-1} \cdots j^1_{-1} \vac{k}$ (here, $\vac{k}$ denotes the vacuum vector of $\OSPUniv{k}$).  This demonstrates equality as vector spaces.  To show equality as associative algebras, we only need show that the Zhu elements $j \equiv j_0 = \pi_0(j_0) = [j_{-1} \vac{k}]$, $j=h,e,f,x,y$, satisfy the same commutation rules with respect to $\zstar$ as they do in $\EnvAlg{\myosp}$.  This is, of course, exactly how $\zstar$ is defined:
	\begin{align}
		j^1 \zstar j^2 - (-1)^{\parityof{j^1} \parityof{j^2}} j^2 \zstar j^1 &= \pi_0(j^1_0 j^2_0) - (-1)^{\parityof{j^1} \parityof{j^2}} \pi_0(j^2_0 j^1_0)
		= \pi_0(\comm{j^1_0}{j^2_0}) = \pi_0(\comm{j^1}{j^2}_0) = \comm{j^1}{j^2}.
	\end{align}
	We recall that $\parityof{j} \in \set{0,1}$ denotes the parity of $j$.

	The argument identifying $\tzhu{\OSPUniv{k}}$ is practically identical once one has shown that it lies in $\EnvAlg{\sltwo}$.  All that we can conclude at present is that $\tzhu{\OSPUniv{k}}$ lies in the even subalgebra of $\EnvAlg{\myosp}$ (which strictly contains $\EnvAlg{\sltwo}$).  To show that it indeed lies in $\EnvAlg{\sltwo}$, consider an even element $u \in \OSPUniv{k}$.  Choosing a \pbw{} ordering in which odd modes appear to the right of even modes (and are then ordered by increasing index), $u$ is expressed as a linear combination of monomials $\cdots j^1_{-m-1} j^2_{-n-1} \vac{k}$, where either the modes appearing are all even or both $j^1_{-m-1}$ and $j^2_{-n-1}$ are odd with $m,n \in \ZZ_{\ge 0}$.

	Consider the image of each such monomial in the twisted Zhu algebra.  If $j^1_{-m-1}$ and $j^2_{-n-1}$ are both odd, then we apply the following identity (obtained by applying $\pi_0$ to \cite[Eq.~(A.2)]{BloSVir16} with $k=m+1$, $n=m-\frac{1}{2}$):
	\begin{equation}\label{eq:tzhuidsb}
		\sqbrac*{\cdots j^1_{-m-1} j^2_{-n-1} \vac{k}}^\parity = -\sum_{\ell=0}^{\infty} \binom{1/2}{\ell+1} \sqbrac*{\cdots j^1_{-m+\ell} j^2_{-n-1} \vac{k}}^\parity, \quad m \in \ZZ_{\ge 0}.
	\end{equation}
	Inductively, the index $r$ of $j^1$ can then be made non-negative at which point we note that
	\begin{equation}
		j^1_r j^2_{-n-1} \vac{k} = \acomm{j^1_r}{j^2_{-n-1}} \vac{k}, \quad r \ge 0,
	\end{equation}
	replaces the two odd modes by an even one (and perhaps a constant).  We conclude that a monomial with a pair of odd modes is equivalent, in $\tzhu{\OSPUniv{k}}$, to a linear combination of monomials in which these odd modes are replaced by an even mode.  By performing this replacement for all odd modes, it follows as before that $\tzhu{\OSPUniv{k}}$ lies in $\EnvAlg{\sltwo}$.  The rest of the argument is identical to the untwisted case.
\end{proof}

For a given admissible level \(k\), recall that \(\chi_{k}\) denotes the \sv{} that generates the (unique) maximal proper ideal of \(\OSPUniv{k}\) by which one quotients in order to obtain the minimal model \vosa{} \(\OSPMinMod{u}{v}\). By the Kac-Kazhdan formula of
\cref{thm:kkdet}, one can determine that the conformal weight and
\(\SLSA{osp}{1}{2}\)-weight of \(\chi_k\) are \(\frac{1}{2}(u-1)v\) and
\((u-1)\alpha\), respectively.
Define \(\phi_k=y_0^{u-1}\chi_k\) and note that this descendant of $\chi_k$ has \(\SLSA{osp}{1}{2}\)-weight $0$.
\begin{prop} \label{prop:SimpZhu}
  The untwisted and twisted Zhu algebras of $\OSPMinMod{u}{v}$ are
  \begin{equation}
  \zhu{\OSPMinMod{u}{v}}\cong \frac{\EnvAlg{\SLSA{osp}{1}{2}}}{\ideal{\left[\phi_k\right]}}, \quad
  \tzhu{\OSPMinMod{u}{v}}\cong \frac{\EnvAlg{\SLA{sl}{2}}}{\ideal{\left[\phi_k\right]^\parity}},
  \end{equation}
  where \(\ideal{\left[\phi_k\right]}\) and \(\ideal{\left[\phi_k\right]^\parity}\) are the two-sided ideals generated by the images of \(\phi_k\) in \(\zhu{\OSPUniv{k}}\cong \EnvAlg{\SLSA{osp}{1}{2}}\) and \(\tzhu{\OSPUniv{k}}\cong \EnvAlg{\SLA{sl}{2}}\), respectively.
\end{prop}
\begin{proof}
	Note that while \(\phi_k\) is a zero mode descendant of \(\chi_k\), the
  converse is also true: \(\chi_k\) is a zero mode descendant of
  \(\phi_k\). This follows either by explicitly evaluating \(x_0^{u-1}\phi_k\)
  or by noting that finite-dimensional \(\myosp\) modules are semisimple and
  that the space of vectors in \(\OSPUniv{k}\) of conformal weight
  \(\frac{1}{2}(u-1)v\) is finite-dimensional. Thus, every vector in the maximal
  proper ideal of \(\OSPUniv{k}\) is a descendant of \(\phi_k\) by
  non-positive modes.

  We have to show that the image of the maximal proper ideal in Zhu's algebra is generated by the image of $\phi_k$.  To this end, let \(j\) and $v$ be homogeneous vectors in \(\OSPUniv{k}\) and suppose that $j$ has conformal weight $1$.  In the \ns{} sector, the corresponding fields \(j(z)\) and \(v(z)\) will have zero modes, regardless of their parities, and we have the following identities in \(\zhu{\OSPUniv{k}}\):
	\begin{equation}\label{eq:zhuids}
		[j_0 v] = [j] \zstar [v] -(-1)^{\parityof{j} \parityof{v}} [v] \zstar [j], \qquad
		[j_{-m-1} v] = (-1)^m (-1)^{\parityof{j} \parityof{v}} [v] \zstar [j], \quad m \ge 0.
	\end{equation}
	The first follows from $\comm{j_0}{v_0} = (j_0 v)_0$, while the second follows from $(j_{-m-1}v)(z) = \frac{1}{m!} \normord{\pd^m j(w) v(w)}$.  It follows inductively that if the image of $v$ in $\zhu{\OSPUniv{k}}$ belongs to the image of the maximal proper ideal, then so do the images of all the descendants of $v$ by non-positive modes.  Taking $v = \phi_k$ establishes the untwisted result.

	In the Ramond sector, \(j(z)\) and \(v(z)\) will only have zero modes if they have even parity. In this case, the identities \eqref{eq:zhuids} also hold in \(\tzhu{\OSPUniv{k}}\) if we replace $[\cdot]$ by $[\cdot]^\parity$.  As above, we conclude that the images of all the descendants of $\phi_k$ by non-positive \emph{even} modes belong to the ideal generated by $[\phi_k]^\parity$ in $\tzhu{\OSPUniv{k}}$.  However, as in the proof of \cref{prop:UnivZhu}, even non-positive mode descendants of $\phi_k$ have images that are equivalent to a linear combination of non-positive mode descendants in which all odd modes are zero modes. Since the action of the zero modes on $\phi_k$ generates a simple weight $\myosp$-module whose even subspace is a simple weight $\sltwo$-module, by \cref{thm:ospclass}, each pair of odd zero modes may be replaced by an even zero mode.  We conclude that the image of every even non-positive mode descendant of $\phi_k$ is in the ideal of $\EnvAlg{\SLA{sl}{2}}$ generated by $[\phi_k]^\parity$, as required.
\end{proof}

\begin{remark}
	We mention that in the proof of \cref{prop:UnivZhu}, we could remove the odd zero modes along with the positive modes because they annihilate $\vac{k}$.  In the proof of \cref{prop:SimpZhu} above, $\phi_k$ need not be annihilated by the odd zero modes, hence they are not removed.
\end{remark}

\begin{remark}
	This proposition remains true for $\zhu{\OSPMinMod{u}{v}}$ if we replace $\phi_k$ by $\chi_k$ throughout. The proof follows along the same lines with minor adjustments. For $\tzhu{\OSPMinMod{u}{v}}$, we can make this replacement if $\chi_k$ is even.
\end{remark}

The simple positive-energy \ns{} modules of \(\OSPMinMod{u}{v}\) are the simple quotients of the inductions of the simple
\(\zhu{\OSPMinMod{u}{v}}\)-modules, while the simple positive-energy
Ramond modules of \(\OSPMinMod{u}{v}\) are the simple quotients
of the inductions of the simple \(\tzhu{\OSPMinMod{u}{v}}\)-modules.
Moreover, the simple \(\zhu{\OSPMinMod{u}{v}}\)-modules are precisely those of \cref{thm:ospclass}
on which \(\left[\phi_k\right]\) acts trivially, while the simple
\(\tzhu{\OSPMinMod{u}{v}}\)-modules are those of \cref{thm:slclass}
on which \(\left[\phi_k\right]^\parity\) acts trivially.

The level of interest here is \(k=-\frac{5}{4}\), thus \(u=2\) and \(v=4\) in the setup of \eqref{eq:DefUV},
which implies that the singular vector has conformal weight $2$ and
\(\SLSA{osp}{1}{2}\)-weight \(\alpha\). One choice of normalisation of this vector is
\begin{equation}
  \chi_{-5/4}=\brac*{x_{-2}-4y_{-1}e_{-1}+2h_{-1}x_{-1}}\Omega_{-5/4}.
\end{equation}
Then,
\begin{equation}
  \phi_{-5/4}=y_0\chi_{-5/4}=
  \brac*{h_{-2}+2h^2_{-1}+8f_{-1}e_{-1}+6y_{-1}x_{-1}}\Omega_{-5/4}.
\end{equation}

\begin{prop} \label{prop:SimpZhu'}
  The images of \(\phi_{-5/4}\) in the untwisted and twisted Zhu
  algebras are
  \begin{equation}
		[\phi_{-5/4}]=\Sigma\brac*{2\Sigma-1}, \quad
    [\phi_{-5/4}]^\parity=4Q+\frac{15}{8},
  \end{equation}
  where $\Sigma$ and $Q$ are the super-Casimir of \(\myosp\) and the Casimir of \(\sltwo\), respectively.
\end{prop}
\begin{proof}
  The proof is very similar to proofs in \cite{TsuExt13,RidJac14,RidRel15,BloSVir16},
  so we will only briefly outline the reasoning.
  Since the \(\myosp\)-weight of \(\phi_{-5/4}\) is 0, its images in
  the Zhu algebras will also have weight 0 and hence these images lie in the centralisers of the
  Cartan subalgebras \(\cosp\) and \(\csl\). Further, since these images are
  uniquely determined by their action on weight spaces, all that remains to complete
  this proof is the computation of this action.

  The field corresponding to \(\phi_{-5/4}\) is
  \begin{equation}
    \phi(z) = \pd h(z) + 2 \normord{h(z) h(z)} + 8 \normord{f(z) e(z)} + 6 \normord{y(z) x(z)}
  \end{equation}
  and thus its zero mode is
  \begin{equation}
    \phi_0 = - h_0 + 2 \normord{h h}_0 + 8 \normord{f e}_0 + 6 \normord{y x}_0.
  \end{equation}

  Evaluating \(\phi_0\) on a \ns{} ground state vector \(u_{\lambda,s}\) of
  \(\SLSA{osp}{1}{2}\)-weight \(\lambda\alpha\) and \(\Sigma\)-eigenvalue
  \(s\) produces
  \begin{align}
    \phi_0u_{\lambda,s}
    &=\left(- h_0 + 2 h_0^2 + 8 e_0f_0 - 6 x_0y_0\right)u_{\lambda,s}
    =s(2s-1)u_{\lambda,s},
  \end{align}
  where the normally ordered products were evaluated using \eqref{eqn:bosnormord}.
  Thus, the image of \(\phi_{-5/4}\) in \(\zhu{\OSPUniv{k}}\) is \(\Sigma(2\Sigma-1)\), as required.

  Similarly, evaluating \(\phi_0\) on a Ramond ground state \(v_{\lambda,q}\) of \(\SLA{sl}{2}\)-weight
  \(\lambda\alpha\) and \(Q\)-eigenvalue \(q\) produces
  \begin{align}
    \phi_0v_{\lambda,q}
    &=\left(- h_0 + 2 h_0^2 + 8 e_0f_0 + 6 \left(\frac{5}{16}-\frac{1}{2}h_0\right)\right)v_{\lambda,q}
    =\brac*{4q+\frac{15}{8}}v_{\lambda,q},
  \end{align}
  where the odd normally ordered product $\normord{yx}_0$ was evaluated using \eqref{eqn:fermgcr}.  This, in turn, implies that the image of \(\phi_{-5/4}\) in \(\tzhu{\OSPUniv{k}}\) is \(4Q+\frac{15}{8}\), completing the proof.
\end{proof}

\subsection{Classifying $\OSPMinMod{2}{4}$-modules} \label{subsec:simospmod}

Given our explicit identifications of the untwisted and twisted Zhu algebras \(\zhu{\OSPMinMod{2}{4}}\) and \(\tzhu{\OSPMinMod{2}{4}}\), we can now easily classify the simple relaxed \hw{} $\OSPMinMod{2}{4}$-modules in both the \ns{} and Ramond sectors.

\begin{thm} \label{thm:classB24NS}
  Any simple \ns{} relaxed \hw{} $\OSPMinMod{2}{4}$-module is isomorphic to
  one of the following mutually non-isomorphic modules:
  \begin{equation}
		\begin{alignedat}{3}
			&\NSFin{0}, & &\NSInf{-1/2}^+, & &\NSInf{1/2}^-, \\
			&\parrmod{\NSFin{0}}, &\quad &\parrmod{\NSInf{-1/2}^+}, &\quad &\parrmod{\NSInf{1/2}^-},
		\end{alignedat}
		\quad \NSRel{\Lambda, 0} \quad \text{(\(\Lambda \in \CC / 2 \ZZ \setminus \set{[\pm\tfrac{1}{2}]}\))}.
	\end{equation}
\end{thm}
\noindent We remark that the parity reversal of $\NSRel{\Lambda, 0}$ is isomorphic to $\NSRel{\Lambda + 1,0}$, by \eqref{eqn:modperiodicity}, so does appear in the list above.
\begin{proof}
  The space of ground states of any simple \ns{} relaxed \hw{} $\OSPMinMod{2}{4}$-module must be isomorphic to one of the simple \(\myosp\)-modules listed in \cref{thm:ospclass} on which \([\phi_{-5/4}] = \Sigma (2 \Sigma - 1)\) acts trivially. The given classification consists of precisely those simple quotients of the inductions of the \(\myosp\)-modules for which this is the case.
\end{proof}

\begin{remark}
  One potential source of confusion that is worth mentioning is that as \(\frac{1}{2}\) is an
  allowed eigenvalue of the super-Casimir \(\Sigma\), yet \(-\frac{1}{2}\) is
  not, any simple $\zhu{\OSPMinMod{2}{4}}$-module with a weight space on which the eigenvalue of
  \(\Sigma\) is \(\frac{1}{2}\) cannot have any weight spaces of opposite
  parity. The only simple $\myosp$-modules satisfying this constraint are \(\Fin{0}\) and its
  parity reversal.
\end{remark}

\begin{remark}
	The gaps in the range of $\Lambda$, for $\NSRel{\Lambda, 0}$, in the previous \lcnamecref{thm:classB24NS} are only to guarantee simplicity.  The reducible, but indecomposable, modules $\NSRel{1/2 + 2 \ZZ, 0}^{\pm}$ and $\NSRel{-1/2 + 2 \ZZ, 0}^{\pm}$ have ground states that obviously satisfy the Zhu constraints, hence they are also $\OSPMinMod{2}{4}$-modules (in the \ns{} sector).
\end{remark}

Given that all the $\NSRel{}$-type $\OSPMinMod{2}{4}$-modules have $s=0$, we shall simplify notation in what follows by writing $\NSRel{\lambda}$, $\NSRel{1/2}^{\pm}$ and $\NSRel{-1/2}^{\pm}$ instead of $\NSRel{\lambda + 2 \ZZ, 0}$, $\NSRel{1/2 + 2 \ZZ, 0}^{\pm}$ and $\NSRel{-1/2 + 2 \ZZ, 0}^{\pm}$, respectively.  We shall also further abuse this notation by identifying $\lambda$ with $\Lambda = \lambda + 2 \ZZ$ when convenient, so that $\lambda \in \CC / 2 \ZZ$ when parametrising a $\NSRel{}$-type module.

\begin{thm} \label{thm:classB24R}
  Any simple Ramond relaxed \hw{} $\OSPMinMod{2}{4}$-module is isomorphic to
  one of the following mutually non-isomorphic modules:
  \begin{equation}
		\begin{alignedat}{5}
			&\RInf{-5/4}^+, & &\RInf{-3/4}^+, & &\RInf{3/4}^-, & &\RInf{5/4}^- & &\RRel{\Lambda, -15/32}, \\
			&\parrmod{\RInf{-5/4}^+}, &\quad &\parrmod{\RInf{-3/4}^+}, &\quad &\parrmod{\RInf{3/4}^-}, &\quad &\parrmod{\RInf{5/4}^-} &\quad &\parrmod{\RRel{\Lambda, -15/32}}
		\end{alignedat}
		\quad \text{(\(\Lambda \in \CC / 2 \ZZ \setminus \set{[\pm\tfrac{5}{4}]}\))}.
	\end{equation}
\end{thm}
\begin{proof}
	The space of ground states 	of any simple Ramond relaxed \hw{} $\OSPMinMod{2}{4}$-module must be isomorphic to one of the	simple \(\sltwo\)-modules listed in \cref{thm:slclass} on which \([\phi_{-5/4}]^\parity = 4Q + \frac{15}{8}\) acts trivially. The given classification consists of those simple quotients of the inductions of the $\sltwo$-modules for which this is the case.
\end{proof}

\begin{remark}
	Again, the gaps in the range of $\Lambda$, for $\RRel{\Lambda, -15/32}$, in the previous \lcnamecref{thm:classB24R} guarantee simplicity.  The reducible, but indecomposable, modules $\RRel{\lambda + 2 \ZZ, -15/32}^{\pm}$, with $\lambda = \pm\frac{5}{4}$, are also $\OSPMinMod{2}{4}$-modules (in the Ramond sector).
\end{remark}

Because all these $\RRel{}$-type modules have $q = -\frac{15}{32}$, we shall again simplify notation in what follows by writing $\RRel{\lambda}$ instead of $\RRel{\lambda + 2 \ZZ, -15/32}$, for $\lambda \neq \frac{5}{4} \bmod{2}$, and $\RRel{\lambda}^{\pm}$ instead of $\RRel{\lambda + 2 \ZZ, -15/32}^{\pm}$, for $\lambda = \pm\frac{5}{4}$.  As above, we shall also abuse this notation by identifying $\lambda$ with $\Lambda = \lambda + 2 \ZZ$ when parametrising $\RRel{}$-type modules.

These theorems classify the simple positive-energy weight modules over $\OSPMinMod{2}{4}$ because a weight vector of minimal conformal weight in a simple module must generate the entire module and hence must be a relaxed \hwv{}.  However, we can also twist each of these modules by conjugation and/or spectral flow (\cref{prop:AutVOAMod}).  If we twist by conjugation, the result will again be a simple positive-energy weight module.  Explicitly, we have
\begin{equation} \label{eq:ConjSimples}
	\begin{aligned}
		\conjmod{\NSFin{0}} &\cong \NSFin{0}, &
		\conjmod{\NSInf{-1/2}^+} &\cong \NSInf{1/2}^-, &
		\conjmod{\NSRel{\lambda}} &\cong \NSRel{-\lambda} &
		&\text{(\(\lambda \neq \pm \tfrac{1}{2} \bmod{2}\)),} \\
		\conjmod{\RInf{-5/4}^+} &\cong \RInf{5/4}^-, &
		\conjmod{\RInf{-3/4}^+} &\cong \RInf{3/4}^-, &
		\conjmod{\RRel{\lambda}} &\cong \RRel{-\lambda} &
		&\text{(\(\lambda \neq \pm \tfrac{5}{4} \bmod{2}\)).}
	\end{aligned}
\end{equation}
The conjugates of the spectral flows of these modules now follow (formally) from \eqref{eqn:twistrels}.

However, twisting these simple positive-energy weight modules by spectral flow will not result in modules that are positive-energy, in general.  Na\"{\i}vely, we arrive at the following list of (isomorphism classes of) simple \ns{} weight modules (we also include their parity reversals, recalling that $\parrmod{\NSRel{\lambda}} \cong \NSRel{\lambda + 1}$):
\begin{equation} \label{eqn:nsadmisslist}
\begin{aligned}
\sfmod{\ell}{\NSFin{0}}, \quad
\sfmod{\ell}{\NSInf{-1/2}^+}, \quad
\sfmod{\ell}{\NSInf{1/2}^-}, \quad
\sfmod{\ell}{\NSRel{\lambda}}& &
&\text{(\(\ell \in \ZZ\) and \(\lambda \neq \tfrac{1}{2} \bmod{1}\)),} \\
\sfmod{\ell}{\RInf{-3/4}^+}, \quad
\sfmod{\ell}{\RInf{-5/4}^+}, \quad
\sfmod{\ell}{\RInf{3/4}^-}, \quad
\sfmod{\ell}{\RInf{5/4}^-}, \quad
\sfmod{\ell}{\RRel{\mu}}& &
&\text{(\(\ell \in \ZZ + \tfrac{1}{2}\) and \(\mu \neq \pm\tfrac{5}{4} \bmod{2}\)).}
\end{aligned}
\end{equation}
The simple Ramond weight modules are obtained from these by applying $\sfaut^{1/2}$ to each of the above.  It turns out however, that these lists contain many pairs of isomorphic modules. In particular, it is easy to show (using similar techniques to those used in the proof of \eqref{eqn:rverma}) that the following coincidences hold:
\begin{equation} \label{eqn:sfA0}
\begin{aligned}
	\sfmod{1/2}{\NSFin{0}} &\cong \RInf{-5/4}^+, & \sfmod{-1/2}{\NSFin{0}} &\cong \RInf{5/4}^-, \\
	\sfmod{1/2}{\NSInf{1/2}^-} &\cong \RInf{-3/4}^+, & \sfmod{-1/2}{\NSInf{-1/2}^+} &\cong \RInf{3/4}^-.
\end{aligned}
\end{equation}
With these isomorphisms in mind, the list of (isomorphism classes of) mutually non-isomorphic simple weight modules over  $\OSPMinMod{2}{4}$ reduces to
\begin{equation} \label{eqn:admisslist}
\sfmod{\ell}{\NSFin{0}}, \quad \sfmod{\ell}{\NSInf{-1/2}^+}, \quad \sfmod{\ell}{\NSInf{1/2}^-}, \quad \sfmod{\ell}{\NSRel{\lambda}}, \quad \sfmod{\ell}{\RRel{\mu}}
\end{equation}
and their parity reversals, where $\ell \in \frac{1}{2} \ZZ$, $\lambda \neq \tfrac{1}{2} \bmod{1}$ and $\mu \neq \pm\tfrac{5}{4} \bmod{2}$. To establish that these isomorphism classes are indeed distinct, one may observe that their spaces of extremal  states are different, these being the weight vectors of minimal conformal weight in each subspace of constant $\myosp$-weight.

We shall assume that the physical category $\categ{P}$ relevant for constructing minimal model \cfts{} has precisely the modules \eqref{eqn:admisslist} as its simple objects.  We shall also assume that the category is closed under conjugation and fusion products.  One of the main aims of this paper is to test this assumption by demonstrating that it leads to satisfactory modular properties and that the \emph{standard} Verlinde formula, introduced below in \eqref{eq:fusecoeff}, returns non-negative integer Grothendieck fusion coefficients.

The assumption of closure under fusion is highly non-trivial from a mathematical perspective because the modules in $\categ{P}$, in particular the relaxed \hwms{} $\NSRel{\lambda}$, $\RRel{\mu}$ and their spectral flows, are not all $C_1$-cofinite.  This means that fusion products might not be finite-length (meaning that they have a finite number of composition factors).  However, the failure of $C_1$-cofiniteness is also observed with relaxed modules for admissible-level $\sltwo$ minimal models and bosonic ghosts \cite{RidBos14}, whilst the known fusion calculations \cite{GabFus01,RidFus10,RidBos14} all give finite-length results.  These calculations were performed using the Nahm-Gaberdiel-Kausch algorithm \cite{NahQua94,GabInd96} whose (essential) equivalence to the rigorous theory of Huang, Lepowsky and Zhang \cite{HuaLog10} will appear in \cite{KanNGK18}.  The results give us confidence that our assumptions on the $\OSPMinMod{2}{4}$ fusion rules all hold, though they will be extremely difficult to establish rigorously.

The conjugates and spectral flow images of the reducible but indecomposable $\OSPMinMod{2}{4}$-modules identified above are, likewise, reducible but indecomposable weight modules over $\OSPMinMod{2}{4}$, by \cref{prop:AutVOAMod}.  The conjugates are easily identified using \eqref{eqn:twistrels}:
\begin{equation}
	\conjmod{\NSRel{1/2}^+} \cong \NSRel{-1/2}^-, \quad
	\conjmod{\RRel{3/4}^+} \cong \RRel{-3/4}^-, \quad
	\conjmod{\RRel{5/4}^+} \cong \RRel{-5/4}^-.
\end{equation}
Up to parity reversal, we therefore obtain the following mutually non-isomorphic reducible, but indecomposable, modules:
\begin{equation} \label{eqn:indeclist}
	\sfmod{\ell}{\NSRel{1/2}^+}, \quad
	\sfmod{\ell}{\NSRel{-1/2}^-}, \quad
	\sfmod{\ell}{\RRel{3/4}^+}, \quad
	\sfmod{\ell}{\RRel{5/4}^+}, \quad
	\sfmod{\ell}{\RRel{-3/4}^-}, \quad
	\sfmod{\ell}{\RRel{-5/4}^-} \quad
	\text{(\(\ell \in \frac{1}{2} \ZZ\)).}
\end{equation}
We record the defining non-split short exact sequences, easily deduced from \eqref{ses:CE}, for future convenience:
\begin{equation} \label{eqn:sesindec}
\begin{gathered}
\dses{\parrmod{\NSInf{-1/2}^+}}{}{\NSRel{1/2}^+}{}{\NSInf{1/2}^-}, \\
\dses{\RInf{-5/4}^+}{}{\RRel{3/4}^+}{}{\RInf{3/4}^-}, \\
\dses{\RInf{-3/4}^+}{}{\RRel{5/4}^+}{}{\RInf{5/4}^-},
\end{gathered}
\qquad
\begin{gathered}
\dses{\parrmod{\NSInf{1/2}^-}}{}{\NSRel{-1/2}^-}{}{\NSInf{-1/2}^+}, \\
\dses{\RInf{5/4}^-}{}{\RRel{-3/4}^-}{}{\RInf{-3/4}^+}, \\
\dses{\RInf{3/4}^-}{}{\RRel{-5/4}^-}{}{\RInf{-5/4}^+}.
\end{gathered}
\end{equation}
Sequences for the spectral flows of these modules now follow because the $\sigma^{\ell}$ induce exact covariant functors.

These results indicate that the minimal model $\OSPMinMod{2}{4}$ provides another example of a \cft{} to which the \emph{standard module formalism} of \cite{CreLog13,RidVer14} applies.  To wit, the simple modules $\sfmod{\ell}{\NSRel{\lambda}}$ and $\sfmod{\ell}{\RRel{\mu}}$ comprise the \emph{typical} modules, while the reducible modules of \eqref{eqn:indeclist} are examples of \emph{atypical} modules.  Together, the $\NSRel{}$- and $\RRel{}$-type modules, along with their spectral flows, constitute the \emph{standard modules} of the minimal model.  Subquotients of atypical modules are also said to be atypical, hence the $\sfmod{\ell}{\NSFin{0}}$ and $\sfmod{\ell}{\NSInf{\pm 1/2}^{\mp}}$ are all atypical and thus so are the $\sfmod{\ell}{\RInf{\pm 3/4}^{\mp}}$ and $\sfmod{\ell}{\RInf{\pm 5/4}^{\mp}}$.

The standard module formalism is a collection of empirical observations, first described in \cite{CreRel11}, that coherently organises the modular properties of a large number of logarithmic \cfts{}.  To summarise, the characters of the standard modules (which form a continuously parametrised family) transform rather simply under a natural action of $\SLG{SL}{2;\ZZ}$ with integral kernels replacing the familiar S- and T-matrices.  The transforms of the atypical characters, which include that of the vacuum module, are much more subtle, but may be deduced by constructing infinite (one-sided) resolutions of each atypical module in terms of standard ones.  This rich formalism will provide the starting point for our investigation of the modular properties of the minimal model.  First however, we need to determine the characters of these modules.

\section{The affine minimal model $\OSPMinMod{2}{4}$: Characters} \label{sec:chars}

In this section, we determine the characters and supercharacters of all the simple $\OSPMinMod{2}{4}$-modules of category $\categ{P}$.  We also introduce the Grothendieck group of this category and show explicitly that the images of the standard $\OSPMinMod{2}{4}$-modules in the Grothendieck group form a basis (of a certain completion).

\subsection{Highest-weight characters} \label{subsec:simpchars}

Having identified the modules of interest, we now wish to calculate their characters. Formally, the \emph{character} of a weight-module $\Mod{M}$ of $\myaosp$ is defined to be the power series
\begin{equation} \label{eqn:defchar}
	\fch{\Mod{M}}{y,z,q} = \traceover{\Mod{M}} \brac*{y^k z^{h_0} q^{L_0 - c/24}},
\end{equation}
where $\traceover{\Mod{M}} \brac*{X}$ denotes the trace of the image in $\End \brac*{\Mod{M}}$ of $X \in \aalgh$ under the given representation of $\myaosp$.  We append the subscript $0$ to our notation for characters to indicate that eigenvectors of different parities count equally in the sum.  The notation for supercharacters, where parity matters, will be introduced in \cref{subsec:supchars}. It is customary to interpret such an expression as defining a (meromorphic) function, however some care is needed to make this identification precise. In particular, we must additionally specify a domain for the variables $z$ and $q$ in order to avoid misleading results (further discussion of this may be found in \cite{RidSL208,CreMod12}).

It is clear that
\begin{equation} \label{eqn:chpar}
	\ch{\parrmod{\Mod{M}}} = \ch{\Mod{M}},
\end{equation}
for any weight module $\Mod{M}$ over $\myaosp$, so it might appear that we are justified in identifying modules with their parity reversals, at least for the purposes of studying modular transformations.  However, we shall see that the parity of a module is completely determined by its character and supercharacter, so one can employ modular methods whilst retaining parity information.

As discussed in \cref{subsec:irrwtmods}, the characters of the simple \hw{} $\OSPMinMod{2}{4}$-modules can be expressed in terms of Verma module characters using resolutions.  As usual, the characters of Verma modules are very easy to calculate. A straightforward application of the \pbw{} theorem gives, recalling that $c = -5$,
\begin{equation} \label{eqn:vermachar}
\fch{\NSVer{0}}{y,z,q} = y^{-5/4} q^{5/24} \prod_{m=1}^\infty \frac{ \brac*{1+z^{-1}q^{m-1}} \brac*{1+zq^{m}} }{ \brac*{1-z^{-2}q^{m-1}} \brac*{1-q^{m}} \brac*{1-z^{2}q^{m}} },
\end{equation}
which can be re-expressed more compactly as
\begin{equation} \label{eqn:vermatheta}
\fch{\NSVer{0}}{y,z,q} = -\ii y^{-5/4} z^{1/2} q^{1/4} \frac{\fjth{2}{z;q}}{\fjth{1}{z^2 ;q} \; \deta{q}},
\end{equation}
using the standard infinite product formulae for the Jacobi theta functions $\jth{j}$ and the Dedekind eta function $\eta$. More generally, a simple calculation shows that \ns{} Verma characters are given by
\begin{equation} \label{eqn:vermageneral}
\fch{\NSVer{\lambda}}{y,z,q} = z^{\lambda} q^{\Delta} \fch{\NSVer{0}}{y,z,q},
\end{equation}
where $\Delta$ was given in \eqref{eqn:cwhw}.  This relation holds for all non-critical $k$, simply by replacing the exponent of $y$ by $k$.

From the embeddings of Verma modules shown in \cref{fig:nsverma}, it follows that the characters of the simple \ns{} \hw{} $\OSPMinMod{2}{4}$-modules are given by
\begin{subequations} \label{eqn:nschar'}
	\begin{align}
		\fch{\NSFin{0}}{y,z,q} &= \sum_{n \in 2 \ZZ} \fch{\NSVer{-5/4,n}}{y,z,q} - \sum_{n \in 2 \ZZ + 1} \fch{\NSVer{-5/4,n}}{y,z,q} \notag \\
		&= \sum_{n \in \ZZ} (-1)^n z^{n} q^{n(n+1)} \fch{\NSVer{0}}{y,z,q}
		= y^{-5/4} \frac{\fjth{1}{z;q^2} \; \fjth{2}{z;q}}{\fjth{1}{z^2 ;q} \; \deta{q}}, \label{eqn:nsfinchar} \\
		\fch{\NSInf{-1/2}^+}{y,z,q} &= \sum_{n \in 2 \ZZ} \fch{\NSVer{n-1/2}}{y,z,q} - \sum_{n \in 2 \ZZ + 1} \fch{\NSVer{n-1/2}}{y,z,q} \notag \\
		&= \sum_{n \in \ZZ} (-1)^n z^{n-1/2} q^{n^2 - 1/4} \fch{\NSVer{0}}{y,z,q}
		= -\ii y^{-5/4} \frac{\fjth{4}{z;q^2} \; \fjth{2}{z;q}}{\fjth{1}{z^2 ;q} \; \deta{q}}. \label{eqn:nsinfchar}
	\end{align}
\end{subequations}

In order to calculate the characters of $\NSInf{1/2}^- \cong \conjmod{\NSInf{-1/2}^+}$ and the other twists of these \hwms{}, we use the following relation:
\begin{equation} \label{eqn:twisttrace}
\traceover{\varphi (\Mod{M})} \brac*{X} = \traceover{\Mod{M}} \brac*{\varphi^{-1} \brac*{X}} \quad
\text{(\(\varphi \in \SLG{Aut}{\aalgg}\), \(X \in \aalgh\)).}
\end{equation}
Specialising to $\varphi = \conjaut$ and $\sfaut^{\ell}$, $\ell \in \frac{1}{2} \ZZ$, we obtain the useful formulae
\begin{equation} \label{eqn:autch}
\fch{\conjmod{\Mod{M}}}{y,z,q} = \fch{\Mod{M}}{y,z^{-1},q}, \quad
\fch{\sfmod{\ell}{\Mod{M}}}{y,z,q} = \fch{\Mod{M}}{y z^{2\ell} q^{\ell^2},z q^{\ell},q},
\end{equation}
the latter following from $\sfmod{\ell}{h_0} = h_0 - 2 \ell K$ and $\sfmod{\ell}{L_0} = L_0 - \ell h_0 + \ell^2 K$.  Under the transformation $z \mapsto z^{-1}$, the Jacobi theta functions $\fjth{j}{z;q}$, for $j=2,3,4$, are invariant whilst $\fjth{1}{z^{-1};q} = -\fjth{1}{z;q}$.  We therefore arrive at the character of $\NSInf{1/2}^- \cong \conjmod{\NSInf{-1/2}^+}$:
\begin{equation} \label{eqn:nscinfchar}
\fch{\NSInf{1/2}^-}{y,z,q} = \fch{\NSInf{-1/2}^+}{y,z^{-1},q}
= \ii y^{-5/4} \frac{\fjth{4}{z;q^2} \; \fjth{2}{z;q}}{\fjth{1}{z^2 ;q} \; \deta{q}}.
\end{equation}

Whilst it is tempting to treat the expressions in \eqref{eqn:nsfinchar}, \eqref{eqn:nsinfchar} and \eqref{eqn:nscinfchar} as defining meromorphic functions of $(y,z,q)$, doing so means that we must carefully restrict their domains. Outside of the appropriate domain of validity, these meromorphic functions will give different power series expansions about the origin. Since characters are defined precisely by such expansions, a given meromorphic function may correspond to a multitude of different characters.

It turns out that \eqref{eqn:nsfinchar} is valid (and its \rhs{} is convergent), for all $y,z,q \in \CC$ with $y\neq0$, $0<\abs{q}<1$ and $z$ restricted to the annulus
\begin{equation} \label{eqn:nsfindomain}
\abs{q}^{1/2} < \abs{z} < \abs{q}^{-1/2}.
\end{equation}
In particular, this includes $z=1$, which one should expect, given that the $L_0$-eigenspaces of $\NSFin{0}$ are all finite-dimensional.  Similarly, \eqref{eqn:nsinfchar} and \eqref{eqn:nscinfchar} are valid for $y \neq 0$, $0<\abs{q}<1$ and $z$ restricted to the disjoint annuli
\begin{equation} \label{eqn:nsinfdomains}
1 < \abs{z} < \abs{q}^{-1/2} \quad \text{and} \quad \abs{q}^{1/2} < \abs{z} < 1,
\end{equation}
respectively. Note that neither of the modules $\NSInf{\mp 1/2}^{\pm}$ has finite-dimensional $L_0$-eigenspaces, so their characters must indeed diverge at $z=1$.

In order to calculate the characters of the spectral flows of the above modules, we apply \eqref{eqn:autch}.  For example, using the properties of Jacobi theta functions, we obtain
\begin{subequations}
	\begin{equation} \label{eqn:sfcha0even}
		\fch{\sfmod{\ell}{\NSFin{0}}}{y,z,q} = (-1)^{\ell / 2} y^{-5/4} \frac{\fjth{1}{z;q^2} \; \fjth{2}{z;q}}{\fjth{1}{z^2 ;q} \; \deta{q}},
	\end{equation}
	if $\ell$ is even, and
	\begin{equation} \label{eqn:sfcha0odd}
		\fch{\sfmod{\ell}{\NSFin{0}}}{y,z,q} = \ii (-1)^{(\ell-1) / 2} y^{-5/4} \frac{\fjth{4}{z;q^2} \; \fjth{2}{z;q}}{\fjth{1}{z^2 ;q} \; \deta{q}},
	\end{equation}
\end{subequations}
if $\ell$ is odd.  It may seem, at first glance, that this sequence of characters holds a (perhaps surprising) four-fold periodicity. This semblance breaks down however, once we take into account the domains on which these expressions are valid. Specifically, \eqref{eqn:sfcha0even} and \eqref{eqn:sfcha0odd} will only recover the appropriate characters, as formal power series, when expanded in the domain $0 < \abs{q} < 1$ and
\begin{equation} \label{eqn:sffindomain}
\abs{q}^{1/2 - \ell} < \abs{z} < \abs{q}^{-1/2 - \ell},
\end{equation}
as follows from \cref{eqn:autch,eqn:nsfindomain}.

We can now apply \eqref{eqn:sfA0} and \eqref{eqn:autch} to the character formulae \eqref{eqn:nschar'} and \eqref{eqn:nscinfchar} in order to calculate the characters of the Ramond \hwms{} and their conjugates.  We summarise the results, along with the \ns{} formulae presented above, in the following \lcnamecref{prop:AtypCh}.
\begin{proposition} \label{prop:AtypCh}
	The characters of the \hw{} $\OSPMinMod{2}{4}$-modules are given by the following expressions, valid for $y \neq 0$, $0 < \abs{q} < 1$ and $z$ in the annuli given:
	\begin{equation} \label{eqn:char}
		\begin{aligned}
			\fch{\NSFin{0}}{y,z,q} &= y^{-5/4} \frac{\fjth{1}{z;q^2} \; \fjth{2}{z;q}}{\fjth{1}{z^2 ;q} \; \deta{q}} & &\text{(\(\abs{q}^{1/2} < \abs{z} < \abs{q}^{-1/2}\)),} \\
			\fch{\NSInf{-1/2}^+}{y,z,q} &= -\ii y^{-5/4} \frac{\fjth{4}{z;q^2} \; \fjth{2}{z;q}}{\fjth{1}{z^2 ;q} \; \deta{q}} & &\text{(\(1 < \abs{z} < \abs{q}^{-1/2}\)),} \\
			\fch{\RInf{-5/4}^+}{y,z,q} &= -\ii y^{-5/4} z^{-1/4} q^{1/16} \frac{\fjth{4}{zq^{-1/2};q^2} \; \fjth{3}{z;q}}{\fjth{1}{z^2 ;q} \; \deta{q}} & &\text{(\(1 < \abs{z} < \abs{q}^{-1}\)),} \\
			\fch{\RInf{-3/4}^+}{y,z,q} &= y^{-5/4} z^{-1/4} q^{1/16} \frac{\fjth{1}{zq^{-1/2};q^2} \; \fjth{3}{z;q}}{\fjth{1}{z^2 ;q} \; \deta{q}} & &\text{(\(1 < \abs{z} < \abs{q}^{-1/2}\)).}
		\end{aligned}
	\end{equation}
\end{proposition}

\begin{remark}
	The characters of the conjugates of the \hw{} $\OSPMinMod{2}{4}$-modules are now easily computed using \cref{eqn:autch}, noting that $\NSFin{0}$ is self-conjugate:
	\begin{equation} \label{eqn:conjchar}
		\begin{aligned}
			\fch{\NSInf{1/2}^-}{y,z,q} &= \ii y^{-5/4} \frac{\fjth{4}{z;q^2} \; \fjth{2}{z;q}}{\fjth{1}{z^2 ;q} \; \deta{q}} & &\text{(\(\abs{q}^{1/2} < \abs{z} < 1\)),} \\
				\fch{\RInf{5/4}^-}{y,z,q} &= -y^{-5/4} z^{-1/4} q^{1/16} \frac{\fjth{1}{zq^{-1/2};q^2} \; \fjth{3}{z;q}}{\fjth{1}{z^2 ;q} \; \deta{q}} & &\text{(\(\abs{q} < \abs{z} < 1\)),} \\
				\fch{\RInf{3/4}^-}{y,z,q} &= \ii y^{-5/4} z^{-1/4} q^{1/16} \frac{\fjth{4}{zq^{-1/2};q^2} \; \fjth{3}{z;q}}{\fjth{1}{z^2 ;q} \; \deta{q}} & &\text{(\(\abs{q}^{1/2} < \abs{z} < 1\)).}
		\end{aligned}
	\end{equation}
	The characters of the remaining simple atypical $\OSPMinMod{2}{4}$-modules are similarly obtained using spectral flow.
\end{remark}

\subsection{Relaxed highest-weight characters} \label{subsec:relchars}

For $\myaosp$, there seems to be no literature addressing the submodule structure of relaxed Verma modules.  Even though character formulae for relaxed Verma modules over $\myaosp$ are easy to obtain, we are therefore unable to deduce character formulae for their simple quotients.  Instead, we resort to an indirect method developed for $\AKMA{sl}{2}$-modules in \cite{CreMod12,CreMod13} and rigorously justified in \cite{KawRel18}.

We start with the short exact sequences given in \eqref{eqn:sesindec} for the reducible, but indecomposable, \ns{} relaxed \hwms{} $\NSRel{\pm 1/2}^{\pm}$, immediately deducing the following character formulae:
\begin{equation} \label{eqn:c1/2+chardecomp}
\ch{\NSRel{\pm 1/2}^{\pm}} =  \ch{\NSInf{-1/2}^+} + \ch{\NSInf{1/2}^-}.
\end{equation}
Of course, this only holds when treating these characters as formal power series. Since the two expressions \eqref{eqn:nsinfchar} and \eqref{eqn:nscinfchar} that we have derived for the characters on the \rhs{}, are valid on disjoint domains, we must be careful to sum these characters as power series rather than as meromorphic functions --- na\"{\i}vely summing the derived characters, as meromorphic functions, gives zero!  For this, we appeal to the following (equivalent) identities, originally derived by Kac and Wakimoto in \cite{KacInt94}:
\begin{subequations} \label{eqn:kwu}
	\begin{align}
		\frac{\fjth{1}{uv;q} \; \deta{q}^3}{\fjth{1}{u;q} \; \fjth{1}{v;q}} &= -\ii \sum_{m \in \ZZ} \frac{u^m}{1-v q^m} \quad \text{(\(\abs{q} < \abs{u} < 1\), \(0<\abs{v}<\abs{q}<1\)),} \label{eqn:kwu1}\\
		\frac{\fjth{1}{uv;q} \; \deta{q}^3}{\fjth{1}{u;q} \; \fjth{1}{v;q}} &= -\ii \sum_{m \in \ZZ} \frac{u^m v q^m}{1-v q^m} \quad \text{(\(1 < \abs{u} < \abs{q}^{-1}\), \(0<\abs{v}<\abs{q}<1\)).} \label{eqn:kwu2}
	\end{align}
\end{subequations}
We refer to \cite[App.~B]{RidSL208} for the conventions on \(\theta\) functions used here.  Setting $u=z^2$ and rearranging, these become
\begin{subequations} \label{eqn:kwz}
	\begin{align}
		\frac{1}{\fjth{1}{z^2;q}} &= -\ii \frac{\fjth{1}{v;q}}{\fjth{1}{z^2 v;q} \; \deta{q}^3} \sum_{m \in \ZZ} \frac{z^{2m}}{1-v q^m} \quad \text{(\(\abs{q}^{1/2} < \abs{z} < 1\), \(0<\abs{v}<\abs{q}<1\)),} \label{eqn:kwz1} \\
		\frac{1}{\fjth{1}{z^2;q}} &= -\ii \frac{\fjth{1}{v;q}}{\fjth{1}{z^2 v;q} \; \deta{q}^3} \sum_{m \in \ZZ} \frac{z^{2m} v q^m}{1-v q^m} \quad \text{(\(1 < \abs{z} < \abs{q}^{-1/2}\), \(0<\abs{v}<\abs{q}<1\)).} \label{eqn:kwz2}
	\end{align}
\end{subequations}

We substitute \eqref{eqn:kwz1} into the character formula \eqref{eqn:nscinfchar}, noting that both have the same domain of validity, obtaining
\begin{equation}
\fch{\NSInf{1/2}^-}{y,z,q} = y^{-5/4} \frac{\fjth{4}{z;q^2} \; \fjth{2}{z,q}}{\deta{q}^4} \frac{\fjth{1}{v;q}}{\fjth{1}{z^2 v;q}} \sum_{m \in \ZZ} \frac{z^{2m}}{1-vq^m}.
\end{equation}
Performing the equivalent substitution of \eqref{eqn:kwz2} into \eqref{eqn:nsinfchar} (which also have the same domain of validity), we instead get
\begin{equation}
\fch{\NSInf{-1/2}^+}{y,z,q} = -y^{-5/4} \frac{\fjth{4}{z;q^2} \; \fjth{2}{z,q}}{\deta{q}^4} \frac{\fjth{1}{v;q}}{\fjth{1}{z^2 v;q}} \sum_{m \in \ZZ} \frac{z^{2m} v q^m}{1-vq^m}.
\end{equation}
Treating these as formal power series in $z$, we may forget their domains and sum them, giving
\begin{equation} \label{eqn:c1/2+charunsimp}
\fch{\NSRel{\pm 1/2}^{\pm}}{y,z,q} = y^{-5/4} \frac{\fjth{4}{z;q^2} \; \fjth{2}{z,q}}{\deta{q}^4} \frac{\fjth{1}{v;q}}{\fjth{1}{z^2 v;q}} \sum_{m \in \ZZ} z^{2m}.
\end{equation}
In this way, we have arrived at a character formula for the reducible \ns{} relaxed \hwms{}.  These are precisely the atypical standard modules in the \ns{} sector.

We may simplify the character formula \eqref{eqn:c1/2+charunsimp} further and eliminate the auxiliary variable $v$, but we must first introduce an alternative set of variables for the characters. Instead of giving characters as power series in $y,z,q$, it is standard (and essential for studying modular properties) to introduce $\psi$, $\zeta$ and $\tau$, defined by
\begin{equation} \label{eqn:modvars}
y = \ee^{2 \pi \ii \psi}, \quad z = \ee^{2 \pi \ii \zeta} \quad \text{and} \quad q = \ee^{2 \pi \ii \tau}.
\end{equation}
The relation between $z$ and $\zeta$ will be used repeatedly in the following manipulations.

Identifying the sum on the \rhs{} of \eqref{eqn:c1/2+charunsimp} as defining a distribution in $\zeta$, we can employ the Fourier-theoretic identity
\begin{equation} \label{eqn:deltasum}
\sum_{n \in \ZZ} \ee^{2 \pi \ii n w} = \sum_{n \in \ZZ} \ddelta{w-n}
\end{equation}
to make the following simplification:
\begin{equation} \label{eqn:elimv}
\frac{\fjth{1}{v;q}}{\fjth{1}{z^2 v;q}} \sum_{m \in \ZZ} z^{2m} = \frac{\fjth{1}{v;q}}{\fjth{1}{\ee^{4 \pi \ii \zeta} v;q}} \sum_{m \in \ZZ} \ddelta{2\zeta - m}
= \sum_{m \in \ZZ} \frac{\fjth{1}{v;q}}{\fjth{1}{\ee^{2 \pi \ii m} v;q}} \ddelta{2\zeta - m}
= \sum_{m \in \ZZ} \ee^{-\pi \ii m} \ddelta{2\zeta - m}.
\end{equation}
Here, we have noted that $\fjth{1}{\ee^{2 \pi \ii} v;q} = \ee^{\pi \ii} \fjth{1}{v;q}$. Now,
\begin{align} \label{eqn:c1/2+charsimp1}
\fch{\NSRel{\pm 1/2}^{\pm}}{y,z,q} &= y^{-5/4} \frac{\fjth{4}{z;q^2} \; \fjth{2}{z;q}}{\deta{q}^4} \sum_{m \in \ZZ} \ee^{-\pi \ii m} \ddelta{2\zeta - m} \notag \\
&= y^{-5/4} \frac{1}{\deta{q}^4} \sum_{m \in \ZZ} \frac{\fjth{4}{\ee^{\pi \ii m};q^2} \; \fjth{2}{\ee^{\pi \ii m};q}}{\ee^{\pi \ii m}} \ddelta{2\zeta - m},
\end{align}
which can be split into sums over even and odd values of $m$, giving
\begin{align} \label{eqn:c1/2+charsimp2}
\fch{\NSRel{\pm 1/2}^{\pm}}{y,z,q} = y^{-5/4} \frac{1}{\deta{q}^4}  \sum_{n \in \ZZ} (-1)^n &\Big[\fjth{4}{1;q^2} \; \fjth{2}{1;q} \; \ddelta{2\zeta - 2n} \notag \\
&+ \fjth{3}{1;q^2} \; \fjth{1}{1;q} \; \ddelta{2\zeta - 2n-1}\Big].
\end{align}
However, $\fjth{1}{1;q}$ vanishes identically, so that
\begin{equation} \label{eqn:c1/2+charsimp3}
\fch{\NSRel{\pm 1/2}^{\pm}}{y,z,q} = y^{-5/4} \frac{\fjth{4}{1;q^2} \; \fjth{2}{1;q}}{2 \, \deta{q}^4} \sum_{n \in \ZZ} \ee^{\pi \ii n} \ddelta{\zeta - n}
= y^{-5/4} \frac{\fjth{4}{1;q^2} \; \fjth{2}{1;q}}{2 \, \deta{q}^4} \sum_{n \in \ZZ} z^{n + 1/2}.
\end{equation}
The $q$-dependent factor can be further simplified, using the product forms of the Jacobi theta functions, giving our final result:
\begin{equation} \label{eqn:c1/2+char}
\fch{\NSRel{\pm 1/2}^{\pm}}{y,z,q} = \frac{y^{-5/4} z^{1/2}}{\deta{q}^2} \sqrt{\frac{\fjth{2}{1;q}}{2 \, \deta{q}}} \sum_{n \in \ZZ} z^n.
\end{equation}
We remark that the decoupling of the $z$- and $q$-dependences of this character is expected for a module whose ground states form a dense $\myosp$-module.

It is convenient, at this point, to extract the $q$-dependent terms from \eqref{eqn:c1/2+char} in the following definition:
\begin{equation} \label{eqn:defA}
	\taufn{j}{q} = \frac{1}{\deta{q}^2} \sqrt{\frac{\fjth{j}{1;q}}{\deta{q}}} \quad \text{(\(j=2,3,4\)).}
\end{equation}
We can rewrite these factors explicitly as
\begin{equation}
	\taufn{2}{q} = \sqrt{2} q^{-1/24} \prod_{i=1}^{\infty} \frac{1+q^i}{(1-q^i)^2}, \quad
	\taufn{3}{q} = q^{-5/48} \prod_{i=1}^{\infty} \frac{1+q^{i-1/2}}{(1-q^i)^2}, \quad
	\taufn{4}{q} = q^{-5/48} \prod_{i=1}^{\infty} \frac{1-q^{i-1/2}}{(1-q^i)^2}.
\end{equation}
Inserting into \eqref{eqn:c1/2+char} and expanding, one finds that
\begin{align}
	\fch{\NSRel{\pm 1/2}^{\pm}}{y,z,q} &= y^{-5/4} z^{1/2} \frac{\taufn{2}{q}}{\sqrt{2}} \sum_{n \in \ZZ} z^n = y^{-5/4} z^{1/2} q^{-1/24} \prod_{i=1}^{\infty} \frac{1+q^i}{(1-q^i)^2} \cdot \sum_{n \in \ZZ} z^n \notag \\
	&= y^{-5/4} z^{1/2} q^{-1/24} \brac*{1+3q+8q^2+19q^3+41q^4+83q^5+\cdots} \sum_{n \in \ZZ} z^n
\end{align}
which indeed gives the correct weight multiplicities (shown for the first six conformal grades) of $\NSRel{\pm 1/2}^{\pm}$.  We mention that the prefactor $q^{-1/24} = q^{\Delta - c/24}$, requires that $\Delta = -\frac{1}{4}$ (recall that $c=-5$), in accord with $s=0$ and \eqref{eq:RelConfWts}.

A similar calculation gives the characters of the atypical relaxed Ramond modules.  We record these results in the following \lcnamecref{prop:AtypStCh}.
\begin{proposition} \label{prop:AtypStCh}
	The characters of the atypical relaxed $\OSPMinMod{2}{4}$-modules are (shifted) formal power series in $z$ whose coefficients are holomorphic functions of $q$ for $\abs{q}<1$.  Explicitly, we have
	\begin{subequations}
		\begin{align}
			\fch{\NSRel{\pm 1/2}^{\pm}}{y,z,q} &= \frac{y^{-5/4} z^{1/2}}{\sqrt{2}} \taufn{2}{q} \sum_{n \in \ZZ} z^n, \label{eqn:c+char} \\
			\fch{\RRel{\pm 5/4}^{\pm}}{y,z,q} = \fch{\RRel{\mp 3/4}^{\mp}}{y,z,q} &= \frac{y^{-5/4} z^{\pm 5/4}}{2} \brac*{ \taufn{3}{q} \sum_{n \in \ZZ} z^n + \taufn{4}{q} \sum_{n \in \ZZ} (-z)^n}. \label{eqn:e+char}
		\end{align}
	\end{subequations}
\end{proposition}
\noindent With \eqref{eqn:autch}, we can now compute the characters of all the atypical standard modules.  However, we also need to compute the characters of the typical standard modules: the $\NSRel{\lambda}$, the $\RRel{\mu}$ and their spectral flows.  Happily, these may be deduced from those of the atypical standards merely by shifting the powers of $z$ in \eqref{eqn:c1/2+char} and \eqref{eqn:e+char}.
\begin{proposition} \label{prop:TypCh}
	The characters of the typical relaxed $\OSPMinMod{2}{4}$-modules are (shifted) formal power series in $z$ whose coefficients are holomorphic functions of $q$ for $\abs{q}<1$.  Explicitly, we have
	\begin{subequations} \label{eqn:relchar}
		\begin{align}
			\fch{\NSRel{\lambda}}{y,z,q} &= \frac{y^{-5/4} z^{\lambda}}{\sqrt{2}} \taufn{2}{q} \sum_{n \in \ZZ} z^n & &\text{(\(\lambda \neq \pm\frac{1}{2} \bmod{2}\)),} \label{eqn:nsrelchar} \\
			\fch{\RRel{\mu}}{y,z,q} &= \frac{y^{-5/4} z^{\mu}}{2} \brac*{\taufn{3}{q} \sum_{n \in \ZZ} z^n + \taufn{4}{q} \sum_{n \in \ZZ} (-z)^n} & &\text{(\(\mu \neq \pm \frac{5}{4} \bmod{2}\)).} \label{eqn:rrelchar}
		\end{align}
	\end{subequations}
\end{proposition}

\noindent We will not provide a proof of this assertion here, referring instead to recent work of Adamovi\'{c} \cite[Sec.~11]{AdaRea17} who proves \eqref{eqn:nsrelchar} using an explicit construction of these modules.  A more general proof, that includes the Ramond case \eqref{eqn:rrelchar}, appears in \cite{KawRel18}.  This approach works in the setting of affine Lie superalgebras and relies on Mathieu's theory of coherent families \cite{MatCla00}.  We expect it to also generalise to higher rank superalgebras.

\begin{remark}
	We note that the Ramond characters \eqref{eqn:rrelchar} are manifestly $2$-periodic in $\mu$, in accord with the isomorphisms $\RRel{\mu} \cong \RRel{\mu+2}$.  However, the \ns{} characters \eqref{eqn:nsrelchar} are manifestly $1$-periodic, despite $\NSRel{\lambda}$ and $\NSRel{\lambda+1}$ being non-isomorphic.  Rather, we have $\NSRel{\lambda} \cong \parr{\NSRel{\lambda+1}}$, by \eqref{eqn:modperiodicity}, explaining this $1$-periodicity.
\end{remark}

\subsection{Supercharacters} \label{subsec:supchars}

As mentioned above, modular transformations require us to consider supercharacters in addition to characters.  The \emph{supercharacter} of a weight module $\Mod{M}$ of $\myaosp$ is defined, analogously to the definition of character in \eqref{eqn:defchar}, to be the formal power series
\begin{equation} \label{eqn:defschar}
\fsch{\Mod{M}}{y,z,q} = \straceover{\Mod{M}} \brac*{y^k z^{h_0} q^{L_0 - c/24}}
\end{equation}
in which $\straceover{\Mod{M}} \brac*{X}$ denotes the supertrace of the image in $\End(\Mod{M})$ of $X \in \aalgh$.  This is, we recall, the trace of $X$ on the even subspace $\Mod{M}^{(0)}$ minus the trace on the odd subspace $\Mod{M}^{(1)}$.  The use of the subscript $\frac{1}{2}$ turns out to be convenient in \cref{sec:modver}, see \cref{eq:chschstandard} for example.  We clearly have
\begin{equation} \label{eqn:schpar}
\sch{\parrmod{\Mod{M}}} = -\sch{\Mod{M}}
\end{equation}
and, similarly to characters, supercharacters of twisted modules are easily computed using
\begin{equation} \label{eqn:autsch}
\fsch{\conjmod{\Mod{M}}}{y,z,q} = \fsch{\Mod{M}}{y,z^{-1},q}, \quad
\fsch{\sfmod{\ell}{\Mod{M}}}{y,z,q} = \fsch{\Mod{M}}{y z^{2\ell} q^{\ell^2},z q^{\ell},q},
\end{equation}
the latter holding for all $\ell \in \frac{1}{2} \ZZ$.

In fact, it is easy to deduce formulae for $\myaosp$-module supercharacters from their character formulae because $x$ and $y$ have odd $\myosp$-weights (measured in units of $\alpha$) while $e$, $h$ and $f$ have even $\myosp$-weights.  The vectors of an $\myaosp$-module $\Mod{M}$ whose $\myosp$-weights are even, relative to some chosen (even) vector, therefore span $\Mod{M}^{(0)}$ while those with relative odd $\myosp$-weights span $\Mod{M}^{(1)}$.  Replacing $z$ by $\ee^{\ii \pi} z = -z$ everywhere in $\ch{\Mod{M}}$, except in the prefactor that fixes the $\myosp$-weight of the chosen vector, will therefore convert the character of $\Mod{M}$ into its supercharacter.  The supercharacters of the relaxed \hw{} $\OSPMinMod{2}{4}$-modules are thereby deduced from \cref{prop:AtypCh,prop:AtypStCh,prop:TypCh}.
\begin{proposition} \label{prop:RelSch}
	The supercharacters of the simple $\OSPMinMod{2}{4}$-modules are specified, for $\abs{q} < 1$, by
	\begin{equation} \label{eqn:schars}
		\begin{aligned}
			\fsch{\NSFin{0}}{y,z,q} &= y^{-5/4} \frac{\fjth{2}{z;q^2} \; \fjth{1}{z;q}}{\fjth{1}{z^2 ;q} \; \deta{q}} & &\text{(\(\abs{q}^{1/2} < \abs{z} < \abs{q}^{-1/2}\)),} \\
			\fsch{\NSInf{-1/2}^+}{y,z,q} &= -\ii y^{-5/4} \frac{\fjth{3}{z;q^2} \; \fjth{1}{z;q}}{\fjth{1}{z^2 ;q} \; \deta{q}} & &\text{(\(1 < \abs{z} < \abs{q}^{-1/2}\)),} \\
			\fsch{\NSRel{\lambda}}{y,z,q} &= \frac{y^{-5/4} z^{\lambda}}{\sqrt{2}} \taufn{2}{q} \sum_{n \in \ZZ} (-z)^n & &\text{(\(\lambda \neq \pm\frac{1}{2} \bmod{2}\)),} \\
			\fsch{\RRel{\mu}}{y,z,q} &= \frac{y^{-5/4} z^{\mu}}{2} \brac*{\taufn{4}{q} \sum_{n \in \ZZ} z^n + \taufn{3}{q} \sum_{n \in \ZZ} (-z)^n} & &\text{(\(\mu \neq \pm \frac{5}{4} \bmod{2}\))}
		\end{aligned}
	\end{equation}
	and \eqref{eqn:autsch} (conjugation and spectral flow).  Similarly, the supercharacters of the atypical standard $\OSPMinMod{2}{4}$-modules are specified by
	\begin{equation} \label{eq:AtypStSChars}
		\begin{aligned}
			\fsch{\NSRel{\pm 1/2}^{\pm}}{y,z,q} &= \frac{y^{-5/4} z^{1/2}}{\sqrt{2}} \taufn{2}{q} \sum_{n \in \ZZ} (-z)^n, \\
			\fsch{\RRel{\pm 5/4}^{\pm}}{y,z,q} = \fsch{\RRel{\mp 3/4}^{\mp}}{y,z,q} &= \frac{y^{-5/4} z^{\pm 5/4}}{2} \brac*{\taufn{4}{q} \sum_{n \in \ZZ} z^n + \taufn{3}{q} \sum_{n \in \ZZ} (-z)^n}.
		\end{aligned}
	\end{equation}
\end{proposition}

\begin{remark}
	Note that $\sch{\NSRel{\lambda}}$ is $1$-antiperiodic in $\lambda$ and $\sch{\RRel{\mu}}$ is $2$-periodic in $\mu$, as expected.
\end{remark}

\subsection{The Grothendieck group} \label{sub:Gr}

Given a category $\categ{C}$ of modules, each of which has finitely many composition factors, one defines its Grothendieck group $\tGr{\categ{C}}$ as the free abelian group generated by the isomorphism classes of the simple modules.  The image $\tGr{\Mod{M}} \in \tGr{\categ{C}}$ of a module $\Mod{M} \in \categ{C}$ is then the sum of the images of its composition factors.  As characters and supercharacters are isomorphism invariants that do not distinguish between a module and the direct sum of its composition factors, they define functions $\chmap$ and $\schmap$ on the Grothendieck group.

Neither of these functions is injective --- in general, $\chmap$ and $\schmap$ do not distinguish $\tGr{\Mod{M}}$ from $\tGr{\parrmod{\Mod{M}}}$ and $-\tGr{\parrmod{\Mod{M}}}$, respectively --- but their direct sum is, provided that the set of pairs $(\tch{\Mod{M}}, \tsch{\Mod{M}})$, where $\tGr{\Mod{M}}$ ranges over the isomorphism classes of the simple modules, is linearly independent (over $\ZZ$).  If this is the case, which it is for the physical category $\categ{P}$ of $\OSPMinMod{2}{4}$-modules, then it follows that the character and supercharacter of a module together completely determine the class of the module in the Grothendieck group.

Whilst the set of isomorphism classes of the simple $\OSPMinMod{2}{4}$-modules should be regarded as a canonical basis of the Grothendieck group, the standard module formalism suggests that another basis set will be useful.  This is the set of Grothendieck images of the standard modules, both typical and atypical, noting that
\begin{subequations}
	\begin{equation}
		\Gr{\sfmod{\ell}{\NSRel{1/2}^+}} = \Gr{\parr\sfmod{\ell}{\NSInf{-1/2}^+}} + \Gr{\sfmod{\ell}{\NSInf{1/2}^-}} = \Gr{\parr\sfmod{\ell}{\NSRel{-1/2}^-}}
	\end{equation}
	in the \ns{} sector, whilst in the Ramond sector we similarly have
	\begin{equation}
		\Gr{\sfmod{\ell}{\RRel{3/4}^+}} = \Gr{\sfmod{\ell}{\RRel{-5/4}^-}}, \quad
		\Gr{\sfmod{\ell}{\RRel{5/4}^+}} = \Gr{\sfmod{\ell}{\RRel{-3/4}^-}}
	\end{equation}
\end{subequations}
and their parity-reversed counterparts.  To maintain linear independence among the images of the atypical standards, we may consistently omit the Grothendieck images of standards with $-$ superscripts in favour of those with $+$ superscripts.  There is, of course, no issue with linear independence among typicals.

We assert that the Grothendieck images of the typicals and the atypical standards with $+$ superscripts together form a basis of a certain completion of the Grothendieck group of our $\OSPMinMod{2}{4}$-module category $\categ{P}$.  To justify this, we must exhibit the images of the simple atypicals as a linear combination of images of atypical standards.  This is achieved by constructing resolutions for the former in terms of the latter.  It clearly suffices to give the resolutions for one member of each spectral flow orbit of simple atypicals.

\begin{prop}\label{thm:resolutions}
  The simple modules \(\NSFin{0}\) and \(\NSInf{\pm1/2}^{\mp}\) have the following resolutions:
  \begin{subequations}
    \begin{align}
      \cdots\lra\parr\sfmod{11/2}{\RRel{3/4}^+}\lra \sfmod{5}{\NSRel{1/2}^+}&\lra\sfmod{9/2}{\RRel{5/4}^+} \notag \\
			\lra\sfmod{7/2}{\RRel{3/4}^+}&\lra \parr\sfmod{3}{\NSRel{1/2}^+}\lra\parr\sfmod{5/2}{\RRel{5/4}^+} \notag \\
      &\lra\parr\sfmod{3/2}{\RRel{3/4}^+}\lra \sfmod{}{\NSRel{1/2}^+}\lra\sfmod{1/2}{\RRel{5/4}^+}\lra \NSFin{0}\lra0, \\
			\cdots\lra \parr\sfmod{11/2}{\RRel{5/4}^+}\lra\parr\sfmod{9/2}{\RRel{3/4}^+}&\lra \sfmod{4}{\NSRel{1/2}^+} \notag \\
			\lra\sfmod{7/2}{\RRel{5/4}^+}&\lra\sfmod{5/2}{\RRel{3/4}^+}\lra \parr\sfmod{2}{\NSRel{1/2}^+} \notag \\
			&\lra \parr\sfmod{3/2}{\RRel{5/4}^+}\lra\parr\sfmod{1/2}{\RRel{3/4}^+}\lra \NSRel{1/2}^+\lra \NSInf{1/2}^-\lra0, \\
			\cdots\lra\sfmod{6}{\NSRel{1/2}^+} \lra\sfmod{11/2}{\RRel{5/4}^+}&\lra\sfmod{9/2}{\RRel{3/4}^+} \notag \\
			\lra \parr\sfmod{4}{\NSRel{1/2}^+}&\lra\parr\sfmod{7/2}{\RRel{5/4}^+}\lra\parr\sfmod{5/2}{\RRel{3/4}^+} \notag \\
			&\lra \sfmod{2}{\NSRel{1/2}^+}\lra\sfmod{3/2}{\RRel{5/4}^+}\lra\sfmod{1/2}{\RRel{3/4}^+}\lra \NSInf{-1/2}^+\lra0.
    \end{align}
  \end{subequations}
\end{prop}
\begin{proof}
  The resolutions are constructed by repeatedly splicing the following short exact sequences, their parity reversals and their spectrally flowed versions:
	\begin{equation}
	  \begin{gathered}
		  \dses{\sfmod{}{\NSInf{1/2}^-}}{}{\sfmod{1/2}{\RRel{5/4}^+}}{}{\NSFin{0}},\\
	    \dses{\parrmod{\NSInf{-1/2}^+}}{}{\NSRel{1/2}^+}{}{\NSInf{1/2}^-},\\
	    \dses{\sfmod{}{\NSFin{0}}}{}{\sfmod{1/2}{\RRel{3/4}^+}}{}{\NSInf{-1/2}^+}.
	  \end{gathered}
	\end{equation}
  These, in turn, are easily derived from the short exact sequences \eqref{eqn:sesindec} and the identifications \eqref{eqn:sfA0}.
\end{proof}

The Euler-Poincar\'e principle now leads to the desired expressions for the images of the simple atypicals in an appropriate completion of the Grothendieck group. Such a completion must admit the required infinite sums and so we shall choose it to consist of (possibly infinite) linear combinations of the form
\begin{equation} \label{eq:thecompletion}
	\sum_{i=1}^n \sum_{\ell \ge L} a_L^i \Gr{\sfmod{\ell}{\Mod{S}_i}} \quad \text{($n \in \ZZ_{\ge 0}$, $L \in \tfrac{1}{2} \ZZ$, $a_L^i \in \ZZ$),}
\end{equation}
where each $\Mod{S}_i$ is either a typical or an atypical standard with $+$
superscript.  The key requirement is that the spectral flow indices be bounded
below.  The Grothendieck images of the typicals and atypical standards with
$+$ superscripts form a basis of this completion, in the obvious sense, as
claimed.

\begin{remark}
The above completion allows for the following interpretation. Let \(\ZZ\brac*{\brac*{\sfaut^{1/2}}}\) be the ring of formal Laurent series in the spectral flow functor $\sfaut^{1/2}$. \footnote{Recall that in rings of formal Laurent series, the powers of the indeterminate are assumed to be bounded below.}
    These series can then be applied to \emph{finite} sums of (Grothendieck images of) standards to get infinite sums
    of the form \eqref{eq:thecompletion}.
    In other words, the set of all elements of the form \eqref{eq:thecompletion} is just the free
    \(\ZZ\brac*{\brac*{\sfaut^{1/2}}}\)-module whose basis consists
    of all typicals and all atypical standards with $+$ superscript.
\end{remark}

\begin{cor}\label{cor:resgr}
  The following identities hold in the given completion of the Grothendieck group (as do their spectrally flowed versions):
  \begin{subequations}\label{eqn:resgr}
	  \begin{align}
	    \Gr{\NSFin{0}}
	    &= \sum_{n=0}^\infty \biggl( \Gr{\sfmod{4n+1/2}{\RRel{5/4}^+}} - \Gr{\sfmod{4n+1}{\NSRel{1/2}^+}} + \Gr{\parr\sfmod{4n+3/2}{\RRel{3/4}^+}} \biggr. \notag \\
	    &\mspace{100mu} \biggl. - \Gr{\parr\sfmod{4n+5/2}{\RRel{5/4}^+}} + \Gr{\parr\sfmod{4n+3}{\NSRel{1/2}^+}} - \Gr{\sfmod{4n+7/2}{\RRel{3/4}^+}} \biggr), \label{eqn:resgra} \\
	    \Gr{\NSInf{1/2}^-}
	    &= \sum_{n=0}^\infty \biggl( \Gr{\sfmod{4n}{\NSRel{1/2}^+}} - \Gr{\parr\sfmod{4n+1/2}{\RRel{3/4}^+}} + \Gr{\parr\sfmod{4n+3/2}{\RRel{5/4}^+}} \biggr. \notag \\
	    &\mspace{100mu} \biggl. - \Gr{\parr\sfmod{4n+2}{\NSRel{1/2}^+}} + \Gr{\sfmod{4n+5/2}{\RRel{3/4}^+}} - \Gr{\sfmod{4n+7/2}{\RRel{5/4}^+}} \biggr), \label{eq:resgrb-} \\
	    \Gr{\NSInf{-1/2}^+}
	    &= \sum_{n=0}^\infty \biggl( \Gr{\sfmod{4n+1/2}{\RRel{3/4}^+}} - \Gr{\sfmod{4n+3/2}{\RRel{5/4}^+}} + \Gr{\sfmod{4n+2}{\NSRel{1/2}^+}} \biggr. \notag \\
	    &\mspace{100mu} - \Gr{\parr\sfmod{4n+5/2}{\RRel{3/4}^+}} + \Gr{\parr\sfmod{4n+7/2}{\RRel{5/4}^+}} - \Gr{\parr\sfmod{4n+4}{\NSRel{1/2}^+}} \biggr). \label{eq:resgrb+}
	  \end{align}
	\end{subequations}
	Applying $\chmap$ and $\schmap$ to these identities then gives analogous identities for simple atypical (super)characters in terms of atypical standard (super)characters.
\end{cor}

It is easy to check that the infinite sums of (super)characters in these relations converge in the sense of formal power series.  This means that the coefficient of each monomial in the indeterminates only receives non-zero contributions from finitely many of the terms in the infinite sum.

\begin{remark}
	Note that one can also derive right-sided resolutions for the simple atypicals, for example by taking contragredient duals of the left-sided ones of \cref{thm:resolutions}.  The powers of spectral flow appearing in these resolutions and the analogue of \cref{cor:resgr} would then be bounded above.

    The natural completion of the Grothendieck group would now be the free \(\ZZ\brac*{\brac*{\sfaut^{-1/2}}}\)-module on the typicals and atypical standards with $+$ superscript. One can check that the results derived in what follows do not depend on the choice of resolutions.
\end{remark}

\section{The affine minimal model $\OSPMinMod{2}{4}$: Modular properties} \label{sec:modver}

We now study the modular properties of the simple and standard \(\OSPMinMod{2}{4}\)-modules.  More precisely, we first show that the action of the modular group closes on the span of the characters and supercharacters of the standard modules.  From \cref{sub:Gr}, the characters and supercharacters of the simple atypical \(\OSPMinMod{2}{4}\)-modules belong to this span, so their modular properties follow directly.  These explicit modular transformations will be used to compute the Grothendieck fusion rules in \cref{sec:fusion} using a version of the standard Verlinde formula.

\subsection{Standard modular transforms} \label{subsec:stmod}

The modular group \(\SLG{SL}{2;\ZZ}\) admits the following well known presentation:
\begin{equation} \label{eqn:modpres}
  \SLG{SL}{2;\ZZ}=\left\langle \modS,\modT\st \modS^2=(\modS\modT)^3,\ \modS^4=1\right\rangle.
\end{equation}
On (super)characters, we define the action of the generating elements $\modS$
and $\modT$ to be the coordinate transforms
\begin{equation}
	\modS \brac*{\psi \vert \zeta \vert \tau} = \brac*{\psi - \frac{\zeta^2}{\tau} + \frac{1}{2 \pi k} \brac*{\arg{\tau} - \frac{\pi}{2}} \relmiddle{\vert} \frac{\zeta}{\tau} \relmiddle{\vert} \frac{-1}{\tau}}, \quad
	\modT \brac*{\psi \vert \zeta \vert \tau} = \brac*{\psi + \frac{1}{12k} \relmiddle{\vert} \zeta \relmiddle{\vert} \tau + 1},
\end{equation}
where $k=-\frac{5}{4}$ (and we recall \cref{eqn:modvars} for the definition of $\psi$, $\zeta$ and $\tau$).  Though the transform for $\psi$ may appear unfamiliar, it is chosen specifically to obtain an honest, as opposed to a projective, representation of $\SLG{SL}{2;\ZZ}$ on the span of the standard characters.  Direct computation verifies that these generators do indeed satisfy the
defining relations \eqref{eqn:modpres} of \(\SLG{SL}{2;\ZZ}\).

The modular transformations of the standard characters can be computed from those of the auxiliary functions \(\taufn{2}{q}\), \(\taufn{3}{q}\) and \(\taufn{4}{q}\) defined in \eqref{eqn:defA}:
\begin{equation} \label{eqn:Ajmod}
  \begin{aligned}
	  \modch{S}{\taufn{2}{q}} &= \frac{1}{- \ii \tau} \taufn{4}{q},
	  &\modch{S}{\taufn{3}{q}} &= \frac{1}{- \ii \tau} \taufn{3}{q},
	  &\modch{S}{\taufn{4}{q}} &= \frac{1}{- \ii \tau} \taufn{2}{q}, \\
	  \modch{T}{\taufn{2}{q}} &= \ee^{-\pi \ii / 12} \taufn{2}{q},
	  &\modch{T}{\taufn{3}{q}} &= \ee^{-5 \pi \ii / 24} \taufn{4}{q},
	  &\modch{T}{\taufn{4}{q}} &= \ee^{-5 \pi \ii / 24} \taufn{3}{q}.
	\end{aligned}
\end{equation}
These formulae follow directly from the well known transforms of the Jacobi theta functions and the Dedekind eta function.

In preparation for the modular S-transforms below, we rewrite the formulae for the standard (super)characters using the Fourier theoretic identity \eqref{eqn:deltasum}.  For example,
\begin{align}
	\ch{\NSRel{\lambda}} &= y^k z^{\lambda}
	\frac{\taufn{2}{q}}{\sqrt{2}} \sum_{n \in \ZZ} z^n
	=y^k\frac{\taufn{2}{q}}{\sqrt{2}}\sum_{n \in
		\ZZ}z^{\lambda}\delta(\zeta-n)
	=y^k\frac{\taufn{2}{q}}{\sqrt{2}}\sum_{n \in
		\ZZ}\ee^{2\pi\ii\lambda n}\delta(\zeta-n)
\end{align}
implies, using \eqref{eqn:autch}, that
\begin{subequations}
  \begin{align}
    \ch{\sfmod{\ell}{\NSRel{\lambda}}}
    &=(yz^{2\ell}q^{\ell^2})^k\frac{\taufn{2}{q}}{\sqrt{2}}\sum_{n \in \ZZ}\ee^{2\pi\ii\lambda n}\delta(\zeta+\ell\tau-n)\nonumber\\
    &=y^k\frac{\taufn{2}{q}}{\sqrt{2}}
    \sum_{n \in \ZZ}\ee^{2\pi\ii\lambda n}\ee^{2\pi\ii k\ell(2 n-\ell\tau)}\delta(\zeta+\ell\tau-n).
  \intertext{Similarly,}
    \sch{\sfmod{\ell}{\NSRel{\lambda}}}
    &=y^k\frac{\taufn{2}{q}}{\sqrt{2}}\sum_{n \in\ZZ}\ee^{2\pi\ii\lambda (n-1/2)}\ee^{2\pi\ii k\ell(2 n-1-\ell\tau)}
    \delta(\zeta+\ell\tau+\tfrac{1}{2}-n),\\
    \ch{\sfmod{\ell}{\RRel{\mu}}}
    &=y^k \left(\frac{\taufn{3}{q}}{2} \sum_{n \in \ZZ} \ee^{2\pi\ii\mu n}\ee^{2\pi\ii k\ell(2 n-\ell\tau)} \delta(\zeta+\ell\tau-n)\right.\nonumber\\
    &\mspace{50mu}\left.+\frac{\taufn{4}{q}}{2} \sum_{n \in \ZZ} \ee^{2\pi\ii\mu(n-1/2)}\ee^{2\pi\ii k\ell(2 n-1-\ell\tau)} \ddelta{\zeta+\ell\tau+\tfrac{1}{2}-n}\right),\\
    \sch{\sfmod{\ell}{\RRel{\mu}}}
    &=y^k \left(\frac{\taufn{4}{q}}{2} \sum_{n \in \ZZ} \ee^{2\pi\ii\mu n}\ee^{2\pi\ii k\ell(2 n-\ell\tau)} \delta(\zeta+\ell\tau-n)\right.\notag\\
    &\mspace{50mu}\left.+\frac{\taufn{3}{q}}{2} \sum_{n \in \ZZ} \ee^{2\pi\ii\mu(n-1/2)}\ee^{2\pi\ii k\ell(2n-1-\ell\tau)}\delta(\zeta+\ell\tau+\tfrac{1}{2}-n)\right).
  \end{align}
\end{subequations}
These formulae also cover the (super)characters of the atypical standard modules.  Note that we can write these formulae in the succinct forms
\begin{subequations} \label{eq:chschstandard}
	\begin{align}
		\nch{\sfmod{\ell}{\NSRel{\lambda}}}{\eps} &= y^k \frac{\taufn{2}{q}}{\sqrt{2}} \sum_{n \in \ZZ + \eps} \ee^{2 \pi \ii \lambda n} \ee^{2 \pi \ii k \ell (2n - \ell \tau)} \delta(\zeta + \ell \tau - n), \\
		\nch{\sfmod{\ell}{\RRel{\lambda}}}{\eps} &= y^k \brac*{\frac{\taufn{3}{q}}{2} \sum_{n \in \ZZ + \eps} + \frac{\taufn{4}{q}}{2} \sum_{n \in \ZZ + 1/2 + \eps}} \, \ee^{2 \pi \ii \lambda n} \ee^{2 \pi \ii k \ell (2n - \ell \tau)} \delta(\zeta + \ell \tau - n),
	\end{align}
\end{subequations}
where $\eps \in \set{0, \frac{1}{2}}$.  It will be useful in what follows to think of $\eps$ as living in the two element abelian group where addition is taken mod $1$.

\begin{prop}\label{thm:Ttransf}
  For $\ell\in\frac{1}{2}\ZZ$, and $\lambda\in\CC$, let
  \begin{align}
    \tker{\ell}{\lambda}=\ee^{2\pi\ii\ell(k\ell+\lambda)}.
  \end{align}
  Then, the T-transforms of the standard (super)characters are given by the following formulae:
  \begin{subequations}\label{eqn:Ttransf}
	  \begin{align}
			\modch{T}{\nch{\sfmod{\ell}{\NSRel{\lambda}}}{\eps}} &= \ee^{\pi \ii /12} \tker{\ell}{\lambda} \, \nch{\sfmod{\ell}{\NSRel{\lambda}}}{\eps + \ell}, \\
			\modch{T}{\nch{\sfmod{\ell}{\RRel{\lambda}}}{\eps}} &= \ee^{-\pi \ii /24} \tker{\ell}{\lambda} \, \nch{\sfmod{\ell}{\RRel{\lambda}}}{\eps + \ell + 1/2}.
		\end{align}
	\end{subequations}
\end{prop}
\begin{proof}
  The proof follows by substituting the appropriate transformation formula from \eqref{eqn:Ajmod} into the (super)character formulae \eqref{eqn:relchar}.  For example, the \(\NSRel{}\)-type (super)characters give
  \begin{align} \label{eqn:deriveTchC}
      \modch{T}{\nch{\sfmod{\ell}{\NSRel{\lambda}}}{\eps}}
      &=\ee^{\pi\ii/6}y^k\ee^{-\pi\ii/12}\frac{\taufn{2}{q}}{\sqrt{2}}
      \sum_{n \in \ZZ+\eps} \ee^{2\pi\ii\lambda n} \ee^{2 \pi \ii k\ell \brac*{2n-\ell\tau-\ell}} \ddelta{\zeta + \ell \tau+\ell - n}  \nonumber\\
      &=\ee^{\pi\ii/12}y^k\frac{\taufn{2}{q}}{\sqrt{2}}
      \sum_{n \in \ZZ+\eps+\ell} \ee^{2\pi\ii\lambda (n+\ell)}\ee^{2 \pi \ii k\ell \brac*{2n-\ell\tau+\ell}}
      \ddelta{\zeta + \ell \tau - n}  \nonumber\\
      &= \ee^{\pi \ii /12} \tker{\ell}{\lambda} \, \nch{\sfmod{\ell}{\NSRel{\lambda}}}{\eps + \ell}.
  \end{align}
  The \(\RRel{}\)-type case proceeds in the same way.
\end{proof}

We conclude that $\modT$ acts diagonally on \ns{} (super)characters, whilst in the Ramond sector, it takes a character to (a multiple of) the supercharacter of the same module and vice versa.  The behaviour of the (super)characters under the S-transform is similar, though the explicit formulae are slightly more complicated.

\begin{prop}\label{thm:Stransf}
  For $\ell,\ell'\in\frac{1}{2}\ZZ$ and $\lambda,\lambda'\in\CC$, let
  \begin{equation} \label{eqn:Sker}
    \sker{\ell,\ell'}{\lambda,\lambda'} = \ee^{-2 \pi \ii
    \brac*{2k \ell \ell' + \ell' \lambda + \ell
    \lambda'}}.
  \end{equation}
  Then, the \(\modS\)-transforms of the standard (super)characters  are given by the following formulae:
  \begin{subequations}\label{eqn:Stransf}
		\begin{align}
			\modch{S}{\nch{\sfmod{\ell}{\NSRel{\lambda}}}{\eps}} &= \frac{1}{\sqrt{2}} \sum_{\ell' \in \ZZ + \eps} \int_{\RR / 2 \ZZ} \sker{\ell,\ell'}{\lambda,\lambda'} \, \nch{\sfmod{\ell'}{\RRel{\lambda'}}}{\ell + 1/2} \, \dd \lambda', \\
			\modch{S}{\nch{\sfmod{\ell}{\RRel{\lambda}}}{\eps}} &= \frac{1}{\sqrt{2}} \sum_{\ell' \in \ZZ+\eps+1/2} \int_{\RR / \ZZ} \sker{\ell,\ell'}{\lambda,\lambda'} \, \nch{\sfmod{\ell'}{\NSRel{\lambda'}}}{\ell} \, \dd \lambda' \notag \\
			&\mspace{50mu} + \frac{1}{2} \sum_{\ell' \in \ZZ + \eps} \int_{\RR / 2 \ZZ} \sker{\ell,\ell'}{\lambda,\lambda'} \, \nch{\sfmod{\ell'}{\RRel{\lambda'}}}{\ell} \, \dd \lambda'.
		\end{align}
	\end{subequations}
\end{prop}
\begin{proof}
  As a preparatory step, we record the S-transform of
  the common prefactor \(y^k\) of the (super)character formulae:
  \begin{align}
    \modch{S}{y^k}=
    \ee^{2 \pi \ii k \brac*{\psi - \zeta^2/\tau}}
    \ee^{\ii \brac*{\arg{\tau} - \pi/2}} = \frac{-\ii \tau}{\abs{\tau}}y^k
    \ee^{2 \pi \ii k \brac*{- \zeta^{2}/\tau}}.
  \end{align}
  The S-transforms of the \lcnamecref{thm:Stransf} now follow by
  evaluating the left- and right-hand sides and comparing. For example,
  \begin{align}
    \modch{S}{\nch{\sfmod{\ell}{\NSRel{\lambda}}}{\eps}} &=
     \frac{-\ii \tau}{\abs{\tau}}y^k \ee^{2 \pi \ii k \brac*{-\zeta^{2}/\tau}}
     \frac{1}{-\ii\tau} \frac{\taufn{4}{q}}{\sqrt{2}}
    \sum_{n \in \ZZ + \eps} \ee^{2 \pi \ii \lambda n} \ee^{2\pi\ii k\ell(2n + \ell/\tau)}
    \ddelta{\frac{\zeta - \ell}{\tau}-n} \nonumber\\
    &=y^k \frac{\taufn{4}{q}}{\sqrt{2}}
    \sum_{n \in \ZZ + \eps} \ee^{2 \pi \ii n \brac*{\lambda - k n \tau}}
    \ddelta{\zeta - \ell - n \tau},
  \end{align}
  where we have used the identity $\abs{a}\ddelta{a z} = \ddelta{z}$, while
  \begin{align}
    &\frac{1}{\sqrt{2}} \sum_{\ell' \in \ZZ + \eps} \int_{\RR / 2 \ZZ}
    \ee^{-2 \pi \ii \brac*{2k \ell \ell' + \ell' \lambda + \ell \lambda'}}
    \nch{\sfmod{\ell'}{\RRel{\lambda'}}}{\ell + 1/2} \,\dd \lambda'\nonumber\\
    &\quad=\frac{1}{\sqrt{2}} \sum_{\ell' \in \ZZ + \eps}
    \ee^{-2 \pi \ii \ell' \brac*{2k \ell + \lambda}}y^k\int_{\RR/2\ZZ}
    \left(\frac{\taufn{3}{q}}{2} \sum_{n \in \ZZ + \ell + 1/2}
    \ee^{2\pi\ii\lambda' (n-\ell)} \ee^{2\pi\ii k\ell'(2n-\ell'\tau)}
    \delta(\zeta+\ell'\tau-n)\right.\nonumber\\
    &\mspace{300mu}\left.+\frac{\taufn{4}{q}}{2} \sum_{n \in \ZZ + \ell}
      \ee^{2\pi\ii\lambda' (n-\ell)}\ee^{2\pi\ii k\ell'(2 n-\ell'\tau)}
      \delta(\zeta+\ell'\tau-n)\right)\dd\lambda'\nonumber\\
    &\quad=
    y^k \frac{\taufn{4}{q} }{\sqrt{2}}
    \sum_{\ell' \in \ZZ + \eps}  \ee^{-2 \pi \ii \ell' \brac*{\lambda + k \ell' \tau}}
    \ddelta{\zeta - \ell + \ell' \tau},
  \end{align}
  where we have used the identity \(\int_{\RR / 2 \ZZ} \ee^{2\pi\ii \lambda' m}\dd\lambda'=2\delta_{m,0}\) for \(m\in\frac{1}{2}\ZZ\).  Setting $\ell' = -n$ completes this case.  The Ramond case is similar.
\end{proof}

By a slight abuse of nomenclature, we shall from here on refer to the
integration kernels appearing in the integrals \eqref{eqn:Stransf}
as S-matrix coefficients. The coefficient of the (super)character of a module
\(\Mod{M}'\) appearing in the S-transform of the (super)character of a module
\(\Mod{M}\) will be denoted by \(\tSmat{\eps}{\Mod{M}}{\eps'}{\Mod{M}'}\), where $\eps$ and $\eps'$ take values $0$ or $\frac{1}{2}$ to refer to the character or supercharacter of $\Mod{M}$ and $\Mod{M}'$, respectively. With this convention, we may rewrite \eqref{eqn:Stransf} in the form
\begin{equation} \label{eqn:SMatrixEntries}
	\begin{aligned}
    \Smat{\eps}{\sfmod{\ell}{\NSRel{\lambda}}}{\eps'}{\sfmod{\ell'}{\NSRel{\lambda'}}} &= 0, &
	  \Smat{\eps}{\sfmod{\ell}{\RRel{\lambda}}}{\eps'}{\sfmod{\ell'}{\NSRel{\lambda'}}} &= \frac{1}{\sqrt{2}} \sker{\ell,\ell'}{\lambda, \lambda'} \wun_{\ell = \eps'} \wun_{\ell' \neq \eps}, \\
	  \Smat{\eps}{\sfmod{\ell}{\NSRel{\lambda}}}{\eps'}{\sfmod{\ell'}{\RRel{\lambda'}}} &= \frac{1}{\sqrt{2}} \sker{\ell,\ell'}{\lambda, \lambda'} \wun_{\ell \neq \eps'} \wun_{\ell' = \eps}, &
	  \Smat{\eps}{\sfmod{\ell}{\RRel{\lambda}}}{\eps'}{\sfmod{\ell'}{\RRel{\lambda'}}} &= \frac{1}{2} \sker{\ell,\ell'}{\lambda, \lambda'} \wun_{\ell = \eps'} \wun_{\ell' = \eps},
	\end{aligned}
\end{equation}
where the condition in the index of the characteristic function $\wun_{\bullet}$ (which is $1$ when $\bullet$ is true and $0$ otherwise) is always understood to be taken mod $1$ whenever it involves an $\eps$.

This notation allows us to define the transpose of the S-matrix elements in the obvious way:
\begin{equation}
	\Smat{\eps}{\Mod{M}}{\eps'}{\Mod{M}'}^t = \Smat{\eps'}{\Mod{M}'}{\eps}{\Mod{M}}.
\end{equation}
The adjoint $\dagger$ is then defined to be the (complex) conjugate transpose, as usual.  We now see that with respect to the standard basis of characters and supercharacters,
\begin{equation} \label{eqn:basis}
	\set*{\nch{\sfmod{\ell}{\NSRel{\lambda}}}{\eps},\ \nch{\sfmod{\ell}{\RRel{\mu}}}{\eps} \st \eps \in \set{0, \tfrac{1}{2}},\ \ell \in \tfrac{1}{2} \ZZ,\ \lambda \in \CC/\ZZ,\ \mu \in \CC/2\ZZ},
\end{equation}
the S-matrix entries of \cref{eqn:SMatrixEntries} are manifestly symmetric with respect to transposition.  One can also check that the S- and T-transforms are unitary in the sense that \(\modT^\dagger\modT=\wun=\modS^\dagger\modS\).  For example, the $(\NSRel{} \ra \NSRel{})$-type entries of $\modS^\dagger\modS$ are (because $\tSmat{\eps}{\sfmod{\ell}{\NSRel{\lambda}}}{\eps''}{\sfmod{\ell''}{\NSRel{\lambda''}}} = 0$)
\begin{align}
	&\sum_{\eps''} \sum_{\ell'' \in \frac{1}{2} \ZZ} \int_{\CC / 2 \ZZ} \Smat{\eps}{\sfmod{\ell}{\NSRel{\lambda}}}{\eps''}{\sfmod{\ell''}{\RRel{\lambda''}}}^* \Smat{\eps''}{\sfmod{\ell''}{\RRel{\lambda''}}}{\eps'}{\sfmod{\ell'}{\NSRel{\lambda'}}} \, \dd \lambda'' \notag \\
	&\mspace{50mu} = \frac{1}{2} \sum_{\eps''} \sum_{\ell'' \in \frac{1}{2} \ZZ} \wun_{\ell \neq \eps''} \wun_{\ell'' = \eps} \wun_{\ell'' = \eps'} \wun_{\ell' \neq \eps''} \ee^{2 \pi \ii (2k(\ell-\ell')\ell'' + (\lambda-\lambda')\ell''} \int_{\CC / 2 \ZZ} \ee^{2 \pi \ii (\ell-\ell') \lambda''} \, \dd \lambda'' \notag \\
	&\mspace{50mu} = \wun_{\eps=\eps'} \wun_{\ell=\ell'} \sum_{\ell'' \in \ZZ + \eps} \ee^{2 \pi \ii (\lambda-\lambda')\ell''} = \delta_{\eps,\eps'} \delta_{\ell,\ell'} \delta(\lambda = \lambda' \bmod{1}) \ee^{2 \pi \ii \eps(\lambda - \lambda')}.
\end{align}
The result matches the corresponding entries of the identity operator because we should take $\lambda = \lambda'$ (choosing a fundamental domain for $\CC/\ZZ$ in the basis \eqref{eqn:basis}), so that $\ee^{2 \pi \ii \eps(\lambda - \lambda')} = 1$.  If we had insisted on taking $\lambda' = \lambda + 1$ instead, then we would have obtained an additional sign when $\eps = \frac{1}{2}$.  This sign is naturally explained by the fact that the \ns{} supercharacters satisfy $\tsch{\sfmod{\ell}{\NSRel{\lambda+1}}} = \tsch{\parr\sfmod{\ell}{\NSRel{\lambda}}} = -\tsch{\sfmod{\ell}{\NSRel{\lambda}}}$.

\subsection{Atypical modular transforms} \label{subsec:atypmod}

Given the modular S- and T-transforms of the standard modules, it is now easy to compute the corresponding transforms of a given atypical simple (super)character using \cref{cor:resgr}.  We could instead compute these transforms using the formulae of \cref{subsec:simpchars} which, after all, are already expressed in terms of theta functions.  However, these formulae only represent the correct character on an appropriate annulus in the $z$-plane, which becomes a horizontal strip in the $\zeta$-plane.  It is easy to see that the modular S-transform does not respect these strips, so one must meromorphically continue the character in $\zeta$ in order to get good modular properties.  Unfortunately, different strips correspond to different spectral flows of the atypical simple module.  The upshot is, as was first noted in \cite{RidSL208}, that meromorphically continuing the characters in $\zeta$ leads to a loss of linear independence.

The standard module formalism, expressing the atypical simple (super)characters as linear combinations of standard (super)characters, neatly sidesteps this linear independence problem by treating all (super)characters as formal power series in $z$ whose coefficients are meromorphic functions of $q$ (in the unit disc).  We therefore combine the results of \cref{cor:resgr,thm:Stransf} to deduce (extending the notation in an obvious way) the S-matrix entries involving atypicals.

\begin{proposition}
  The S-matrix entries describing the decomposition of the S-transforms of the atypical simple modules are
  \begin{equation} \label{eqn:AtypSMatrixEntries}
		\begin{aligned}
	    \Smat{\eps}{\sfmod{\ell}{\NSFin{0}}}{\eps'}{\sfmod{\ell'}{\NSRel{\lambda'}}}
	    &=
	    \frac{\sker{\ell,\ell'}{0,\lambda'}}{2\sqrt{2}\cos(\pi\lambda')}
	    \wun_{\ell\neq\eps'}\wun_{\ell'\neq\eps},\\
	    \Smat{\eps}{\sfmod{\ell}{\NSFin{0}}}{\eps'}{\sfmod{\ell'}{\RRel{\lambda'}}}
	    &=
	    \frac{\sker{\ell,\ell'}{0,\lambda'}}{4\cos(\pi\lambda')+2\sqrt{2}}
	    \wun_{\ell\neq\eps'}\wun_{\ell'=\eps},\\
	    \Smat{\eps}{\sfmod{\ell}{\NSInf{\mp1/2}^\pm}}{\eps'}{\sfmod{\ell'}{\NSRel{\lambda'}}}
	    &=
	    \frac{\sker{\ell,\ell'}{\mp \frac{1}{2},\lambda'}}{2\sqrt{2}\cos(\pi\lambda')}
	    \wun_{\ell\neq\eps'}\wun_{\ell'\neq\eps}, \\
	    \Smat{\eps}{\sfmod{\ell}{\NSInf{\mp1/2}}^\pm}{\eps'}{\sfmod{\ell'}{\RRel{\lambda'}}}
	    &=
	    \brac*{1+\sqrt{2}\ee^{\mp\pi\ii\lambda'}}
		  \frac{\sker{\ell,\ell'}{\mp \frac{1}{2},\lambda'}}{4\cos(\pi\lambda')+2\sqrt{2}}
	    \wun_{\ell\neq\eps'}\wun_{\ell'=\eps}.
  \end{aligned}
	\end{equation}
\end{proposition}
\begin{proof}
	The calculations are all very similar, so we only detail one example:
	\begin{align}
		\Smat{\eps}{\sfmod{\ell}{\NSFin{0}}}{\eps'}{\sfmod{\ell'}{\NSRel{\lambda'}}}
		&= \sum_{n=0}^{\infty} \left( \Smat{\eps}{\sfmod{\ell+4n+1/2}{\RRel{5/4}^+}}{\eps'}{\sfmod{\ell'}{\NSRel{\lambda'}}} + (-1)^{2\eps} \Smat{\eps}{\sfmod{\ell+4n+3/2}{\RRel{3/4}^+}}{\eps'}{\sfmod{\ell'}{\NSRel{\lambda'}}} \right. \notag \\
		&\mspace{50mu} \left. - (-1)^{2\eps} \Smat{\eps}{\sfmod{\ell+4n+5/2}{\RRel{5/4}^+}}{\eps'}{\sfmod{\ell'}{\NSRel{\lambda'}}} - \Smat{\eps}{\sfmod{\ell+4n+7/2}{\RRel{3/4}^+}}{\eps'}{\sfmod{\ell'}{\NSRel{\lambda'}}} \right) \notag \\
		&= \frac{\ee^{-2 \pi \ii (\ell + 2) (2k\ell' + \lambda')}}{\sqrt{2}} \sum_{n=0}^{\infty} \ee^{-8 \pi \ii \lambda' n} \cdot 2 \ii \brac*{(-1)^{2\ell'} \sin(3 \pi \lambda') + (-1)^{2\eps} \sin(\pi \lambda')} \wun_{\ell \neq \eps'} \wun_{\ell' \neq \eps} \notag \\
		&= \frac{\ee^{-2 \pi \ii \ell (2k\ell' + \lambda')}}{\sqrt{2}} \frac{\sin(3 \pi \lambda') - \sin(\pi \lambda')}{\sin (4 \pi \lambda')} \wun_{\ell \neq \eps'} \wun_{\ell' \neq \eps} \notag \\
		&= \frac{\ee^{-2 \pi \ii \ell (2k\ell' + \lambda')}}{\sqrt{2}} \frac{1}{2 \cos(\pi \lambda')} \wun_{\ell \neq \eps'} \wun_{\ell' \neq \eps}.
	\end{align}
	Note here that the denominator should be expanded in non-negative powers of $\ee^{-\pi \ii \lambda'}$.
\end{proof}

\section{The affine minimal model $\OSPMinMod{2}{4}$: Grothendieck fusion rules} \label{sec:fusion}

This final section applies the modular S-transforms, computed explicitly in the previous section, to the calculation of the Grothendieck fusion rules of $\OSPMinMod{2}{4}$.  This requires a version of the Verlinde formula and we begin by deducing such a formula, valid for \vosas{} like the one under consideration, from the standard Verlinde formula of its bosonic orbifold.  The latter is of course conjectural, but has (in this case) already passed one very stringent consistency test:  the coefficients computed by the standard Verlinde formula for the bosonic orbifold of $\OSPMinMod{2}{4}$ are all non-negative integers \cite{CreMod13}.

\subsection{A Verlinde formula for \(\OSPMinMod{2}{4}\)} \label{subsec:verlinde}

The purpose of this section is to determine the standard Verlinde formula for the $\OSPMinMod{2}{4}$ Grothendieck fusion ring.  This is a formula for the (super)characters of fusion products of $\OSPMinMod{2}{4}$-modules that uses only the modular properties of the (super)characters as input data.  As knowing the character and supercharacter of a module in the category $\categ{P}$ amounts to knowing the image of the module in the Grothendieck group, this formula allows us to compute fusion at the level of the Grothendieck group, thereby turning the Grothendieck group into a ring.  Of course, this ring structure will only be well defined if fusing with any module of $\categ{P}$ defines an exact endofunctor of $\categ{P}$.  We conjecture that this is indeed the case, though we have no proof.  The results obtained below serve, among many things, as a consistency check on this conjecture.

We arrive at the standard Verlinde formula by interpreting $\OSPMinMod{2}{4}$ as a simple current
extension of its bosonic orbifold --- the \voa{} formed by the even elements of $\OSPMinMod{2}{4}$ --- and then lifting the standard Verlinde formula of the orbifold back to \(\OSPMinMod{2}{4}\). In the case at hand, the bosonic orbifold is actually isomorphic to the
\(\AKMA{sl}{2}\) minimal model \voa{} $\SLMinMod{3}{4}$ of level \(-\frac{5}{4}\), see \cite[Sec.~10]{CreMod13}, although this identification is not needed for what follows.  The aim of the standard Verlinde formula is to determine the Grothendieck fusion rules of $\OSPMinMod{2}{4}$ directly without recourse to those of the orbifold.  Of course, it is easy, in this particular example, to use the orbifold's Grothendieck fusion rules (which were computed in \cite{CreMod13}), but the point is to develop the methodology in a generality that applies for all admissible levels.

In order to be concise, we shall in this subsection denote the affine minimal model $\OSPMinMod{2}{4}$ by \(\gvoa\).  Its bosonic orbifold \voa{} (isomorphic to $\SLMinMod{3}{4}$) will be denoted by \(\orb\). Let \(\gfuse\) denote the fusion product of \(\gvoa\)-modules and let \(\bfuse\) denote the fusion product of \(\orb\)-modules.  Denote by \(\Res{}\) the restriction of \(\gvoa\)-modules to \(\orb\)-modules and by \(\Ind{}\) the induction of \(\orb\)-modules to \(\gvoa\)-modules.  We note the following relevant properties and assumptions (which are understood to be in force):
\begin{enumerate}
	\item The \vosa{} \(\gvoa\) is a order $2$ simple current extension of \(\orb\),
  that is, there exists a simple \(\orb\)-module \(\simp\) such that
  \(\simp\bfuse\simp\cong\orb\) and \(\Res{\gvoa}\cong\orb\oplus\simp\).  This property is expected for all bosonic orbifolds and has been proven rigorously under certain hypotheses in \cite{MiyFla11,CarReg16}.  While these works do not address \vosas{} \emph{per se}, their methods are easily adapted to this situation (see \cite[App.~A]{CreSch16}).  Unfortunately, the required hypotheses are not known to be satisfied in our case, but this property was verified in \cite[see Prop.~15 and Sec.~10]{CreMod13} under the assumption that the standard Verlinde formula for $\orb \cong \SLMinMod{3}{4}$ computes Grothendieck fusion coefficients for $\orb$-modules. \label{it:DecompV}
	\item \(\gvoa\) admits a set of \emph{even parity} standard modules \(\gnsr{}=\gnsr{0}\cup\gnsr{1/2}\), where \(\gnsr{0}\) and \(\gnsr{1/2}\) denote the sets of even parity standard modules in the \ns{} and Ramond sectors, respectively.  This means that the characters and supercharacters of the modules in $\gnsr{}$ span a representation of the modular group and every simple module admits a resolution in terms of these modules and their parity reversals.  For $\xi = 0, \frac{1}{2}$, let \(\bnsr{\xi}\) be the set of $\orb$-modules that are obtained as direct summands of restrictions of $\gvoa$-modules in \(\gnsr{\xi}\). Then, \(\bnsr{}=\bnsr{0}\cup\bnsr{1/2}\) forms a set of standard $\orb$-modules. \label{it:NSRParameter}
\item The fusion products of \(\gvoa\) and \(\orb\) are compatible, that is,
  for any \(\orb\)-modules \(\Mod{W}\) and \(\Mod{X}\),
  \begin{equation} \label{eq:FusIndCompat}
    \Ind{\Mod{W}}\gfuse\Ind{\Mod{X}}\cong\Ind{\brac*{\Mod{W}\bfuse\Mod{X}}}.
  \end{equation}
  This was proven in \cite{RidVer14} under the assumption that $\gfuse$ and $\bfuse$ associate.  This assumption may be removed \cite{CreTen17} if the $\orb$-modules form a vertex tensor category in the sense of Huang, Lepowsky and Zhang \cite{HuaLog10}.  However, this is not known for the categories we consider; nevertheless, we assume that \eqref{eq:FusIndCompat} holds.
	\item Fusing with a $\gvoa$-module from the given category $\categ{P}$ of $\gvoa$-modules defines an exact endofunctor on $\categ{P}$. Similarly, fusing with an $\orb$-module defines an exact endofunctor on the category obtained by applying the restriction functor to $\categ{P}$.  In other words, the corresponding Grothendieck fusion products are well defined.
	\item The standard module formalism of \cite{CreLog13,RidVer14} applies to the orbifold \voa{} \(\orb\), so that
	\begin{subequations}
	  \begin{equation}
	    \bch{\Mod{W}}\bGrfuse\bch{\Mod{X}}=\bch{\Mod{W}\bfuse\Mod{X}}=
	    \int_{\bnsr{}}\bfuscoeff{\Mod{W}}{\Mod{X}}{\Mod{Y}}{}\bch{\Mod{Y}}\,\dd \Mod{Y}
	  \end{equation}
	  for any (appropriate) \(\orb\)-modules \(\Mod{W}\) and \(\Mod{X}\), where $\bGrfuse$ denotes the Grothendieck fusion product of $\orb$,
	  \begin{equation} \label{eq:orbstverlinde}
	    \bfuscoeff{\Mod{W}}{\Mod{X}}{\Mod{Y}}{}=\int_{\bnsr{}}\frac{\slSmat{\Mod{W}}{\Mod{Z}}\slSmat{\Mod{X}}{\Mod{Z}}\slSmat{\Mod{Y}}{\Mod{Z}}^\ast}{\slSmat{\orb}{\Mod{Z}}}\dd \Mod{Z}
	  \end{equation}
	  is a Grothendieck fusion coefficient for $\orb$, and \(\slSmat{\cdot}{\cdot}\) denotes the \(\modS\)-matrix entries of the \(\orb\)-module characters (in the basis of standard characters).  This is the most important assumption that we are making.  It was also assumed in \cite{CreMod13} when deriving the Grothendieck fusion rules for all admissible level $\AKMA{sl}{2}$ minimal models.  \cref{eq:orbstverlinde} is the standard Verlinde formula for $\orb$.
	\end{subequations} \label{it:orbstverlinde}
\end{enumerate}

These assumptions imply the following easy consequences.
\begin{cor}\label{thm:basicproperties}\leavevmode
  \begin{enumerate}
  \item The restriction of any \(\gvoa\)-module \(\Mod{M}\) decomposes as
    \begin{equation}
      \Res{\Mod{M}}\cong\Mod{M}^{\even}\oplus\Mod{M}^{\odd}\cong\Mod{M}^{\even}\oplus(\simp\bfuse\Mod{M}^{\even})
      \cong\brac*{\orb\oplus\simp}\bfuse\Mod{M}^{\even}\cong\Res{\gvoa}\bfuse\Mod{M}^{\even},
    \end{equation}
    where $\Mod{M}^{\even}$ and $\Mod{M}^{\odd} \cong \simp \bfuse \Mod{M}^{\even}$ denote the $\orb$-submodules consisting of the even and odd elements of $\Mod{M}$, respectively.
	\item The simple current $\simp$ has no fixed points: if $\Mod{W}$ is a simple $\orb$-module, then $\simp \bfuse \Mod{W} \ncong \Mod{W}$. \label{it:NoFixedPts}
  \item For any $\gvoa$-module $\Mod{M}$, the inductions of $\Mod{M}^{\even}$ and $\Mod{M}^{\odd}$ are given by \label{it:EvenOddInd}
		\begin{equation} \label{eq:EvenOddInd}
			\Ind{\Mod{M}^{\even}} \cong \Mod{M}, \quad \Ind{\Mod{M}^{\odd}} \cong \parrmod{\Mod{M}}.
		\end{equation}
  \item For any \(\gvoa\)-module \(\Mod{M}\), we have
    \begin{equation}
      \nch{\Mod{M}}{\eps}=\bch{\Mod{M}^{\even}}+(-1)^{2\eps}\bch{\Mod{M}^{\odd}}=\brac*{\bch{\orb}+(-1)^{2\eps}\bch{\simp}}\bGrfuse\bch{\Mod{M}^{\even}}.
    \end{equation}
  \end{enumerate}
\end{cor}

We now state the standard Verlinde formula for \(\gvoa\).  For this, it is convenient to treat the parameter $\xi \in \set{0,\frac{1}{2}}$ that distinguishes the \ns{} and Ramond sectors (see assumption \ref{it:NSRParameter} above) as an element of $\ZZ_2$.
\begin{thm}\label{thm:Verlindeformula}
	\begin{subequations}
	  Let \(\Mod{M}\) and \(\Mod{N}\) be \(\gvoa\)-modules and let \(\xi\) and \(\xi'\) be \(0\) or \(\frac{1}{2}\), depending on whether \(\Mod{M}\) and \(\Mod{N}\) are in the \ns{} or Ramond sector, respectively.  Under the assumptions (i)--(v) above, we then have
	  \begin{equation} \label{eq:chfusion}
	    \nch{\Mod{M}}{\eps}\Grfuse{}\nch{\Mod{N}}{\eps}=\nch{\Mod{M}\gfuse\Mod{N}}{\eps}
	    =\int_{\gnsr{\xi+\xi'}}\fuscoeff{\Mod{M}}{\Mod{N}}{\Mod{P}}{\eps}\nch{\Mod{P}}{\eps}\,\dd \Mod{P},
	  \end{equation}
	  where the Verlinde coefficients of $\gvoa$, indexed by $\eps=0$ (characters) or $\eps=\frac{1}{2}$ (supercharacters), are given by the following standard Verlinde formula:
	  \begin{equation}\label{eq:fusecoeff}
	    \fuscoeff{\Mod{M}}{\Mod{N}}{\Mod{P}}{\eps}
	    =\int_{\gnsr{\eps+1/2}}
	    \frac{\Smat{\eps}{\Mod{M}}{\xi+1/2}{\Mod{Q}} \Smat{\eps}{\Mod{N}}{\xi'+1/2}{\Mod{Q}} \Smat{\eps}{\Mod{P}}{\xi+\xi'+1/2}{\Mod{Q}}^*}{\Smat{\eps}{\gvoa}{1/2}{\Mod{Q}}} \dd \Mod{Q}.
	  \end{equation}
\end{subequations}
\end{thm}
\noindent We refer to the $\tfuscoeff{\Mod{M}}{\Mod{N}}{\Mod{P}}{\eps}$ as Verlinde coefficients, rather than Grothendieck fusion coefficients, because they depend upon whether we compute the character ($\eps=0$) or supercharacter ($\eps=\frac{1}{2}$) of the fusion product.  Of course, knowing both allows one to deduce the true Grothendieck fusion coefficients (\cref{sub:Gr}) which will be denoted by $\tfuscoeff{\Mod{M}}{\Mod{N}}{\Mod{P}}{}$.

We shall prove \cref{thm:Verlindeformula} below after recording some additional identities as preparatory lemmas.  First however, we note that \eqref{eq:fusecoeff} may be simplified by noting that all the S-matrix elements \eqref{eqn:SMatrixEntries} and \eqref{eqn:AtypSMatrixEntries} may be factored as follows:
\begin{equation} \label{eq:DefSigma}
	\Smat{\eps}{\Mod{M}}{\eps'}{\Mod{N}} = \Selt{\Mod{M}}{\Mod{N}} \wun_{\xi \neq \eps'} \wun_{\xi' \neq \eps}.
\end{equation}
Here, $\xi$ is $0$ or $\frac{1}{2}$ according as to whether $\Mod{M}$ is in the \ns{} or Ramond sectors, respectively, and $\xi'$ similarly indicates the sector of $\Mod{N}$.  The simplified standard Verlinde formula is therefore
\begin{equation} \label{eq:Grfuscoeff}
	\fuscoeff{\Mod{M}}{\Mod{N}}{\Mod{P}}{\eps} = \wun_{\Mod{P} \in \gnsr{\xi+\xi'}} \int_{\gnsr{\eps+1/2}} \frac{\Selt{\Mod{M}}{\Mod{Q}} \Selt{\Mod{N}}{\Mod{Q}} \Selt{\Mod{P}}{\Mod{Q}}^*}{\Selt{\gvoa}{\Mod{Q}}} \dd \Mod{Q},
\end{equation}
where the characteristic function takes care of the fact that this Grothendieck fusion coefficient is $0$ if $\Mod{P} \notin \gnsr{\xi + \xi'}$.

\begin{lemma} \label{lem:IndS}
	Let $\Mod{W}$ and $\Mod{X}$ be standard $\orb$-modules.  Then, we have
	\begin{subequations}
		\begin{equation} \label{eq:IndS}
			\slSmat{\Mod{W}}{\Mod{X}} = \frac{1}{2} \sum_{\eps,\eps'} \Smat{\eps}{\Ind{\Mod{W}}}{\eps'}{\Ind{\Mod{X}}}, \quad
			\slSmat{\simp \bfuse \Mod{W}}{\Mod{X}} = \frac{1}{2} \sum_{\eps,\eps'} (-1)^{2\eps} \Smat{\eps}{\Ind{\Mod{W}}}{\eps'}{\Ind{\Mod{X}}},
		\end{equation}
		where $\eps$ and $\eps'$ are summed over $0$ and $\frac{1}{2}$.  If $\Mod{W} \in \bnsr{\xi}$ and $\Mod{X} \in \bnsr{\xi'}$, then we may write the first of these relations in the equivalent form
	  \begin{equation} \label{eq:IndS'}
	    \Smat{\eps}{\Ind{\Mod{W}}}{\eps'}{\Ind{\Mod{X}}} = 2 \slSmat{\Mod{W}}{\Mod{X}} \wun_{\xi\neq\eps'} \wun_{\xi'\neq\eps}.
	  \end{equation}
	\end{subequations}
\end{lemma}
\begin{proof}
	To establish the first relation of \eqref{eq:IndS}, we note that
	\begin{align}
		\modch{S}{\bch{\Mod{W}}} &
		= \frac{1}{2} \sum_{\eps} \modch{S}{\nch{\Ind{\Mod{W}}}{\eps}}
		= \frac{1}{2} \sum_{\eps, \eps'} \int_{\gnsr{}} \Smat{\eps}{\Ind{\Mod{W}}}{\eps'}{\Mod{N}} \nch{\Mod{N}}{\eps'} \, \dd \Mod{N} \notag \\
		&= \frac{1}{2} \sum_{\eps, \eps'} \int_{\gnsr{}} \Smat{\eps}{\Ind{\Mod{W}}}{\eps'}{\Mod{N}} \brac*{\bch{\Mod{N^{\even}}} + (-1)^{2\eps'} \bch{\simp \bfuse \Mod{N}^{\even}}} \, \dd \Mod{N} \notag \\
		&= \frac{1}{4} \sum_{\eps, \eps'} \int_{\bnsr{}} \Smat{\eps}{\Ind{\Mod{W}}}{\eps'}{\Ind{\Mod{X}}} \brac*{\bch{\Mod{X}} + (-1)^{2\eps'} \bch{\simp \bfuse \Mod{X}}} \, \dd \Mod{X} \notag
		\intertext{($\dd \Mod{N} = \frac{1}{2} \dd \Mod{X}$ because $\Mod{X}$ and $\simp \bfuse \Mod{X}$ are distinct, by \cref{thm:basicproperties}\ref{it:NoFixedPts}, yet only one has a lift in $\gnsr{}$ because of the even parity requirement, see assumption \ref{it:NSRParameter} above)}
		&= \frac{1}{4} \sum_{\eps, \eps'} \int_{\bnsr{}} \brac*{\Smat{\eps}{\Ind{\Mod{W}}}{\eps'}{\Ind{\Mod{X}}} + (-1)^{2\eps'} \Smat{\eps}{\Ind{\Mod{W}}}{\eps'}{\Ind{(\simp \bfuse \Mod{X})}}} \bch{\Mod{X}} \, \dd \Mod{X} \notag
		\intertext{(since the square of $\simp \bfuse {-}$ is equivalent to the identity functor)}
		&= \frac{1}{2} \sum_{\eps, \eps'} \int_{\bnsr{}} \Smat{\eps}{\Ind{\Mod{W}}}{\eps'}{\Ind{\Mod{X}}} \bch{\Mod{X}} \, \dd \Mod{X},
	\end{align}
	(since $\Ind{(\simp \bfuse \Mod{X})} \cong \parrmod{(\Ind{\Mod{X}})}$, by \cref{thm:basicproperties}\ref{it:EvenOddInd}).  	The second relation follows using the same technique.  To deduce \eqref{eq:IndS'}, note from \eqref{eqn:SMatrixEntries} that the S-matrix elements of $\gvoa$ are all proportional to $\wun_{\xi \neq \eps'} \wun_{\xi' \neq \eps}$ and so the sums in \eqref{eq:IndS} each have a single contributing summand.
\end{proof}

\begin{lem}\label{thm:scurtransf}
  Let \(\Mod{W}\in\bnsr{\xi}\) and \(\Mod{X}\in\bnsr{\xi'}\) be standard \(\orb\)-modules, for some \(\xi,\xi'\in\set{0,\frac{1}{2}}\). Then, we have
  \begin{align} \label{eqn:orbSrelation}
    \slSmat{\simp \bfuse \Mod{W}}{\Mod{X}}=-(-1)^{2\xi'}\slSmat{\Mod{W}}{\Mod{X}}, \quad
    \slSmat{\Mod{W}}{\simp \bfuse \Mod{X}}=-(-1)^{2\xi}\slSmat{\Mod{W}}{\Mod{X}}.
  \end{align}
\end{lem}
\begin{proof}
  These identities follow, almost immediately, from the general theory of simple currents due to Schellekens and Yankielowicz \cite[Sec.~4]{SchSim90}.  Indeed, they show\footnote{Strictly speaking, the material in \cite{SchSim90} assumes the setting of a rational \voa{}.  However, it is possible to generalise many of their results to the standard module formalism, see \cite{MelSim17} for details.} that
  \begin{equation} \label{eq:propconst}
		\slSmat{\simp \bfuse \Mod{W}}{\Mod{X}} = \frac{\slSmat{\simp}{\Mod{X}}}{\slSmat{\orb}{\Mod{X}}} \slSmat{\Mod{W}}{\Mod{X}}
	\end{equation}
	and that the $\Mod{X}$-dependent proportionality factor is just a sign, as $\simp$ is an order two simple current.  They moreover show that, up to a global $\Mod{X}$-independent sign, this proportionality constant is determined by the difference in the conformal weights (mod $1$) of $\Mod{X}$ and $\simp \bfuse \Mod{X}$.  As this difference is $0$ if $\Ind{\Mod{X}}$ is \ns{} and $\frac{1}{2}$ if $\Ind{\Mod{X}}$ is Ramond, the proportionality factor is $\pm (-1)^{2 \xi'}$.

	To determine the global sign, we compute the \rhss{} of both identities in \eqref{eq:IndS} in a simple case and compare.  Taking $\Ind{\Mod{W}} = \NSRel{\lambda'}$ and $\Ind{\Mod{X}} = \RRel{\lambda'}$, so that both sides are non-zero, it is obvious that the proportionality factor in \eqref{eq:propconst} is $1$.  Since $\RRel{\lambda'}$ is Ramond ($\xi' = \frac{1}{2}$), the global sign is thus $-1$.  This completes the proof as the second identity follows from the first by the symmetry of the orbifold S-matrix.
\end{proof}

\begin{proof}[Proof of \cref{thm:Verlindeformula}]
  First, we combine \eqref{eq:FusIndCompat} with \eqref{eq:EvenOddInd} to give
  \begin{equation}
		\Mod{M} \gfuse \Mod{N} \cong \Ind{\Mod{M}^{\even}} \gfuse \Ind{\Mod{N}^{\even}} \cong \Ind{(\Mod{M}^{\even} \bfuse \Mod{N}^{\even})},
	\end{equation}
	which implies that
  \begin{equation}
	\Res{(\Mod{M} \gfuse \Mod{N})} \cong \Res{\gvoa} \bfuse (\Mod{M}^{\even} \bfuse \Mod{N}^{\even}).
  \end{equation}
	As $\Res{\gvoa} \cong \orb \oplus \simp$ (item \ref{it:DecompV} above), this implies the (super)character identity
  \begin{align}
    \nch{\Mod{M}\times\Mod{N}}{\eps}
    &=\brac*{\bch{\orb}+(-1)^{2\eps}\bch{\simp}}\bGrfuse
    \bch{\Mod{M}^{\even}\bfuse\Mod{N}^{\even}}\nonumber\\
    &= \int_{\bnsr{}} \bfuscoeff{\Mod{M}^{\even}}{\Mod{N}^{\even}}{\Mod{W}}
    \brac*{\bch{\orb}+(-1)^{2\eps}\bch{\simp}}\bGrfuse \bch{\Mod{W}}\,\dd \Mod{W}\nonumber\\
    &= \int_{\bnsr{}} \bfuscoeff{\Mod{M}^{\even}}{\Mod{N}^{\even}}{\Mod{W}} \nch{\Ind{\Mod{W}}}{\eps}\,\dd \Mod{W}\nonumber\\
    &= \int_{\gnsr{}} \brac*{\bfuscoeff{\Mod{M}^{\even}}{\Mod{N}^{\even}}{\Mod{P}^{\even}} \nch{\Mod{P}}{\eps}
    +\bfuscoeff{\Mod{M}^{\even}}{\Mod{N}^{\even}}{\simp\bfuse\Mod{P}^{\even}} \nch{\parrmod{\Mod{P}}}{\eps}} \,\dd \Mod{P}\nonumber\\
    &= \int_{\gnsr{}} \brac*{\bfuscoeff{\Mod{M}^{\even}}{\Mod{N}^{\even}}{\Mod{P}^{\even}}
    + (-1)^{2\eps} \bfuscoeff{\Mod{M}^{\even}}{\Mod{N}^{\even}}{\simp\bfuse\Mod{P}^{\even}}} \nch{\Mod{P}}{\eps} \,\dd \Mod{P},
  \end{align}
  where we set $\Mod{P} = \Ind{\Mod{W}}$, hence $\Mod{P}^{\even} \cong \Mod{W}$, and recall that $\gnsr{}$ consists of only the even parity standard $\gvoa$-modules (item \ref{it:NSRParameter} above).  Using \cref{thm:scurtransf}, the contribution of the standard \(\gvoa\)-module \(\Mod{P}\) to \(\nch{\Mod{M}\gfuse\Mod{N}}{\eps}\) is therefore
  \begin{align}
    \fuscoeff{\Mod{M}}{\Mod{N}}{\Mod{P}}{\eps}
    &=\bfuscoeff{\Mod{M}^{\even}}{\Mod{N}^{\even}}{\Mod{P}^{\even}} +(-1)^{2\eps} \bfuscoeff{\Mod{M}^{\even}}{\Mod{N}^{\even}}{\simp\bfuse\Mod{P}^{\even}} \nonumber\\
    &= \int_{\bnsr{}} \frac{\slSmat{\Mod{M}^{\even}}{\Mod{X}} \slSmat{\Mod{N}^{\even}}{\Mod{X}}}{\slSmat{\orb}{\Mod{X}}}
    \brac*{\slSmat{\Mod{P}^{\even}}{\Mod{X}} + (-1)^{2\eps} \slSmat{\simp\bfuse\Mod{P}^{\even}}{\Mod{X}}}^\ast \,\dd \Mod{X}\nonumber\\
    &= 2\int_{\bnsr{\eps+1/2}} \frac{\slSmat{\Mod{M}^{\even}}{\Mod{X}}\slSmat{\Mod{N}^{\even}}{\Mod{X}}\slSmat{\Mod{P}^{\even}}{\Mod{X}}^\ast}{\slSmat{\orb}{\Mod{X}}} \dd \Mod{X}\nonumber\\
    &= 2\int_{\gnsr{\eps+1/2}} \left( \frac{\slSmat{\Mod{M}^{\even}}{\Mod{Q}^{\even}}\slSmat{\Mod{N}^{\even}}{\Mod{Q}^{\even}}\slSmat{\Mod{P}^{\even}}{\Mod{Q}^{\even}}^\ast}{\slSmat{\orb}{\Mod{Q}^{\even}}} \right. \nonumber\\
    &\mspace{100mu} \left. + \frac{\slSmat{\Mod{M}^{\even}}{\simp\bfuse\Mod{Q}^{\even}}\slSmat{\Mod{N}^{\even}}{\simp\bfuse\Mod{Q}^{\even}}\slSmat{\Mod{P}^{\even}}{\simp\bfuse\Mod{Q}^{\even}}^\ast}{\slSmat{\orb}{\simp\bfuse\Mod{Q}^{\even}}} \right) \dd \Mod{Q}.
  \end{align}
  Again using \cref{thm:scurtransf}, the integrand can only be non-zero if \(\Mod{P}\in\gnsr{\xi+\xi'}\). Assuming this, we have
  \begin{align}
    \fuscoeff{\Mod{M}}{\Mod{N}}{\Mod{P}}{\eps}
    &= 4\int_{\gnsr{\eps+1/2}} \frac{\slSmat{\Mod{M}^{\even}}{\Mod{Q}^{\even}}\slSmat{\Mod{N}^{\even}}{\Mod{Q}^{\even}}\slSmat{\Mod{P}^{\even}}{\Mod{Q}^{\even}}^\ast}{\slSmat{\orb}{\Mod{Q}^{\even}}}\dd\Mod{Q} \nonumber\\
    &=\int_{\gnsr{\eps+1/2}} \frac{\Smat{\eps}{\Mod{M}}{\xi+1/2}{\Mod{Q}} \Smat{\eps}{\Mod{N}}{\xi'+1/2}{\Mod{Q}}\Smat{\eps}{\Mod{P}}{\xi+\xi'+1/2}{\Mod{Q}}^\ast}{\Smat{\eps}{\gvoa}{1/2}{\Mod{Q}}}\dd \Mod{Q}
  \end{align}
  by \eqref{eq:IndS'}, as required.
\end{proof}

\subsection{Evaluating the Verlinde formula}

In this section, we assume that \(\OSPMinMod{2}{4}\) (and its bosonic orbifold) satisfy the assumptions laid out at the beginning of \cref{subsec:verlinde}.  The aim is to use \cref{thm:Verlindeformula} to compute the Grothendieck fusion rules of the simple \(\OSPMinMod{2}{4}\)-modules in the category $\categ{P}$.  To simplify the calculations, we shall make use of the following observations, writing $A \tGr{\Mod{M}}$ for $\tGr{A \Mod{M}}$, $A = \sfaut, \conjaut, \parr$, when convenient.

\begin{prop}\label{thm:twistfuse}
  Let \(\Mod{M}\) and \(\Mod{N}\) be \(\OSPMinMod{2}{4}\)-modules in the category $\categ{P}$.  Then, under the assumptions (i)--(v) above, for all \(m,n\in\frac{1}{2}\ZZ\), we have
  \begin{subequations}
		\begin{align}
			\Gr{\parrmod{\Mod{M}}} \gGrfuse \Gr{\Mod{N}} &= \parrmod{\brac*{\Gr{\Mod{M}} \gGrfuse \Gr{\Mod{N}}}}, \label{eq:parrfuse} \\
			\Gr{\conjmod{\Mod{M}}} \gGrfuse \Gr{\conjmod{\Mod{N}}} &= \conjmod{\brac*{\Gr{\Mod{M}} \gGrfuse \Gr{\Mod{N}}}}, \label{eq:conjfuse} \\
			\Gr{\sfmod{m}{\Mod{M}}} \gGrfuse \Gr{\sfmod{n}{\Mod{N}}} &= \sfmod{m+n}{\brac*{\Gr{\Mod{M}} \gGrfuse \Gr{\Mod{N}}}}. \label{eq:specid}
		\end{align}
	\end{subequations}
\end{prop}
\begin{proof}
  Note that the standard Verlinde formula \eqref{eq:fusecoeff} for the Verlinde coefficients is bilinear in \(\Mod{M}\) and \(\Mod{N}\).  We may therefore restrict to $\Mod{M}$ and $\Mod{N}$ simple (or standard).

  \cref{eq:parrfuse} is clear from \eqref{eqn:chpar} and \eqref{eqn:schpar}:  $\tnch{\parrmod{\Mod{M}}}{\eps} = (-1)^{2\eps} \tnch{\Mod{M}}{\eps}$.  For \eqref{eq:conjfuse}, note that combining \eqref{eq:ConjSimples} with \eqref{eqn:SMatrixEntries} and \eqref{eqn:AtypSMatrixEntries} yields
	\begin{equation}
		\Smat{\eps}{\conjmod{\Mod{M}}}{\eps'}{\conjmod{\Mod{Q}}} = \Smat{\eps}{\Mod{M}}{\eps'}{\Mod{Q}}.
	\end{equation}
  Since $\Mod{Q} \mapsto \conjmod{\Mod{Q}}$ is a bijection of $\gnsr{\eps+1/2}$ and $\gvoa = \OSPMinMod{2}{4}$ is self-conjugate, it follows from the standard Verlinde formula \eqref{eq:fusecoeff} that
	\begin{equation}
		\fuscoeff{\conjmod{\Mod{M}}}{\conjmod{\Mod{N}}}{\conjmod{\Mod{P}}}{\eps} = \fuscoeff{\Mod{M}}{\Mod{N}}{\Mod{P}}{\eps},
	\end{equation}
  as required.

  The key to \eqref{eq:specid} is to note from \eqref{eqn:SMatrixEntries} and \eqref{eqn:AtypSMatrixEntries} that
	\begin{equation}
		\Selt{\sfmod{\ell}{\Mod{M}}}{\Mod{Q}} = \ee^{-2 \pi \ii \alpha(\Mod{Q}) \ell} \Selt{\Mod{M}}{\Mod{Q}},
	\end{equation}
  where $\alpha(\Mod{Q}) \in \RR$ does not depend on $\Mod{M}$ or $\ell$.  As $\Mod{Q}$ is the same in all $\Sigma$-factors appearing in the integrand of the simplified standard Verlinde formula \eqref{eq:Grfuscoeff}, the exponential factors involving $\alpha(\Mod{Q})$ cancel and we have
	\begin{align}
		\fuscoeff{\sfmod{\ell}{\Mod{M}}}{\sfmod{\ell'}{\Mod{N}}}{\sfmod{\ell+\ell'}{\Mod{P}}}{\eps}
		&= \wun_{\sfmod{\ell+\ell'}{\Mod{P}} \in \gnsr{\xi+\xi'+\ell+\ell'}} \int_{\gnsr{\eps+1/2}} \frac{\Selt{\sfmod{\ell}{\Mod{M}}}{\Mod{Q}} \Selt{\sfmod{\ell'}{\Mod{N}}}{\Mod{Q}} \Selt{\sfmod{\ell+\ell'}{\Mod{P}}}{\Mod{Q}}^*}{\Selt{\gvoa}{\Mod{Q}}} \dd \Mod{Q} \notag \\
		&= \wun_{\Mod{P} \in \gnsr{\xi+\xi'}} \int_{\gnsr{\eps+1/2}} \frac{\Selt{\Mod{M}}{\Mod{Q}} \Selt{\Mod{N}}{\Mod{Q}} \Selt{\Mod{P}}{\Mod{Q}}^*}{\Selt{\gvoa}{\Mod{Q}}} \dd \Mod{Q}
		= \fuscoeff{\Mod{M}}{\Mod{N}}{\Mod{P}}{\eps}.
	\end{align}
  Here, $\xi$ and $\xi'$ indicate whether $\Mod{M}$ and $\Mod{N}$ belong to the \ns{} or Ramond sector, as usual.
\end{proof}

\begin{thm}\label{thm:fuseformulae}
  Under the assumptions (i)--(v) above, the image $\tGr{\NSFin{0}}$ of the vacuum module is the unit of the Grothendieck fusion ring of $\OSPMinMod{2}{4}$.  Moreover, the Grothendieck fusion rules of the remaining simple $\OSPMinMod{2}{4}$-modules are then specified by
  \begin{subequations}
	  \begin{gather}
			\begin{gathered}
		    \Gr{\NSInf{-1/2}^+} \gGrfuse \Gr{\NSInf{-1/2}^+}
					= \Gr{\parrmod{\sfmod{}{\NSInf{+1/2}^-}}} + \Gr{\sfmod{1/2}{\RRel{1/4}}},\\
		    \Gr{\NSInf{-1/2}^+} \gGrfuse \Gr{\NSInf{+1/2}^-} = \Gr{\NSFin{0}} + \Gr{\NSRel{0}},
		  \end{gathered}
		  \label{GrFR:AtypAtyp} \\
		  \begin{gathered}
		    \Gr{\NSInf{-1/2}^+} \gGrfuse \Gr{\NSRel{\lambda}}
					= \Gr{\sfmod{1/2}{\RRel{\lambda-k-1/2}}} + \Gr{\parrmod \sfmod{1/2}{\RRel{\lambda-k+1/2}}}
					+ \Gr{\NSRel{\lambda-1/2}},\\
		    \Gr{\NSInf{-1/2}^+} \gGrfuse \Gr{\RRel{\lambda}}
					= \Gr{\sfmod{1/2}{\NSRel{\lambda-k-1/2}}} + \Gr{\RRel{\lambda-1/2}},
		  \end{gathered}
		  \label{GrFR:AtypSt} \\
		  \begin{gathered}
		    \Gr{\NSRel{\lambda}} \gGrfuse \Gr{\NSRel{\mu}}
					= \Gr{\parrmod{\sfmod{1/2}{\RRel{\lambda+\mu-k-1}}}} + \Gr{\sfmod{1/2}{\RRel{\lambda+\mu-k}}}
					+ 2 \Gr{\NSRel{\lambda+\mu}}
					+ \Gr{\sfmod{-1/2}{\RRel{\lambda+\mu+k}}} + \Gr{\parrmod{\sfmod{-1/2}{\RRel{\lambda+\mu+k+1}}}}, \\
		    \Gr{\NSRel{\lambda}} \gGrfuse \Gr{\RRel{\mu}}
					= \Gr{\sfmod{1/2}{\NSRel{\lambda+\mu-k}}} + \Gr{\RRel{\lambda+\mu}}
					+ \Gr{\parrmod{\RRel{\lambda+\mu+1}}}	+ \Gr{\sfmod{-1/2}{\NSRel{\lambda+\mu+k}}}, \\
		    \Gr{\RRel{\lambda}} \gGrfuse \Gr{\RRel{\mu}}
					= \Gr{\sfmod{1/2}{\RRel{\lambda+\mu-k}}} + \Gr{\NSRel{\lambda+\mu}} + \Gr{\sfmod{-1/2}{\RRel{\lambda+\mu+k}}},
		  \end{gathered}
		  \label{GrFR:StSt}
		\end{gather}
	\end{subequations}
  together with \cref{thm:twistfuse}.
\end{thm}

\begin{proof}
  These Grothendieck fusion rules are grouped into three classes: atypical by atypical, atypical by typical and typical by typical.  The rules are also correct whenever a typical module should be replaced by an atypical standard module (on the left-hand or \rhss{}).  They all follow by evaluating the standard Verlinde formula of \cref{thm:Verlindeformula}.  We discuss the evaluation for two examples, one from each of the last two classes, and explain how to then deduce the rules for the remaining class.  The omitted calculations are similar to those presented.

  First, we compute \(\tGr{\NSRel{\lambda}} \gGrfuse \tGr{\NSRel{\mu}}\). Using the simplified standard Verlinde formula \eqref{eq:Grfuscoeff}, the Verlinde coefficient of \(\tGr{\sfmod{\ell}{\NSRel{\nu}}}\) in this fusion product is
  \begin{align}
    \fuscoeff{\NSRel{\lambda}}{\NSRel{\mu}}{\sfmod{\ell}{\NSRel{\nu}}}{\eps}
    &= \wun_{\sfmod{\ell}{\NSRel{\nu}} \in \gnsr{0}} \int_{\gnsr{\eps+1/2}}
			\frac{\Selt{\NSRel{\lambda}}{\Mod{Q}} \Selt{\NSRel{\mu}}{\Mod{Q}} \Selt{\sfmod{\ell}{\NSRel{\nu}}}{\Mod{Q}}^*}{\Selt{\NSFin{0}}{\Mod{Q}}} \dd \Mod{Q} \notag \\
		&= \wun_{\ell \in \ZZ} \sum_{m \in \ZZ+\eps} \int_{\RR/2\ZZ}
			\frac{\Selt{\NSRel{\lambda}}{\sfmod{m}{\RRel{\rho}}} \Selt{\NSRel{\mu}}{\sfmod{m}{\RRel{\rho}}} \Selt{\sfmod{\ell}{\NSRel{\nu}}}{\sfmod{m}{\RRel{\rho}}}^*}{\Selt{\NSFin{0}}{\sfmod{m}{\RRel{\rho}}}} \dd \rho \nonumber
		\intertext{(since $\Mod{Q} = \sfmod{m}{\NSRel{\rho}}$ gives no contribution)}
    &= \wun_{\ell \in \ZZ} \sum_{m \in \ZZ+\eps} \ee^{-2 \pi \ii (\lambda+\mu-\nu-2k\ell) m}
			\int_{\RR/2\ZZ} \ee^{2 \pi \ii \ell \rho} \brac*{\frac{1}{\sqrt{2}} \ee^{\pi \ii \rho} + \frac{1}{\sqrt{2}} \ee^{-\pi \ii \rho} + 1} \, \dd \rho \nonumber\\
    &= 2\delta_{\ell,0} \ee^{-2 \pi \ii (\lambda+\mu-\nu) \eps} \delta(\nu=\lambda+\mu \bmod{1})
  \end{align}
  (as the terms involving $\ee^{\pm \pi \ii \rho}$ lead to $\ell = \mp \frac{1}{2} \notin \ZZ$).  The contribution to the (super)character of \(\NSRel{\lambda} \gfuse \NSRel{\mu}\), as specified in \eqref{eq:chfusion}, is thus
  \begin{equation}
		\sum_{\ell \in \frac{1}{2} \ZZ} \int_{\RR/\ZZ} \fuscoeff{\NSRel{\lambda}}{\NSRel{\mu}}{\sfmod{\ell}{\NSRel{\nu}}}{\eps} \nch{\sfmod{\ell}{\NSRel{\nu}}}{\eps} \, \dd \nu = 2 \nch{\NSRel{\lambda+\mu}}{\eps}.
	\end{equation}
	The contribution to the Grothendieck fusion rule \(\tGr{\NSRel{\lambda}} \gGrfuse \tGr{\NSRel{\mu}}\) is therefore $2 \tGr{\NSRel{\lambda+\mu}}$.

	A very similar calculation gives
	\begin{align}
		\fuscoeff{\NSRel{\lambda}}{\NSRel{\mu}}{\sfmod{\ell}{\RRel{\nu}}}{\eps}
		&= \delta_{\ell,1/2} \ee^{-2 \pi \ii (\lambda+\mu-\nu-k) \eps} \delta(\nu=\lambda+\mu-k \bmod{1}) \notag \\
		&\mspace{50mu} + \delta_{\ell,-1/2} \ee^{-2 \pi \ii (\lambda+\mu-\nu+k) \eps} \delta(\nu=\lambda+\mu+k \bmod{1}).
	\end{align}
	To compute the contribution to the (super)character, we must account for the fact that $\sfmod{\ell}{\RRel{\nu}}$ is $2$-periodic in $\nu$:
	\begin{align}
		\sum_{\ell \in \frac{1}{2} \ZZ} \int_{\RR/2\ZZ} \fuscoeff{\NSRel{\lambda}}{\NSRel{\mu}}{\sfmod{\ell}{\RRel{\nu}}}{\eps} \nch{\sfmod{\ell}{\RRel{\nu}}}{\eps} \, \dd \nu
		&= \nch{\sfmod{1/2}{\RRel{\lambda+\mu-k}}}{\eps} + (-1)^{2\eps} \nch{\sfmod{1/2}{\RRel{\lambda+\mu-k-1}}}{\eps} \notag \\
		&\mspace{50mu} + \nch{\sfmod{-1/2}{\RRel{\lambda+\mu+k}}}{\eps} + (-1)^{2\eps} \nch{\sfmod{-1/2}{\RRel{\lambda+\mu+k+1}}}{\eps}.
	\end{align}
	The signs $(-1)^{2\eps}$ therefore indicate parity reversal in the contribution to the Grothendieck fusion rule.  Putting this together with the contribution from the previous calculation completes the determination of \(\tGr{\NSRel{\lambda}} \gGrfuse \tGr{\NSRel{\mu}}\).

  For our second example, we consider \(\tGr{\NSInf{-1/2}^+} \gGrfuse \tGr{\RRel{\lambda}}\).  The contributions from the $\tGr{\sfmod{\ell}{\NSRel{\mu}}}$ are calculated as in the previous example, hence are omitted.  We instead consider the contributions from the $\tGr{\sfmod{\ell}{\RRel{\mu}}}$.  This requires us to take both $\Mod{Q} = \sfmod{m}{\NSRel{\rho}}$ and $\Mod{Q} = \sfmod{m}{\RRel{\rho}}$ in the simplified standard Verlinde formula, giving (after some familiar steps) the result that
  \begin{equation}
		\fuscoeff{\NSInf{-1/2}^+}{\RRel{\lambda}}{\sfmod{\ell}{\RRel{\mu}}}{\eps} = \frac{1}{2} \brac*{1 + \ee^{-\pi \ii (\lambda-\mu-1/2)}} \delta_{\ell,0} \ee^{-2 \pi \ii (\lambda-\mu-1/2) \eps} \delta(\mu=\lambda-\tfrac{1}{2} \bmod{1}).
	\end{equation}
	Integrating over $\RR/2\ZZ$ to get the contribution to the (super)character, the first factor is $1$ when $\mu=\lambda-\frac{1}{2} \bmod{2}$ but is $0$ when $\mu=\lambda+\frac{1}{2} \bmod{2}$.  In this way, we avoid introducing a parity-reversed module in the contribution $\tGr{\RRel{\lambda-1/2}}$ to the Grothendieck fusion rule.

  Finally, we note that the most straightforward way to compute the atypical by atypical Grothendieck fusion rules \eqref{GrFR:AtypAtyp} is to combine the identities of \cref{cor:resgr} with the atypical by standard rules \eqref{GrFR:AtypSt} that have already been computed.  (One can also use the same method to compute the atypical by standard rules in terms of the standard by standard ones \eqref{GrFR:StSt}, but the former are easy to calculate directly as we have seen.)  This completes the proof.
\end{proof}

\begin{remark}
	We note that the reason why one can use the identities of \cref{cor:resgr} in this proof is that the Grothendieck fusion product extends to the given completion \eqref{eq:thecompletion} of the Grothendieck group.  This follows easily from \cref{eq:specid,GrFR:StSt} upon recalling that the completion only admits infinite-linear combinations of images of standards in which the spectral flow indices are bounded below.
\end{remark}

\begin{remark}
	Of course, one can easily confirm that the $\OSPMinMod{2}{4}$ Grothendieck fusion rules computed here are reproduced by combining \eqref{eq:FusIndCompat} with the $\SLMinMod{3}{4}$ Grothendieck fusion rules reported in \cite[Tab.~3]{CreMod13}.  When doing so, note that one ``unit'' of spectral flow for $\OSPMinMod{2}{4}$ amounts to two units of spectral flow for $\SLMinMod{3}{4}$ (with the conventions used here and in \cite{CreMod13}).
\end{remark}

\flushleft

\providecommand{\opp}[2]{\textsf{arXiv:\mbox{#2}/#1}}\providecommand{\pp}[2]{\textsf{arXiv:#1
  [\mbox{#2}]}}

\end{document}